\documentclass[12pt]{article}   %
\usepackage[margin=1in]{geometry}

\usepackage[prependcaption,colorinlistoftodos,textsize=small]{todonotes}

\usepackage{anyfontsize}
\usepackage{adjustbox}
\usepackage{bm}

\usepackage{fancyvrb}
\usepackage{microtype,setspace}
\usepackage{amsmath,amsthm,amssymb,amsfonts,bbm,framed,mathtools}
\usepackage{graphicx}
\usepackage{subcaption}
\usepackage[hidelinks]{hyperref}
\usepackage[capitalise]{cleveref}
\crefname{tcb@cnt@mybox}{box}{boxes}
\usepackage{tikz}
\usepackage{tabularx}
\usepackage[group-separator={,}, scientific-notation=fixed]{siunitx}
\usepackage{comment}
\usepackage{makecell} 
\usepackage[percent]{overpic}
\usepackage[inline, shortlabels]{enumitem}
\setlist[enumerate]{
  itemjoin={{, }},
  itemjoin*={{, and }},
}

\usepackage{natbib}
\usepackage{glossaries}
\glsdisablehyper
\setacronymstyle{long-short}

\newacronym{anova}{ANOVA}{analysis of variance}
\newacronym{nhanes}{NHANES}{National Health and Nutrition Examination Survey}
\newacronym{bmd}{BMD}{bone mineral density}
\newacronym{cesd}{CES-D}{Center for Epidemiologic Studies-Depression Scale}
\newacronym{mmse}{MMSE}{Mini-Mental State Examination}
\newacronym{mdrs}{MDRS}{Mattis Dementia Rating Scale}
\newacronym{mle}{MLE}{Maximum Likelihood Estimator}

\newtheorem{theorem}{Theorem}[section]
\newtheorem{remark}{Remark}
\newtheorem{definition}{Definition}

\newcommand\iidsim{\stackrel{\mathclap{iid}}{\sim}}
\newcommand{\1}{\mathbbm{1}}
\newcommand{\R}{\mathbb{R}}
\newcommand{\hov}{\ensuremath{H_0^{1:p}}}
\newcommand{\hoM}{\ensuremath{H_0^M}}
\newcommand{\honM}{\ensuremath{H_0^{-M}}}
\newcommand{\FhoM}{\ensuremath{F_{\hoM}}}

\newcommand{\phoM}{\ensuremath{p_{\hoM}}}
\newcommand{\ptest}{\ensuremath{p_{\hoM}}^{\text{test}}}

\newcommand{\E}{{\mathbb{E}}}
\newcommand{\hoj}{\ensuremath{H_0^j}}
\newcommand{\hob}{\ensuremath{H_0^j(b)}}
\newcommand{\hoone}{\ensuremath{H_0^1}}
\newcommand{\alphaov}{\ensuremath{\alpha_0}}
\newcommand{\Fov}{{F_{\hov}}}

\newcommand{\Fovtrain}{\ensuremath{F_{\hov}^{\mathrm{train}}}}
\newcommand{\Xtrain}{X^{\mathrm{tr}}}
\newcommand{\Xtest}{X^{\mathrm{te}}}
\newcommand{\Ytrain}{Y^{\mathrm{tr}}}
\newcommand{\Ytest}{Y^{\mathrm{te}}}
\newcommand{\Rrho}{R_\rho}

\newcommand{\psel}{p_{H_0^M \mid E}}
\newcommand{\pselb}{p_{\hob \mid \tilde E^b}}

\newcommand{\CI}{\mathrm{CI}_\alpha^j(y)}

\newcommand{\RR}{\mathrm{R}}
\newcommand{\RSE}{\mathrm{RSE}}

\newcommand{\dretro}{d_{\mathrm{retro}}}
\newcommand{\aretro}{a_{\mathrm{retro}}}
\newcommand{\BvA}{B\hspace{.5mm}\text{vs.}\hspace{.5mm}A}

\newcommand{\kvk}{k_1\hspace{.25mm}\text{vs.}\hspace{.25mm}k_2}
\newcommand{\tZ}{\tilde{Z}}
\newcommand{\tW}{\tilde{W}}
\newcommand{\tV}{\tilde{V}}

\newtheorem{proposition}{Proposition}[section]

\newtheorem{lemma}{Lemma}[section]

\usepackage{tcolorbox}
\newtcolorbox[auto counter]{mybox}[2][]{float,title={\textcolor{black}{Box~\thetcbcounter: #2}},#1, 
colframe=white!70!gray}

\usepackage{authblk}

\title{Valid F-screening in linear regression}
\author[1]{Olivia McGough}
\author[1,2]{Daniela Witten}
\author[3,4]{Daniel Kessler}
\affil[1]{Department of Statistics, University of Washington}
\affil[2]{Department of Biostatistics, University of Washington}
\affil[3]{Department of Statistics and Operations Research, University of North Carolina at Chapel Hill}
\affil[4]{School of Data Science and Society, University of North Carolina at Chapel Hill}

\begin{document}

\maketitle

\begin{abstract}
  Suppose that a data analyst wishes to report the results of a least squares
linear regression \emph{only if the overall null hypothesis---$\hov: \beta_1=
  \beta_{2} = \ldots = \beta_p=0$---is rejected}. This practice, which we refer
to as \emph{F-screening} (since the overall null hypothesis is typically tested
using an $F$-statistic), is in fact common  across a number of applied
fields. Unfortunately, it poses a problem: standard guarantees for the
inferential outputs of linear regression, such as Type 1 error control of
hypothesis tests and nominal coverage of confidence intervals, hold
\emph{unconditionally}, but fail to hold conditional on rejection of the overall
null hypothesis. In this paper, we develop an inferential toolbox for the
coefficients in a least squares model that are valid conditional on rejection of
the overall null hypothesis. We develop selective p-values that lead to tests
that are consistent and control the selective Type 1 error, i.e., the Type 1 error conditional on
having rejected the overall null hypothesis. 
Furthermore, they can be computed
without access to the raw data, i.e., using only the standard outputs of a least
squares linear regression, and therefore are suitable for use in a retrospective
analysis of a published study. We also develop confidence intervals that attain the
nominal selective coverage, and point estimates that account for
having rejected the overall null hypothesis. We derive an expression for the Fisher information about the coefficients resulting from the proposed approach, and compare this to the Fisher information that results from an alternative approach that relies on sample splitting. We investigate the proposed approach in simulation and via re-analysis of two
datasets from the biomedical literature.

\end{abstract}

\section{Introduction}
\label{sec:intro}

Not all data analyses yield interesting results, and many uninteresting findings
are left in the ``file drawer'' \citep{rosenthal1979FileDrawerProblem} rather than published. Furthermore, when
papers are published, analyses yielding null results are often omitted \citep{benjamini2005FalseDiscoveryRate}. Such \emph{selective reporting}
has been empirically observed in the biomedical literature \citep{chan2004EmpiricalEvidenceSelective, ioannidis2005WhyMostPublished}.

Why is selective reporting a problem? Standard statistical guarantees, such as
Type 1 error control for hypothesis tests and nominal coverage of confidence
intervals, only hold when results are reported regardless of whether they are
``interesting.''

In this work, we consider a particular and pervasive form of selective reporting
that arises in the context of a regression model of the form
\begin{equation}\label{eq:model}
  Y = \beta_0 + X_1 \beta_1 + \ldots + X_p \beta_p  + \epsilon, \;\;\; \epsilon \sim \mathcal{N}(0, \sigma^2).
\end{equation}
Given $n$ independent observations from \eqref{eq:model}, suppose that the data
analyst first tests
\begin{equation} \label{eq:hov}
\hov: \beta_1 = \beta_2 = \ldots = \beta_p = 0
\end{equation}
at level \(\alphaov\in (0,1)\) in order to determine whether the model in \eqref{eq:model} merits further study. If they fail to reject $ \hov$, then \eqref{eq:model} receives no further
consideration, and the fact that \eqref{eq:model} was considered may not even be reported.
On the other hand, if (and only if) they reject \(\hov\), they then proceed to
test $\hoM : \beta_M = 0$ for a fixed (and typically strict) subset $M$ of $\left\{1,\ldots,p\right\}$, where $\beta_M := \{\beta_j:j\in M\}.$
Because \(\hov\) is typically tested by means of an \(F\)-statistic, we refer to
this procedure---summarized in Box~\ref{box:box1}---as \emph{F-screening}. If conducted, the test in Step 2 of Box~\ref{box:box1} is sometimes referred to as a \emph{post hoc} test. 

\begin{mybox}[floatplacement=!h,label={box:box1}]{\emph{F-screening in the model \eqref{eq:model}.}
  } \begin{enumerate}[start=1,label={Step \arabic*:},leftmargin=*]
  \item Test $\hov$ as given in \eqref{eq:hov} at level $\alphaov$, using an $F$-test.
  \item Only if $\hov$ is rejected, conduct inference on $\beta_M$ for fixed $M \subset \{ 1,\ldots, p \}$.
  \end{enumerate}
\end{mybox}

To see that conducting ``standard'' inference in Step 2 of F-screening will not provide appropriate
statistical guarantees, consider the extreme setting where \(p = 1\) and $\beta_1 = 0$. In this
case, \(\hov \equiv \hoone\). Suppose we carry out F-screening in Box~\ref{box:box1}
using the ``standard'' $t$-test for \(\hoone\) in Step 2. Then for fixed $\alphaov$, the tests in Steps 1 and
2 are identical, and so we will report a result in Step 2 only if $\hoone$ is
rejected. Thus the Type 1 error rate among reported
results equals one! The left plot in \cref{fig:triple_plot} illustrates the inflation of Type 1
error that results from conducting ``standard'' inference in Step 2 in a setting with
$p=10$ and $n=100$.

As pointed out by \citet{olshen1973ConditionalLevelTest}, conducting a test of $\hoM$ that accounts for the fact that $\hoM$ is tested if and only if 
$\hov$ is rejected is distinct from the task of correcting
multiple tests of the form $\hoM: \beta_M=0$ for
multiple comparisons. The latter has been very well-explored in the literature:
see \citet{dunn1961MultipleComparisonsMeans,sidak1967RectangularConfidenceRegions,tukey1949ComparingIndividualMeans,scheffe1953MethodJudgingAll}. For instance, if our interest in Step 2 of Box~\ref{box:box1} lies in $\hoj :
\beta_j = 0 $ for $j\in\{1,\ldots, p\}$, then we could consider testing each null hypothesis using a standard test, and applying a Bonferroni correction to account
for the fact that $p$ tests were conducted. However, this will not address the
fact that F-screening was performed: namely, that we test $\{\hoj\}_{j=1}^p$ \emph{if and only if we have rejected $\hov$}. 
In Appendix \ref{app:multipletesting}, we elaborate on how the guarantees provided by multiple testing corrections---such as family-wise error rate control---differ from those achieved by our proposed method. Further, we describe a closed-testing procedure that parallels F-screening and show that its guarantees are quite different. \cref{fig:qqplots-appendix} demonstrates empirically that muliple testing approaches fail to provide valid inference conditional on the rejection of $\hov$.

Our goal of conducting inference on $\beta_M$ that is valid \emph{conditional on the fact that we rejected $\hov$} falls under the umbrella of \emph{selective
inference} or \emph{post-selection inference}; we will view rejection of $\hov$ as the ``selection event.'' The selective inference framework was laid out in
\citet{fithian2017OptimalInferenceModel} and further explored in a variety of
settings \citep[among others]{taylor2015StatisticalLearningSelective,
  lee2016ExactPostselectionInference, zhang2019ValidPostclusteringDifferential,
  chen2020ValidInferenceCorrected, jewell2022TestingChangeMean,
  chen2024TestingDifferenceMeans, chen2023MorePowerfulSelective,
  chen2023SelectiveInferenceKmeans, gao2024SelectiveInferenceHierarchical,
  neufeld2022TreeValues}. We with to control the selective Type 1 error for the Step 2 test, conditional on rejection in Step 1; that is,
$$\Pr_{\hoM}\left(\text{$\hoM$ is rejected} \mid \text{$\hoM$ is tested}\right)\leq \alpha.$$
Our work is related to \citet{heller2019PostSelectionEstimationTesting},
who considered inference following aggregate testing for a coefficient that
follows a multivariate normal distribution with known covariance. However, in \cref{app:heller}, we show that their proposal does not control the selective Type 1 error in the setting of Box~\ref{box:box1}.

Our contributions are as follows:
\begin{enumerate}
\item We provide a valid end-to-end approach for conducting F-screening in the sense of Box~\ref{box:box1}. That is, we enable reporting of all of the standard multiple linear regression outputs---such as p-values, point estimates, and confidence intervals---conditional on rejection of $\hov$ in Step 1; we will refer to them as \textit{selective} p-values, point estimates, and confidence intervals. %
\item We establish consistency and  selective Type 1 error control of the test resulting from our selective p-value.
\item We characterize the Fisher information about $\beta_M$ resulting from our procedure, and we show that it compares favorably to that of a procedure that relies on sample splitting.
\item We show that the selective p-values can be computed using
only the standard multiple linear regression outputs \emph{without access to the raw data}. This is particularly critical because, while it is firmly established that selective reporting is common practice in the published literature \citep{chan2004EmpiricalEvidenceSelective}, \emph{the raw data is often not available to the reader who wishes to correct for it}.
\item We specialize our approach to the setting of \gls{anova}, in which the F-screening procedure in Box~\ref{box:box1}, with Step 2 conducted using a ``standard'' approach that does not account for Step 1, is particularly pervasive (see, e.g., \cref{fig:table1}).
\end{enumerate}

\begin{figure}[h]
  \includegraphics[width = \textwidth]{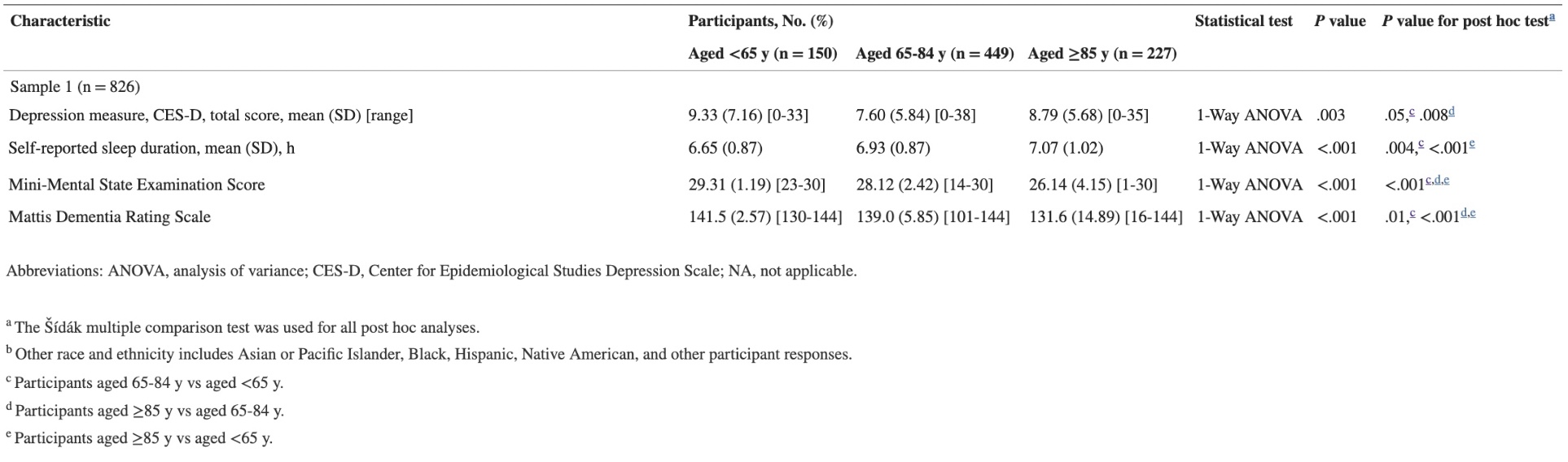}
  \caption{A subset of Table 1 from \citet{keil2023LongitudinalSleepPatterns}. The $n=826$ participants in the dataset are split into three age groups: under 65 years old, between 65 and 85 years old, and over 85 years old. Each row in the table corresponds to a continuous variable of interest, with means and standard deviations reported for each age group. For each of these continuous variables, \gls{anova} is conducted to test for a difference in means between the three age groups (corresponding to Step 1 of Box~\ref{box:box1}). The associated p-value is reported in the ``P Value'' column. The last column holds p-values for ``post hoc tests'' which are typically conducted \emph{only} if the p-value associated with the \gls{anova} test is small (corresponding to Step 2 of Box~\ref{box:box1}). In this case, the p-values in the last column would be invalid since they do not account for the fact that they were only computed because $\hov$ \eqref{eq:hov} was rejected. }
  \label{fig:table1}
\end{figure}

Addressing the inferential problems posed by F-screening is especially important
from an applied perspective. A theoretical statistician unaccustomed to scientific
collaboration might believe that F-screening is an unusual procedure unlikely to
be conducted in practice, or that conducting Step 2 of Box~\ref{box:box1}
without accounting for Step 1 is ``obviously wrong'' and no data analyst would
consider it. The theoretical statistician would be mistaken on both counts.
Statistical references that encourage F-screening with ``standard'' inference in Step 2 as part of a routine data
analysis abound, and some are authored by luminaries in the field.\footnote{\citet[\S 24]{fisher1935DesignExperiments} writes:
  ``If the yields of
  the different varieties in the experiment fail to satisfy the test of
  significance they will not often need to be considered further . . . If,
  however . . . the null hypothesis has been falsified . . . we shall thereafter
  proceed to interpret the differences between the varietal yields . . . and
  shall be concerned to know with what precision these different yields have
  been evaluated.''
    A few years later, \citet[page
  98]{lindquist1940StatisticalAnalysisEducational} writes
  ``The student may well ask why it was necessary to apply the \(F\)-test at
  all, if we were subsequently going to test the individual differences anyway.
  The answer is that the \(F\)-test tells us whether it is worthwhile to test the
  individual difference at all. If the \(F\)-test had not proven significant, we
  would have known at once that \emph{all} observed differences in . . . means
  could be due to chance alone. In that case it would not only have been
  unnecessary, but decidedly \emph{improper} to apply the \(t\)-test to individual
  differences.''
  The abstract of \citet{scheffe1953MethodJudgingAll} poses
  the following question: ``Under the usual assumptions, if the conventional
  $F$-test of $H: \mu_1=\mu$ . . . rejects $H$, what further inference are valid
  about the contrasts $H_i$?'' \citet{scheffe1953MethodJudgingAll} later
  remarks that ``The problem of making further inferences about the contrasts,
  arising when the $F$-test rejects $H$, has been considered by various writers,
  including \citet{fisher1935DesignExperiments},
  \citet{newman1939DistributionRangeSamples},
  \citet{tukey1949ComparingIndividualMeans},
  \citet{nandi1951AnalysisVarianceTest}.''
  As pointed out in \citet{olshen1973ConditionalLevelTest}, to the extent that these references consider selection, they focus on the multiplicity issues that arise from testing $\hoj: \beta_j=0$, $j=1,\ldots,p$, rather than on the challenge that arises from inference conditional on rejecting $\hov$. See Appendix~\ref{app:multipletesting} for simulation results demonstrating that multiple testing corrections do not correct for selection in the F-screening procedure.}
Thus, it is hardly surprising that in the scientific literature, conducting Step 2 of Box~\ref{box:box1} with a ``standard'' approach that does not account for having conducted Step 1 is often
viewed as the correct analysis.\footnote{As one of many possible examples, \citet{webb2024NeighborhoodResourcesAssociated} describe their analysis as follows: ``A one-way \gls{anova} revealed that amygdala reactivity was not significantly different by trajectory groups . . . Therefore, additional analyses testing whether greenspace was associated with trajectory assignment via amygdala reactivity were not conducted.''} In clinical journals, F-screening is a
routine part of ``Table 1,'' in which patient characteristics are summarized; an
example is shown in \cref{fig:table1}.

In \cref{fig:table1}, \citet{keil2023LongitudinalSleepPatterns} report the results of multiple ANOVAs for several outcomes of interest, along with the corresponding post hoc tests. We note that if these post hoc results are reported for all outcomes, regardless of the result in Step 1 of Box~\ref{box:box1}, then no selective reporting problem arises and standard marginal guarantees remain intact. However, if subsequent inference is only conducted or reported when $\hov$ \eqref{eq:hov} is rejected, then this inference is implicitly conditional on the rejection event, and standard procedures no longer provide valid guarantees. Further, we emphasize that this issue is not limited to studies that consider multiple ANOVAs or linear models: \emph{even when Step 1 of Box~\ref{box:box1} is conducted for a single model, Step 2 must account for the fact that we would not have conducted it had $\hov$ not been rejected in Step 1.} 

Finally, we note that F-screening is likely conducted much more
frequently than the published literature suggests, both because (i) it is typically unknown whether an analysis would have been included in a published paper if $\hov$ had not been rejected, and (ii) an analysis for which $\hov$ is
not rejected may not be published at all. Given that F-screening is both
commonplace and problematic if inference in Step 2 does not account for Step 1, it came as a surprise to us that---to the best of
our knowledge---no statistical solution has yet been proposed for valid inference in Step 2. The rest of this paper fills this gap.

The paper is organized as follows. In Section~\ref{sec:method}, we propose a
conditional selective inference approach for valid inference in Step 2 of Box~\ref{box:box1}. We establish its theoretical properties and derive an expression for the leftover Fisher information after rejection of $\hov$ \eqref{eq:hov} in Step 1. Section~\ref{sec:retrospective} shows that it is possible to compute selective p-values using the results of a ``standard'' (non-selective) analysis in Step 2, \emph{without having access to the raw data}. Section~\ref{sec:anova} specializes the results in
Sections~\ref{sec:method} and \ref{sec:retrospective} to the case of
\gls{anova}. Simulation results are in Section~\ref{sec:simulation} (and \cref{app:extended_sim}), and re-analyses of published findings in \cref{sec:data}. The discussion is in \cref{sec:discussion}. The main text focuses on developing valid selective p-values; confidence intervals
and point estimates are deferred to Appendix \ref{app:ci_point_estimates}.
\cref{app:geometry} contains mathematical details for Figures~\ref{fig:3D_geometry} and \ref{fig:2D_geometry}. Proofs of all results are given in Appendix~\ref{app:proofs}.

\section{Developing valid p-values in Step 2 of F-screening }
\label{sec:method}

\subsection{Preliminaries} \label{subsec:formulation}

Let $P_A := A(A^\top A)^{-1}A^\top$ denote the projection onto the column space
of a matrix $A$ with full column rank, and let $A_{-S}$ denote the matrix $A$
with the columns indexed by $S\subset \{1,\ldots,p\}$ removed. Furthermore,
let $\mathbbm{1}_n$ denote the $n$-dimensional column vector of 1's, and $I_n$ the $n\times n$ identity matrix. This section restates some standard results in
the theory of linear models \citep[a more
thorough treatment can be found elsewhere,
e.g.,][]{seber2003LinearRegressionAnalysis}.

Consider $n$ independent observations from the model in \eqref{eq:model}. We
assume that the intercept is not of scientific interest, and so both $\hov$ and $\hoM$ will operate on centered data. To
do this, we let $U\in \R^{n\times (n-1)}$ be any matrix satisfying both \(U^\top
U = I_{n-1}\) and $UU^\top = I_n - P_{\1_n}$, and then with slight abuse of
notation we redefine $Y := U^\top Y$, $X_1 := U^\top X_1$, and so on. Thus, in
what follows, we will treat $\mathbb{R}^{n-1}$ as the ambient space and dispense
with the intercept.

First, consider the null hypothesis $\hov$ given in \eqref{eq:hov}. Under $\hov$, the
$F$-statistic defined as
  \begin{align}
        \Fov := \frac{n-p-1}{p} \cdot \frac{Y^\top P_X Y}{Y^\top (I_{n-1}-P_X) Y}  \label{eq:Foverall}
\end{align}
follows an $F_{p, n-p-1}$ distribution (that is, an $F$ distribution with $p$ and $n-p-1$ degrees of freedom). Next, fix \(\hoM\) by
letting $M\subset \{1,\ldots,p\}$ denote a set of indices with
cardinality $\lvert M \rvert=m<p$, and let $\beta_M$ denote the subset of coefficients indexed
by $M$ (in the simplest case, $M=\{j\}$ for fixed \(j \in \left\{ 1, \ldots, p
\right\}\)). Under the null hypothesis $\hoM$, the $F$-statistic
\begin{equation}
\FhoM := \frac{n-p-1}{m}\cdot \frac{Y^\top (P_X - P_{X_{-M}}) Y }{Y^\top (I_{n-1}- P_X) Y} 
 \label{eq:fstat}
\end{equation}
follows an $F_{m,n-p-1}$ distribution. Thus, for a realization \(y\) of the
random variable \(Y\), we can compute the p-value for $\hoM: \beta_M=0$ as
\begin{align}
\phoM ( y ) & := \Pr_{\beta_M=0} \left(  \frac{Y^\top (P_X - P_{X_{-M}}) Y }{Y^\top (I_{n-1}- P_X) Y} \geq \frac{ y^\top (P_X - P_{X_{-M}}) y }{y^\top (I_{n-1}- P_X) y} \right)\nonumber \\
&= 1-F_{m,n-p-1} \left(  \frac{ y^\top (P_X - P_{X_{-M}}) y }{y^\top (I_{n-1}- P_X) y} \cdot \frac{(n-p-1)}{m}   \right), \label{eq:pnaive}
\end{align}
where \(F_{m, n - p - 1} \left( \cdot \right)\) is the cumulative distribution
function of a random variable with an \(F_{m, n - p - 1}\) distribution. This is the ``standard'' test of $\hoM:\beta_M=0$, i.e., it does not account for Step 1 of Box~\ref{box:box1}. 

\cref{fig:3D_geometry} displays the rejection regions corresponding to a test of $\hov$ using \eqref{eq:Foverall} and a test of $\hoM$ using \eqref{eq:fstat} in a simple setting with $n=3$, $p=2$, and $X$ orthogonal. 
The geometry of these two events interplays in non-trivial ways, and failing to account for the first when evaluating the second can yield invalid inference.
By contrast, the selective test that we will develop in the next section will take the sequential nature of the analysis into account, and it will adjust the geometry of the second test accordingly.

\begin{figure}[h]
  \includegraphics[width = \textwidth]{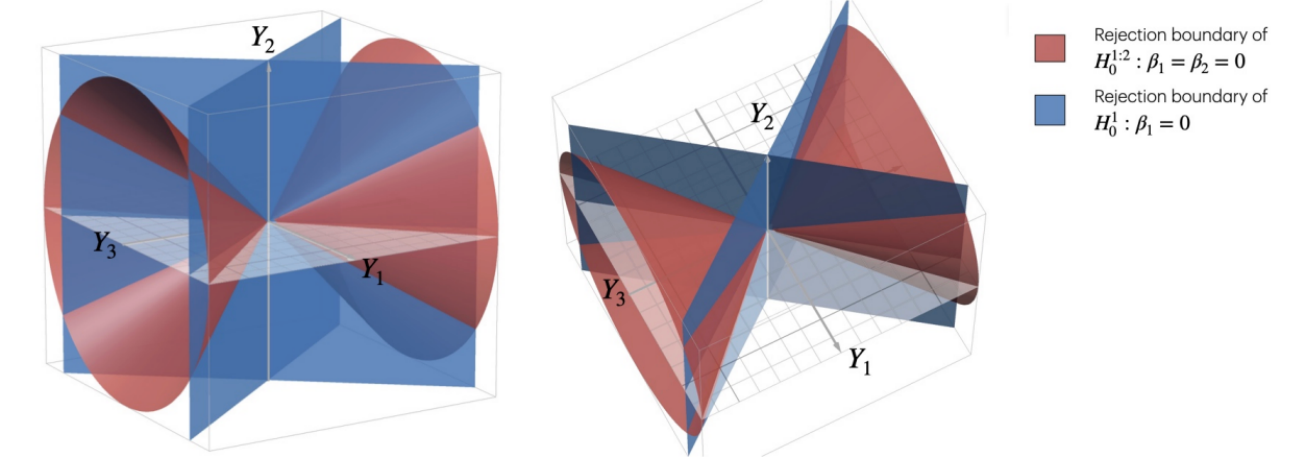}
  \caption{We consider testing $H_0^{1:2}:\beta_1=\beta_2=0$ using $F_{H_0^{1:2}}$ \eqref{eq:Foverall} and $\hoone:\beta_1=0$ using $F_{H_0^1}$ \eqref{eq:fstat} in the model $Y=X\beta+\epsilon$, where $ X = \begin{bmatrix} 1 & 0 \\ 0 & 1
  \\ 0 & 0 \end{bmatrix}$. The red double cone corresponds to the rejection boundary of the test $H_0^{1:2}$, and the blue planes correspond to the rejection boundary of $\hoone$. Thus data points that lie farther from the $Y_3$-axis than the double cone lead to rejection of $H_0^{1:2}$, and data points that lie farther from the $(Y_2,Y_3)$-plane than the two blue planes lead to rejection of $\hoone$. See \cref{app:geometry} for mathematical details. 
    }
  \label{fig:3D_geometry}
\end{figure}

\subsection{A conditional selective inference approach}
\label{subsec:method}

We test $\hoM$ in Step 2 of Box~\ref{box:box1} if and only if we reject $\hov$ \eqref{eq:hov}
at level $\alphaov\in (0,1)$: that is, if and only if $Y\in E_1$, where
\begin{equation}
  E_1 := \left\{Y:\frac{n-p-1}{p} \cdot \frac{Y^\top P_X Y}{Y^\top (I_{n-1}- P_X) Y}
    \geq F^{-1}_{p,n-p-1}(1-\alphaov)\right\} .\label{eq:rejectFov}
\end{equation}

The left panel of \cref{fig:triple_plot} shows that conditional on \(Y \in E_{1}\) \eqref{eq:rejectFov}, the p-values obtained using \(\FhoM\) are far from uniform and indeed anti-conservative---i.e., a test based on $\phoM$ \eqref{eq:pnaive} does not control the selective Type 1 error in Step 2 of Box~\ref{box:box1}. This is because $\FhoM$ \eqref{eq:fstat} does not follow an $F_{m, n-p-1}$ distribution under $\hoM$ when we condition on \(Y \in E_{1}\) \eqref{eq:rejectFov}. Our goal is to develop an alternative to
$\phoM$ \eqref{eq:pnaive} that \emph{does} lead to selective Type 1 error control.

To illustrate this idea, we consider in \cref{fig:2D_geometry} a simple setting with $p=2$ and  $\sigma^2$  known. The left-hand panel displays the rejection regions for $\chi^2$-tests of $H_0^{1:2}$ and $H_0^1$, respectively; this does not account for the sequential nature of the two tests in the context of F-screening.  The right-hand panel displays the rejection region of a $\chi^2$-test of $H_0^{1:2}$, alongside the rejection region of
a test of $H_0^1$ that conditions upon rejection of $H_0^{1:2}$. 

\begin{figure}[h]
  \includegraphics[width = \textwidth]{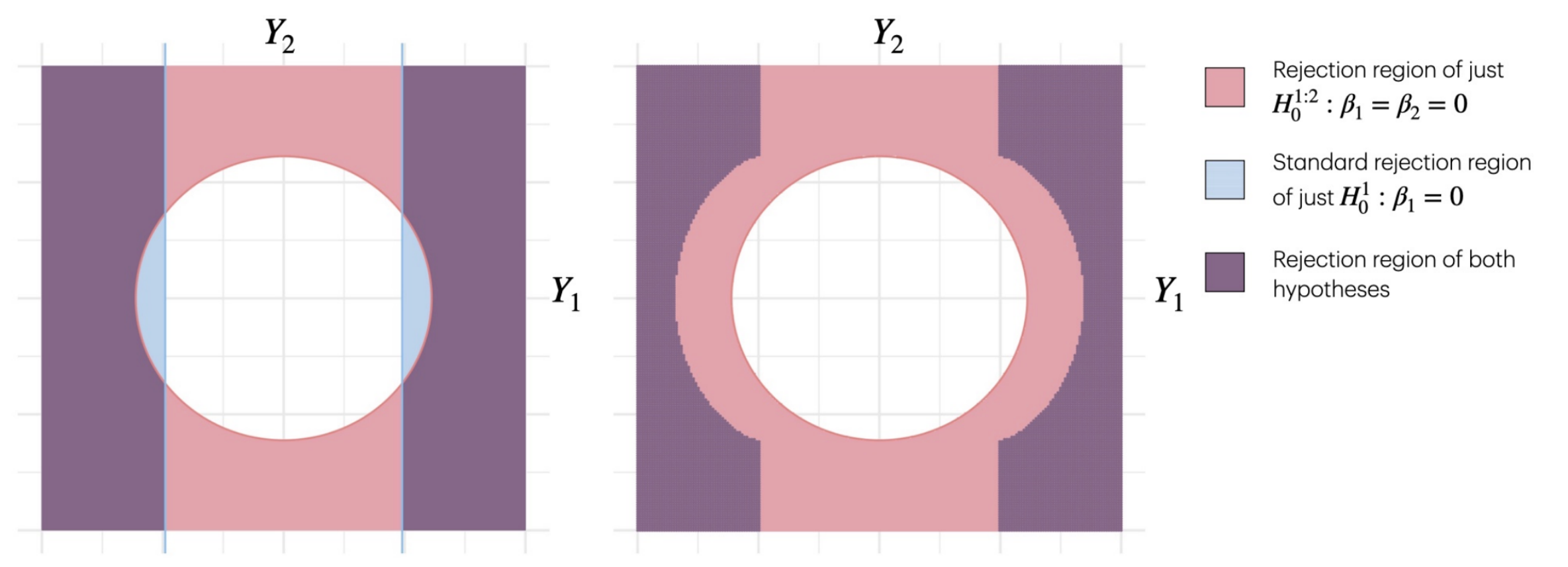}
  \caption{We consider testing $H_0^{1:2}:\beta_1=\beta_2=0$ and $\hoone:\beta_1=0$ in the model $Y=X\beta+\epsilon$ where $ X = \begin{bmatrix} 1 & 0 \\ 0 & 1 \end{bmatrix}$ and we take the variance of $\epsilon$ to be known. In the left plot we use $\chi^2$-tests to test $H_0^{1:2}$ and $\hoone$; this is in analogy to F-screening in Box~\ref{box:box1} where Step 2 is conducted using a ``standard'' test. On the right, we use a $\chi^2$-test to test $H_0^{1:2}$, and then we use a test of $\hoone$ in Step 2 that accounts for the rejection of $H_0^{1:2}$ in Step 1. 
  This leads to a substantial change in the geometry of the rejection region of $\hoone$. Mathematical details are in \cref{app:geometry}.
    }
  \label{fig:2D_geometry}
\end{figure}

To extend this idea to the setting of the paper, we make use of $F$-tests rather than $\chi^2$-tests. 
We modify $\phoM$ in \eqref{eq:pnaive} to
condition on the selection event in \eqref{eq:rejectFov}, leading to
\begin{align}
\label{eq:p-ideal}
 \Pr_{\beta_M=0} \left(  \frac{ Y^\top (P_X - P_{X_{-M}}) Y } {Y^\top (I_{n-1}- P_X) Y}  \geq  \frac{ y^\top (P_X - P_{X_{-M}}) y }{y^\top (I_{n-1}- P_X) y}  \middle|\: Y\in E_1 \right).
\end{align}
Unfortunately, \eqref{eq:p-ideal} is not a well-defined
probability: it involves nuisance parameters $\left\{ \beta_j \right\}_{j \notin M}$ and $\sigma^2$ which are not fixed under \(\hoM\), and thus it cannot be computed. To remedy this, we
additionally condition on the events
\begin{equation}\label{eq:E2E3}
    E_2(y):= \left\{ Y: \|P_{{X}_{-M}} Y\| =  \|P_{{X}_{-M}} y\|\right\}\quad \text{and} \quad E_3(y) =\left\{ Y: \|Y\| = \|y\| \right\}.
\end{equation}
We then define
\begin{equation}
E(y) := E_1 \cap E_2(y) \cap E_3(y).
\label{eq:E}
\end{equation}
This leads to the p-value
\begin{equation}
\psel (y)   :=  \Pr_{\beta_M=0} \left(  \frac{ Y^\top (P_X - P_{{X}_{-M}}) Y }{Y^\top (I_{n-1}- P_X) Y}  \geq  \frac{ y^\top (P_X - P_{{X}_{-M}}) y }{y^\top (I_{n-1}- P_X) y} \middle| Y \in E(y)  \right).
\label{eq:pselective}
\end{equation}
\begin{remark}\label{rmk:conditioning}
    In $\psel$, we condition on the realizations of $\|Y\|$ and $\|P_{X_{-M}}Y\|$. As discussed in Section 4.1 of \citet{fithian2017OptimalInferenceModel}, the natural parameters for the model $Y\sim \mathcal{N}_{n-1}(X\beta, \sigma^2 I_{n-1})$ are $\beta/\sigma^2$ and $1/(2\sigma^2),$ with sufficient statistics $X_k^\top Y$ for each $\beta_k/\sigma^2$ and $\|Y\|^2$ for $1/(2\sigma^2)$. Hence, conditioning on $\|Y\|$ and $P_{X_{-M}}Y$ removes dependence on $\beta_{-M}$ and $\sigma^2$, and we are left with a distribution that depends on $\beta_M/\sigma^2$. In fact, it suffices to condition on $\|P_{X_{-M}}Y\|$ in place of the full vector $P_{X_{-M}}Y$. Thus, we are able to test $\beta_M/\sigma^2=0$, or equivalently $\hoM:\
    \beta_M=0$, as we show in \cref{prop:chisq}.  
\end{remark}
The next result establishes that a test of $\hoM$ based on $\psel(y)$ controls the selective Type 1 error, conditional on the event that $\hov$ was rejected. 
\begin{theorem}
  \label{thm:psel-type1}
A test of $\hoM$ based on the p-value $\psel(y)$ \eqref{eq:pselective} controls the selective Type 1 error, conditional on the event that $\hov$ was rejected. That is, for any $\alpha' \in (0,1)$, 

$$\Pr_{\beta_M=0}\left( \psel(Y) \leq \alpha' \mid Y\in E_1 \right) = \alpha'. $$
\end{theorem}  
The left panel of \cref{fig:triple_plot} empirically corroborates this result. The next result shows that $\psel(y)$ can be easily approximated via Monte Carlo.

\begin{proposition}
    \label{prop:chisq}
    Let $Z \sim \chi_{n-p-1}^2$ and $W \sim \chi_{m}^2$, where $Z \perp W$ (i.e., \(Z\) is independent of \(W\)). Define 
    \begin{equation}\label{eq:acdr}
        a(y) : = \|y\|^2, \quad c := F^{-1}_{p,n-p-1}(1-\alphaov)\cdot\tfrac{p}{n-p-1}, \quad d(y) := \|P_{X_{-M}}y\|^2, \quad 
    r(y) := \tfrac{ y^\top (P_X - P_{{X}_{-M}}) y }{y^\top (I_{n-1}- P_X) y}.
    \end{equation} 
    Then,
\begin{equation}\label{eq:pselchisq}
\psel(y) = 
        \Pr \left ( 
        \frac{W}{Z} \geq r(y)\,\middle|\frac{a(y)\cdot W + d(y)\cdot  Z}{(a(y)-d(y))\cdot Z}\geq c\right).
\end{equation}
\end{proposition}

To summarize, a test based on $\psel(y)$ in \eqref{eq:pselective} controls the
selective Type 1 error for $\hoM$ in the context of F-screening as described in
Box~\ref{box:box1} (Theorem~\ref{thm:psel-type1}). Furthermore, $\psel(y)$ can be
easily approximated via Monte Carlo (Proposition~\ref{prop:chisq}). Obtaining confidence intervals for
$\beta_1,\ldots,\beta_p$ that attain the nominal selective coverage (conditional
on rejecting $\hov$), and correcting the corresponding point estimates for
selection, are straightforward extensions of the ideas in this section. Details
are provided in Appendix~\ref{app:ci_point_estimates}.

We close this subsection with a result concerning the power of the test that results from using the p-value in \eqref{eq:pselchisq}.

\begin{proposition}\label{prop:test_consistency}
Consider a fixed sequence of design matrices \(X_{n} \in \mathbb{R}^{n \times p}\) and corresponding realizations of the linear model \(Y_{n} = X_{n} \beta + \epsilon_{n}\) from \eqref{eq:model}. Suppose that the design matrices satisfy \(\frac{1}{n} X_{n}^{\top} X_{n} \to Q\) for some positive definitive Q. Then the test of $\hoM:\beta_M=0$  that uses the selective p-value in \eqref{eq:pselchisq} is consistent under every alternative, conditional on $Y_{n} \in E_1$. That is, for any \(\beta\) satisfying \(\beta_{M} \neq 0\) and arbitrary \(\alpha' \in (0, 1)\),
$\lim_{n\to\infty}\;\Pr_{\beta}\left(\psel \left( Y_{n} \right) \leq \alpha'|Y_{n} \in E_1\right)=  1.$
\end{proposition}

\begin{remark}\label{remark:assumptions}
    In \cref{prop:test_consistency}, the following quantities are fixed with respect to $n$: $p\geq 1$, $\beta\in \R^p$, $\sigma^2>0$, $M\subset \{1,\ldots,p\}$. In contrast, $X_n$ and $Y_n$ grow in length with $n$, and the p-value in \eqref{eq:pselchisq} depends on $n$ through $Z\sim \chi^2_{n-p-1}$, as well as $c$, $d$, and $r$ in \eqref{eq:acdr}. 
    In the remainder of the paper, for conciseness, we only include the subscript $n$ when we wish to emphasize dependence on $n$. 
    We may also suppress the dependence of some quantities on $y$ (e.g., $\psel(y)$ in \eqref{eq:pselchisq} may be written \(\psel\)). %
\end{remark}

\subsection{F-Screening and (Conditional) Fisher Information}
\label{sec:fisher-info}

We now consider the test of $H_0^M$  in Step 2 of Box~\ref{box:box1} through the lens of Fisher information in a simplified setting in which $\sigma^2$ is known and $\beta_0=0$. Hence, we consider the model \eqref{eq:model} without the intercept, and the F-tests in Steps 1 and 2 of Box~\ref{box:box1} are replaced with  $\chi^2$-tests. The hypothesis of interest in Step 2 is $H_0^1:\beta_1=0$. We  compare our proposed approach (specialized to this setting) to a sample splitting procedure.

In the sample splitting procedure, we partition the data into $(\Xtrain, \Ytrain)$ and $(\Xtest, \Ytest)$, where a proportion $\rho \in (0,1)$ of the observations are assigned to the training set and the remaining observations are assigned to the test set such that $\rho n$ and $(1-\rho)n$ are integers. We test $\hov$ with $(\Xtrain, \Ytrain)$, and then $(\Xtest, \Ytest)$ is used to test $\hoone$ if and only if $\hov $ is rejected.

In the spirit of \citet[Section 2.5]{fithian2017OptimalInferenceModel}, we derive expressions for:
(i) the leftover Fisher information in $\Ytest$ about $\beta_1$ conditional on the event that $H_0^{1:p}$ is rejected with $(\Xtrain,\Ytrain)$; and
(ii) the leftover Fisher information in $Y$ about $\beta_1$ conditional on the event that both the overall null $H_0^{1:p}$ is rejected using the full data $(X,Y)$ and $P_{X_{-1}}Y = P_{X_{-1}}y$ for some realization $y\in \R^n$ (note that we do not additionally condition on the realization of $\|Y\|$ because $\sigma^2$ is known). %

\begin{definition}[adapted from \citealp{fithian2017OptimalInferenceModel}]\label{def:leftover_info}
    Let $\ell(\theta; Y\mid Y\in A)$ denote the conditional log-likelihood of $Y$ given $Y\in A$. The \textit{leftover Fisher information} about the parameter $\theta$ in data $Y$ conditional on $Y \in A$ is $\mathcal{I}_{Y|Y\in A}(\theta\mid Y\in A) := -\E\left[\nabla ^2 \ell (\theta; Y\mid Y\in A)\mid Y\in A\right].$ 
    \end{definition}

\begin{proposition}\label{prop:fisher-info}
\textcolor{white}{.}
\begin{enumerate}[i.]
\item Define $c_\chi = \sigma^2 (\chi^2_p)^{-1}(1-\alphaov)$ and the set 
\begin{equation}\Rrho(X)  := \{Y \in  \R^{(\rho  n) }: Y^\top P_X Y \geq c_\chi\}. \label{eq:Rrho}
\end{equation}

Treating $\beta_2,\ldots,\beta_p$ as nuisance parameters, the leftover Fisher information in $\Ytest$ about $\beta_1$ conditional on  $\Ytrain \in \Rrho(\Xtrain)$ equals the (unconditional) Fisher information about $\beta_1$ in $\Ytest$. This quantity is 
$\frac{1}{\sigma^2}\left((\Xtest)_1^\top (I - P_{(\Xtest)_{-1}})(\Xtest)_1\right)$.\\ 
\item For any $\tilde y\in \R^n,$
define 
\begin{equation}
R(\tilde y;X) := \{Y\in \R^n: Y^\top P_XY \geq c_\chi,\; P_{X_{-1}}Y = P_{X_{-1}}\tilde y\}.\label{eq:RyX}
\end{equation}
Treating $\beta_2,\ldots,\beta_p$ as nuisance parameters, the leftover Fisher information in $Y$ about $\beta_1$ conditional on $Y \in R(\tilde y; X)$ is given by %
\begin{align}
&\mathcal{I}_{Y\mid   Y \in R(\tilde y; X)}(\beta_1;Y \in R(\tilde y; X)) \nonumber \\
&=\begin{cases}
    \tfrac{X_1^\top (I-P_{X_{-1}})X_1 }{\sigma^2} & \text{if }\; \tilde y^\top P_{X_{-1}} \tilde y \geq  c_\chi\\
    \tfrac{X_1^\top (I-P_{X_{-1}})X_1 }{\sigma^2} \left[ 1 - \left(\tfrac{[\phi(a) -\phi(b)]^2}{[1-\Phi(b)+\Phi(a)]^2}+\tfrac{\phi(a)\cdot a - \phi(b) \cdot b}{1-\Phi(b)+\Phi(a)}\right)\right] & \text{if } \; \tilde y^\top P_{X_{-1}} \tilde y < c_\chi \label{eq:info_sel}
\end{cases}
\end{align} 
where
\begin{align}a=\tfrac{-\sqrt{c_\chi - \tilde y^\top P_{X_{-1}} \tilde y}-\beta_1\sqrt{X_1^\top (I_n-P_{X_{-1}})X_1}}{\sigma}, \quad b=\tfrac{\sqrt{c_\chi - \tilde y^\top P_{X_{-1}} \tilde y}-\beta_1\sqrt{X_1^\top (I_n-P_{X_{-1}})X_1}}{\sigma},%
\end{align}
and $\phi$ and $\Phi$ denote the density and distribution functions, respectively, of the standard normal distribution.
\end{enumerate}
\end{proposition}

For insight into the two cases in \eqref{eq:info_sel}, note that if $\tilde y^\top P_{X_{-1}} \tilde y \geq c_\chi$, then the set $R(\tilde y; X)$ in \eqref{eq:RyX} can more succinctly be characterized as $R(\tilde{y}; X) := \left\{Y \in \mathbb{R}^n : P_{X_{-1}} Y = P_{X_{-1}} \tilde y \right\}$, because $\tilde y^\top P_{X_{-1}} \tilde y \geq c_\chi$ and $P_{X_{-1}} Y = P_{X_{-1}} \tilde y$ imply that $Y^\top  P_{X_{-1}} Y \geq c_\chi$ and hence $Y^\top  P_X Y \geq c_\chi$.
Furthermore,
$\tfrac{\partial^2}{\partial \beta_1^2}\ell(\beta_1; Y \mid P_{X_{-1}} Y =P_{X_{-1}} \tilde{y})  =  X_1^\top (I-P_{X_{-1}})X_1 / \sigma^2 ;$
this is simply the (unconditional) Fisher information about $\beta_1$ when $\beta_2,\ldots,\beta_p$ are treated as nuisance parameters.
Therefore, when $\tilde y^\top P_{X_{-1}} \tilde y \geq c_\chi$, conditioning on the set $R(\tilde y; X)$  does not affect the Fisher information about $\beta_1$.
On the other hand,
when $\tilde y^\top P_{X_{-1}}\tilde y < c_\chi$, this more succinct characterization of $R(\tilde y; X)$ is no longer possible, and
$\tfrac{\partial^2}{\partial \beta_1^2}\ell(\beta_1; Y \mid R(\tilde y;X)) \neq  X_1^\top (I-P_{X_{-1}})X_1 / \sigma^2.$ Thus, when $\tilde y^\top P_{X_{-1}}\tilde y < c_\chi$, conditioning on $Y\in R(\tilde y; X)$ indeed affects the Fisher information about $\beta_1$.

Recall that we test $\hoone$ if and only $\hov$ is rejected with a $\chi^2$-test, so the information about $\beta_1$ is only relevant when $Y^\top P_XY \geq c_\chi$. Therefore, the quantity of interest is actually the leftover Fisher information available for inference about $\beta_1$ \emph{times an indicator for the rejection of $\hov$}: that is, the random variable $\mathcal{I}_{Y|Y\in R(\tilde Y ;X)}(\beta_1;Y\in R(\tilde Y ;X))\cdot \bm{I}_{R_1(X)}(\tilde Y)$,
where $\tilde Y\overset{d}{=} Y$, $R_1(X)$ arises from setting $\rho=1$ in \eqref{eq:Rrho}, and $\bm{I}_{R_1(X)}(\tilde Y)$ is the indicator for $\tilde Y \in R_1(X)$. 
The expectation of this quantity (with respect to $\tilde{Y})$ can be expressed as
\begin{align}\label{eq:info_exp_sel}
    \E\Big[\mathcal{I}_{Y|Y\in R(\tilde Y ;X)}&(\beta_1;Y\in R(\tilde Y ;X))\cdot \bm{I}_{R_1(X)}(\tilde Y)\Big]\nonumber  \\
    &=\E\left[\mathcal{I}_{Y|Y\in R(\tilde Y ;X)}(\beta_1;Y\in R(\tilde Y ;X))\mid \tilde Y\in R_1(X)\right]\cdot \Pr(\tilde Y\in R_1(X)).
\end{align}
The right-hand side of \eqref{eq:info_exp_sel} is the product of two terms: the expected leftover Fisher information conditional on rejecting $\hov$, and  the probability of rejecting $\hov$.

To approximate the conditional expectation in \eqref{eq:info_exp_sel}, in the left-hand panel of \cref{fig:info} we generate 1000 draws of $\tilde Y$ according to \eqref{eq:model} (with $\beta_0=0$) and display the average of \eqref{eq:info_sel} \textit{over all realizations of $\tilde Y$ such that $\tilde Y^\top P_X \tilde Y\geq c_\chi$}.
 The second term on the right-hand side of \eqref{eq:info_exp_sel} can be computed explicitly by noting that, for $R_\rho(X)$ defined in \eqref{eq:Rrho} and $\lambda = \frac{1}{2\sigma^2}\beta^\top X^\top X\beta$, 
$\Pr(Y\in \Rrho(X))  = 1  - F_{p,n\cdot\rho - p - 1}(\lambda; F^{-1}_{p,n\cdot\rho-p-1}(0;1-\alphaov)),$ where $F_{df_1,df_2}(\lambda; \cdot)$ denotes the cumulative distribution function of a noncentral $F$ random variable with $df_1$ and $df_2$ degrees of freedom and noncentrality parameter $\lambda$. This quantity is displayed in the middle panel of \cref{fig:info}. Finally, in the right-hand panel of \cref{fig:info} we display the product of the two terms in the left-hand and center panels, i.e., the right-hand side of \eqref{eq:info_exp_sel}.

For the sample splitting procedure, we plot an analogous set of quantities. Recalling from Proposition~\ref{prop:fisher-info}(i) that the leftover Fisher information in $\Ytest$ about $\beta_1$ conditional on $\Ytrain \in R_\rho(\Xtrain)$ is $\frac{1}{\sigma^2}\left((\Xtest)_1^\top (I - P_{(\Xtest)_{-1}})(\Xtest)_1\right)$, it follows that 
\begin{align}
    \E&\left[\tfrac{(\Xtest)_1^\top (I - P_{(\Xtest)_{-1}})(\Xtest)_1}{\sigma^2} \cdot \bm{I}_{ R_\rho(\Xtrain)}(\Ytrain)\right] =\tfrac{(\Xtest)_1^\top (I - P_{(\Xtest)_{-1}})(\Xtest)_1}{\sigma^2} \cdot \Pr(\Ytrain \in \Rrho(\Xtrain)).\label{eq:info_exp_split}
     \end{align}
In the left-hand and center panels of \cref{fig:info}, we display the two terms on the right-hand side of  \eqref{eq:info_exp_split}, and in the right-hand panel we display their product.

In the left-hand panel of \cref{fig:info}, we see that the selective approach generally results in higher  expected leftover Fisher information about $\beta_1$, conditional on rejection of $\hov$, as compared to sample splitting. Furthermore, the center panel shows that the probability of rejecting $\hov$ is always larger using the selective approach than using sample splitting; this is due to the fact that the former uses all $n$ observations whereas the latter uses only $\rho n \leq n$.  Finally, the right-hand panel shows that the product of the left and middle panels (representing \eqref{eq:info_exp_sel} and \eqref{eq:info_exp_split}) is larger for the selective approach than for sample splitting, for all values of $\beta_1,\ldots,\beta_p$ considered.

\begin{remark}
We see from \eqref{eq:info_sel} that when $\beta_1 = 0$, conditioning on $R(\tilde{Y}; X)$ in \eqref{eq:RyX} can increase---and cannot decrease---the leftover Fisher information about $\beta_1$ relative to the unconditional Fisher information. This is reflected in the bottom-left panel of Figure~\ref{fig:info}. However, for larger values of $|\beta_1|$, conditioning typically does not increase the leftover Fisher information. 
\end{remark}

\begin{figure}
\centering
\begin{tabular}{@{}m{.5cm} c}   %
  \raisebox{13em}{%
    \begin{tabular}{c}
      \scriptsize $\beta_1 = 0.5$ \\[8.8em]
      \scriptsize $\beta_1 = 0.25$ \\[8.8em]
      \scriptsize $\beta_1 = 0$
    \end{tabular}
  }
  &
  \begin{tabular}{c}
    \scriptsize
    \begin{tabular}{ccc}
    \hspace{2em}\textbf{$\E\big[$Leftover Fisher $\mid \text{rej }\bm{\hov}\big]$} & \hspace{4em}\textbf{Pr(rej $\bm{\hov}$)} & \hspace{4em}\makecell{\textbf{$\E\big[$Leftover Fisher $\mid \text{rej }\bm{\hov}\big]\bm{\cdot} $}\\\textbf{Pr(rej $\bm{\hov}$)}}
    \end{tabular} \\
    \includegraphics[width=0.85\linewidth]{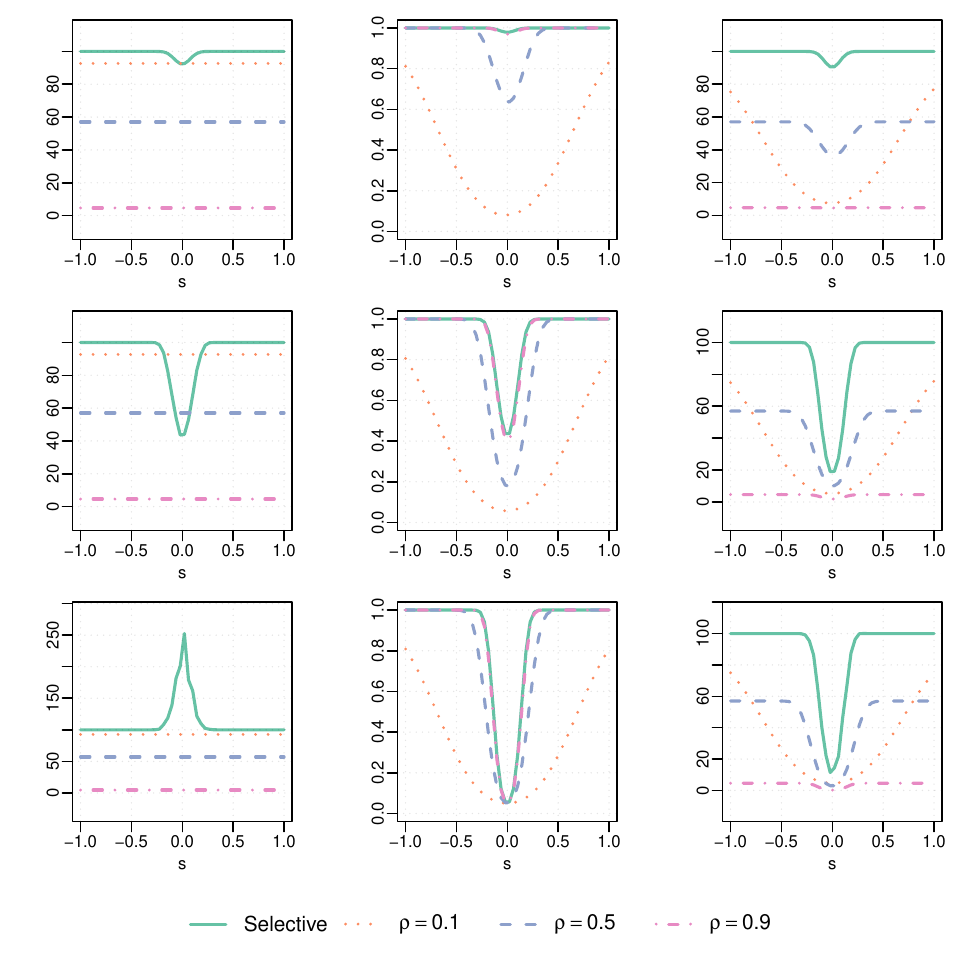}
  \end{tabular}
\end{tabular}
\caption{We construct an $n \times p$ orthogonal design matrix $X$ with $n=100$ and $p=5$. Then, for  $\beta_1\in \{0.5,0.25,0\}$  and for a range of values of $s\in [-1,1]$, we generate 1000 response vectors $\tilde Y$ according to \eqref{eq:model} with  $\sigma = 1$,  $\beta_{j} = s$ for all $j\in \{2, \ldots, p \}$, and $\beta_0=0$. 
We then compute quantities related to the leftover Fisher information about $\beta_1$ for our selective approach, and for sample splitting with split proportions  $\rho\in\{0.1,0.5,0.9\}$. 
\textit{Left:} For each of the 1000 realizations $\tilde y$ of $\tilde Y$, we compute the leftover Fisher information $\mathcal{I}_{Y\mid   Y \in R(\tilde y; X)}(\beta_1;Y \in R(\tilde y; X))$ for our conditional selective inference procedure and display the average \textit{over all datasets such that $\tilde Y\in R_1(X)$}. This approximates $\E\left[\mathcal{I}_{Y\mid   Y \in R(\tilde Y; X)}(\beta_1;Y \in R(\tilde Y; X))\mid \tilde Y\in R_1(X)\right]$. Further, for $\rho\in \{0.1,0.5,0.9\}$, we display $(\Xtest)_1^\top (I - P_{(\Xtest)_{-1}})(\Xtest)_1/\sigma^2$. \textit{Center:} We display $\Pr(Y \in R_1(X))$, along with  $\Pr(\Ytrain_\rho\in \Rrho(\Xtrain))$ for $\rho \in \{0.1,0.5,0.9\}$. \textit{Right:} We display the product of the left-hand and center columns, i.e., the right-hand sides of \eqref{eq:info_exp_sel} and \eqref{eq:info_exp_split}. \label{fig:info}}
\end{figure}

\section{Selective p-values in a retrospective analysis}\label{sec:retrospective}

In the
applied literature, examples abound of analyses where Step 2 would not have been carried out if Step 1 had not rejected\footnote{Here are some examples of published papers that clearly state that Step 2 was carried out if and only if $\hov$ was rejected in Step 1: (i) ``If the groups
  tested are different, we report the p-values of post hoc tests but omit the p-value of the omnibus test. If the groups did not test as different, only the P
  value of the omnibus test is reported'' \citep[page
  13]{kula2024DbhydroxybutyrateStabilizesHippocampala}; (ii) Table 2 of
  \citet{wang2024MarginStabilityAffected}, where NA's denote hypotheses that
  were left untested because the overall test was not rejected; (iii) ``Post hoc
  pairwise comparisons were only performed if the omnibus test was significant''
  \citep[caption of Figure 5]{tams2020MethadoneClinicalPathway}.}. If the
original data for such an analysis were available to the reader, then they could reanalyze the data using the methods in
Section~\ref{sec:method}. Unfortunately, the raw data are often not available.
In this section, we demostrate that one can use the usual
software output of a least squares linear regression (e.g., the output
of \verb=summary(lm(y~x))= in \verb=R=) to ``correct'' Step 2 of
Box~\ref{box:box1} when a standard test of $\hoM$ was
used. We refer to this as a \emph{retrospective} analysis, as it can be
applied retrospectively to the results of a published paper, without access to
the raw data.

Standard software packages for fitting linear models
output the proportion of variance explained ($\RR^2$) and the residual standard error ($\RSE$), defined as
\begin{align} \label{eq:RSE_R2}
\RR^2 & = 1 - \frac{ y^\top (I_{n-1}-P_X)y }{ y^\top y},  \;\;\; 
\RSE  = \sqrt{\frac{y^\top (I_{n-1}- P_X) y}{n-p-1}},  
\end{align}
as well as $\FhoM$ \eqref{eq:fstat}. These three quantities suffice to compute $\psel$ \eqref{eq:pselchisq}.

\begin{proposition} \label{prop:retropsel}
Let $\FhoM$ be defined as in \eqref{eq:fstat}, 
$c$ defined as in \eqref{eq:acdr},
and $\RR^2$ and $\RSE$ defined as in \eqref{eq:RSE_R2}.
Define $Z \sim  \chi_{n-p-1}^2$ and $W \sim  \chi_{m}^2$, where $Z \perp W$, and 
$$\dretro :=  \RSE^2 \cdot \left( \tfrac{\RR^2\cdot (n-p-1)}{1-\RR^2}-\FhoM \cdot m \right), \quad \aretro := \tfrac{\RSE^2\cdot (n-p-1)}{1-\RR^2}.$$
The p-value in \eqref{eq:pselchisq} can be written as
\begin{equation}\label{eq:pselectiveretro}
        \psel(y) = 
        \Pr \left (  \frac{W}{Z} \geq \FhoM \cdot \frac{n-p-1}{m}\,\,\middle|\frac{\aretro\cdot W + \dretro\cdot  Z}{(\aretro-\dretro)\cdot Z}\geq c\right).
    \end{equation}
\end{proposition}

In contrast, the selective confidence interval and point estimates described in
Appendix~\ref{app:ci_point_estimates} involve quantities that are not generally
reported in software output, and thus cannot be easily obtained retrospectively.

\section{The special case of \gls{anova}} \label{sec:anova}

F-screening is particularly pervasive in the context of \gls{anova}, where it dates back almost a century \citep{fisher1935DesignExperiments,lindquist1940StatisticalAnalysisEducational}. As a simple example, suppose that we are interested in how plant height, $Y$, varies between treatments by fertilizers
$A$, $B$, and $C$. We fit the model $Y = \beta_A \cdot \bm{I}_A + \beta_B \cdot \bm{I}_B + \beta_C \cdot \bm{I}_C + \epsilon,$
where $\bm{I}_A$, $\bm{I}_B$, and $\bm{I}_C$ are $n$-dimensional indicator variables for
treatment with fertilizers $A$, $B$, and $C$. In the language of \gls{anova}, $\hov:
\beta_A= \beta_B= \beta_C$ is often referred to as the ``omnibus test.'' If
$\hov$ is rejected, then one typically conducts post hoc tests
\citep{everitt2010CambridgeDictionaryStatistics} to see how the expected plant
height differs across fertilizers; for instance, we may test $H_0: \beta_A =
\beta_B$. This post hoc test can be equivalently expressed as a test of $H_0:
\beta_{\BvA}=0$ in the reparameterized model
\begin{align}\label{eq:plant2}
  Y &= \beta_{A} \cdot (\bm{I}_{A}+\bm{I}_{B}) + (\beta_{B} -\beta_{A})\cdot \bm{I}_{B} + \beta_{C}\cdot 
  \bm{I}_{C} + \epsilon \nonumber\\
  &=\beta_{A} \cdot (\bm{I}_{A}+\bm{I}_{B}) + \beta_{\BvA} \cdot \bm{I}_{B} + \beta_{C} \cdot \bm{I}_{C} + \epsilon.
\end{align} 
However, as we have seen in \cref{sec:method}, a post hoc test of $H_0: \beta_{\BvA}=0$ that does not account for the fact that it is tested if and only if $\hov$ was rejected will fail to
control the Type 1 error.

The machinery developed in Sections~\ref{sec:method}--\ref{sec:retrospective} for valid inference in Step 2 of F-screening applies directly in the setting of \gls{anova}. In fact, retrospective analysis, as detailed in Section~\ref{sec:retrospective}, is even more straightforward because in the context of \gls{anova}, it is not necessary to have access to $\FhoM$,
$\RR^2$, and $\RSE$ to conduct a retrospective analysis: instead, it suffices to know the number of samples within each group, and the mean and standard deviation of the outcome variable within each group. This is made precise in the following remark.

\begin{remark}\label{prop:retro}
Consider a one-way \gls{anova} with $K$ groups, parameterized analogously to \eqref{eq:plant2}. Suppose we have access to only the following quantities, for $k = 1,\ldots, K$:
\begin{enumerate}
    \item $n_k,$ the number of observations in the $k^{\text{th}}$ group,
    \item $\hat{\mu}_k$, the sample mean of $Y$ among observations in the $k^{\text{th}}$ group, and 
    \item $\hat{\sigma}_k = \sqrt{\frac{1}{n_k-1}\sum_{i\in \text{ group }k}(y_i-\hat\mu_k)^2}$, the sample standard deviation of $Y$ among observations in the $k^{\text{th}}$ group.
\end{enumerate} 
Then $\RR^2$ \eqref{eq:RSE_R2}, $\RSE$ \eqref{eq:RSE_R2}, $F_{H_0^{\kvk}}$ \eqref{eq:fstat}, and $F_{H_0^{1:K}}$ \eqref{eq:Foverall} take the form
\begin{align*}
&\RR^2 = \tfrac{\text{SSB}}{\text{SSB}+\text{SSW}} , \quad
\RSE = \sqrt{\tfrac{\text{SSW}}{n-K}},\quad F_{H_0^{\kvk}} = \tfrac{(n-K)(\hat{\mu}_{k_1} - \hat{\mu}_{k_2})^2}{\text{SSW} \left( \frac{1}{n_{k_1}} + \frac{1}{n_{k_2}} \right)}, \quad F_{H_0^{1:K}} = \tfrac{\text{SSB}/(K-1)}{\text{SSW}/(n-K)},
\end{align*}
where we define $n = \sum_{k=1}^K n_k$, $\bar{y} = \frac{1}{n} \sum_{k=1}^K \hat{\mu}_k n_k$, $\text{SSB} = \sum_{k=1}^{K} n_k (\hat{\mu}_k - \bar{y})^2$, and $\text{SSW} = \sum_{k=1}^{K} (n_k - 1) \hat{\sigma}_k^2$. Thus $p_{H_0^{1:K}} = 1 - F_{K-1,n-K}(F_{H_0^{1:K}})$, and $\psel$ can be computed via \eqref{eq:pselectiveretro}.

\end{remark}

\section{Simulation study}
\label{sec:simulation}

\subsection{Type 1 error}\label{subsec:t1control}

We generate data according to \eqref{eq:model} with $n=100$, $p=10$, $\sigma^2
=1$, $X_{ij} \iidsim \mathcal{N}(0,1)$, and
$\beta_1=\ldots=\beta_{10}=0$, so that $\hov$ in
\eqref{eq:hov} holds. We consider testing $\hoone: \beta_1=0$ if and only if $\hov$ is
rejected at level $\alphaov = 0.05$ using $\Fov$ \eqref{eq:Foverall}. For the subset of \num{100000} datasets for which $\hov$
is rejected, the left-hand
panel of \cref{fig:triple_plot} displays (i) $\phoM$ from \eqref{eq:pnaive}, which
does not account for rejection of $\hov$, and (ii) $\psel$ from \eqref{eq:pselchisq},
which accounts for rejection of $\hov$. We see that $\psel$ controls selective Type 1 error. Because $\phoM$ does not control the selective Type 1 error, we do not consider it in Section~\ref{subsec:power}.

\begin{figure}[h]
    \centering
    \includegraphics[width=\linewidth]{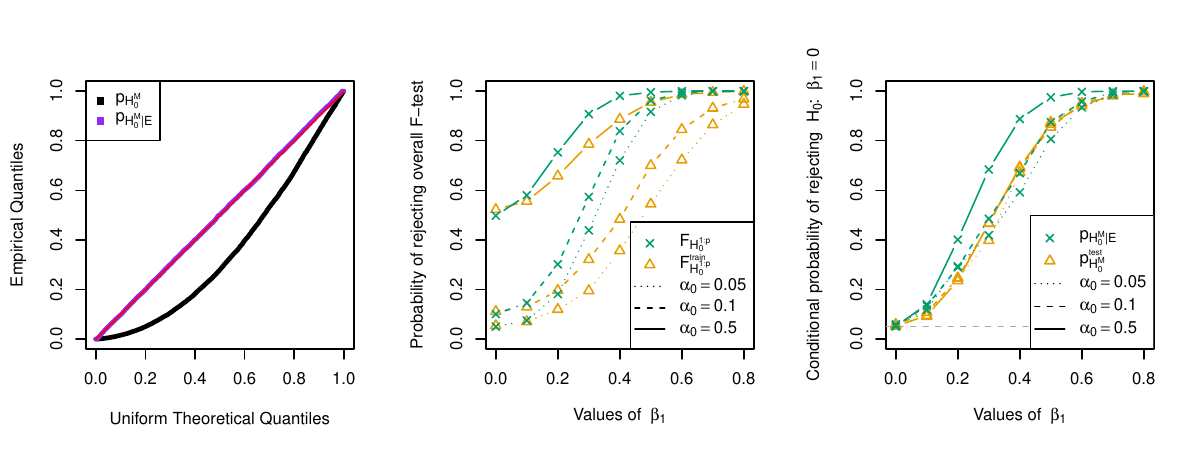}
    \caption{ \emph{Left:} For the simulation study detailed in
      Section~\ref{subsec:t1control}, we display a $\operatorname{Unif}(0,1)$
      quantile-quantile plot of p-values for the test of $\hoone:\beta_1 = 0$,
      conditional on the event that $\hov$ \eqref{eq:hov} is rejected. Here,
      $\hov$ (and hence $\hoone$) holds, and so a p-value for $\hoone:\beta_1=0$
      that controls the selective Type 1 error should follow a
      $\operatorname{Unif}(0,1)$ distribution. If the trace lies below the diagonal, the
      p-value is anti-conservative. For the subset of simulated datasets for
      which $\hov$ is rejected using $\Fov$, we display $\phoM$ and
      \eqref{eq:pnaive} $\psel$ \eqref{eq:pselchisq}.
      \emph{Middle:} For the simulation setting described in
      \cref{subsec:power}, we plot the power to reject $\hov$ at three different
      significance levels (\(\alphaov \in \left\{ 0.05, 0.1, 0.5 \right\}\)) with
      (i) the $F$-statistic based on all of the data, $F_{\hov}$
      \eqref{eq:Foverall}, and (ii) the $F$-statistic based on only the training
      set from sample splitting, $\Fovtrain$. \emph{Right:} For the simulation
      setting described in \cref{subsec:power}, we plot the conditional
      probability of rejecting $\hoone$ at significance level 0.05 (marked with a
      gray dashed line) given that $\hov$ was rejected at three levels of
      significance ($\alphaov\in \{0.05,0.1, 0.5\})$. When we reject $\hov$ with $\Fov$, we compute the
      p-value $\psel$ \eqref{eq:pselchisq}. When we reject $\hov$ with $\Fovtrain,$ we compute
      $\ptest$.
       \label{fig:triple_plot}}
\end{figure}

\subsection{Power}\label{subsec:power}

We compare the power of our proposed approach for F-screening to an approach that
relies on sample splitting \citep{cox1975NoteDataSplittingEvaluation}.

First, we formalize the sample splitting procedure. We randomly partition $n$
independent realizations of \eqref{eq:model} into two equally-sized training and test sets. We use the training set to test $\hov$ using an F-statistic (which we refer to as $\Fovtrain$), and, if and only if we reject \(\hov\), we use the test set to test
$\hoM$ using an F-statistic, which yields the p-value $\ptest$. Since the training and test sets are independent, $\ptest$ is
uniformly distributed under $\hoM$, regardless of the value of $\Fovtrain$.

Now we return to the simulation set-up. We generate data according to \eqref{eq:model} with $n=100$, $p=10$, $\sigma^2
=1$, $X_{ij} \iidsim \mathcal{N}(0,1)$, $\beta_2=\ldots=\beta_{10}=0$, and
with $\beta_1$ ranging from $0$ to $0.8$.

First, we test $\hov$ \eqref{eq:hov} at \(\alphaov \in \left\{ 0.05, 0.1, 0.5
\right\}\) using either (i) all the data, or (ii) just the training data. The
middle panel of \cref{fig:triple_plot} displays the proportion of simulated
datasets for which $\hov$ is rejected. Unsurprisingly, we have
substantially more power to reject $\hov$ using all the data versus using only
the training data.

Next, for the subset of simulated datasets for which we rejected $\hov$ in the
middle panel of \cref{fig:triple_plot}, we let $M=\{1\}$ and test the hypothesis $\hoone:
\beta_1=0$. Simulated datasets for which $\hov$ was rejected using all the data
are tested using $\psel$ \eqref{eq:pselchisq}, and those for which $\hov$ was rejected using the
training data are tested using $\ptest$. The right-hand panel of
\cref{fig:triple_plot} displays the power to reject $\hoone: \beta_1=0$,
conditional on having rejected $\hov$. When $\alphaov =0.05$ or $\alphaov = 0.1$, the power curves are similar for $\psel$
\eqref{eq:pselchisq} and $\ptest$, while the power curve is substantially higher for $\psel$ when $\alphaov =0.5 $. The power curves for $\ptest$ for the three
values of $\alphaov$ are nearly identical, since the training data used to test
$\hov$ is independent of the test data used to test $\hoone$, and $\alphaov$
appears only in the test of $\hov$. By contrast, the selective p-value $\psel$ shifts substantially according to the value of $\alphaov$: when
$\alphaov$ is small, the conditional probability of rejecting $H_0^1:\beta_1=0$
is smaller, since relatively little signal for $\beta_1$ is available for inference after
conditioning.

It is worth noting that the benefits of our selective p-value over the
sample splitting p-value extend beyond the empirical results seen in
\cref{fig:triple_plot}. Depending on how sample splitting is performed, the
testable hypotheses may differ between train and test sets (e.g., if rank
deficiency emerges from random partitioning). Also, it may be hard in an
applied setting to convince a collaborator to use only part of their data to
test $\hov$ and to reserve the remainder for testing $\hoM$. By contrast, our selective p-values enable the collaborator to
use \emph{all} of the data to test $\hov $, thereby yielding a result for the
test of $\hov$ that coincides with the standard (uncorrected) F-screening
procedure. Furthermore, the sample splitting p-values cannot be computed
retrospectively using summary statistics, as in Section~\ref{sec:retrospective}.

\subsection{Confidence intervals}\label{subsec:CI}

We now investigate coverage and width of the confidence intervals derived in Appendix~\ref{app:ci_point_estimates}.

To investigate confidence interval coverage, we generate data according to
\eqref{eq:model} with $\beta_1=\ldots=\beta_p=0$, $n=100$,
$p=5$, and $\sigma^2=1$ until we obtain \num{1000} datasets $(X,Y)$ for which we
reject $\hov$ \eqref{eq:hov} at level $\alphaov = 0.05$.
Then, for a fixed
confidence level (defined as \((1 - \alpha)\) for $\alpha\in (0,1)$), we
compute, for each of the \num{1000} datasets: (i) the standard confidence
interval for $\beta_1$ that does not account for selection, and (ii) the selective
confidence interval for $\beta_1$ \eqref{eq:CI}.

As shown in the left panel of \cref{fig:CI_globalnull}, the standard confidence
intervals for $\beta_1$ that do not account for rejection of $\hov$
\eqref{eq:hov} exhibit severe under-coverage for all confidence levels. In
contrast, the selective confidence interval
achieves nominal coverage across all confidence levels.

\begin{figure}[h]
    \centering
    \includegraphics[width = \textwidth]{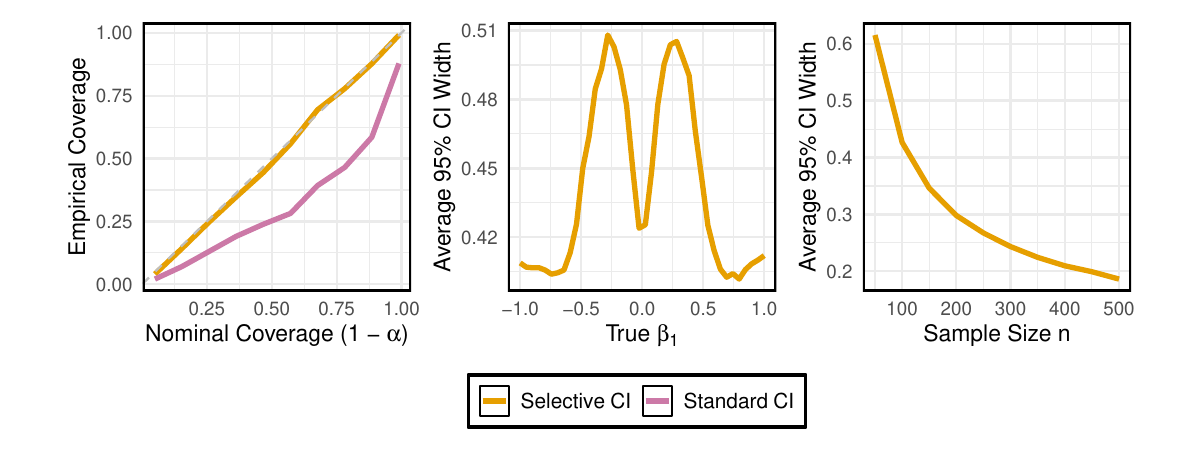}
    \caption{For the setting described in \cref{subsec:CI}, we construct
      confidence intervals for $\beta_1$ in \eqref{eq:model} according to
      Appendix~\ref{app:ci_point_estimates}, conditional on rejection of $\hov$ \eqref{eq:hov} with $\Fov$ \eqref{eq:Foverall}. \emph{Left:} We plot the coverage of the selective and standard confidence intervals for the datasets for which $\hov$ \eqref{eq:hov} is rejected, where the significance level $\alpha$ varies between 0 and 1, and $\beta_1 = \ldots = \beta_p=0$. Since the standard intervals do not achieve the nominal coverage, we do not display them in the middle and right panels.
    \emph{Middle:} We plot the mean width of the selective confidence intervals for the datasets for which $\hov$ \eqref{eq:hov} is rejected, where $\beta_1$ varies and $n=100$. \emph{Right:} We plot the mean width of the selective confidence intervals for which $\hov$ \eqref{eq:hov} is rejected for $\beta_1=0$ as $n$ varies.} 
    \label{fig:CI_globalnull}
\end{figure}

To examine confidence interval width, we first vary $\beta_1 \in \left[ -1, 1 \right]$ while holding $n = 100$ constant. We set $p = 5$, $\sigma^2 = 1$, and $\beta_2 = \dots = \beta_5 = 0$. For each fixed value of $\beta_1 \in \left[ -1, 1 \right]$, we generate data according to \eqref{eq:model} until we obtain 1000 datasets where $\hov$ is rejected. For each dataset, we compute the selective confidence interval. (We do not depict the standard confidence interval, as we showed in the left-hand panel of \cref{fig:CI_globalnull} that it is invalid.) Finally, we compute the average confidence interval width across the 1000 datasets and plot these values as a function of $\beta_1$ in the middle panel of \cref{fig:CI_globalnull}. In the right-hand panel, we repeat the analysis in the middle panel, but with $\beta_1 = 0$ and $n\in \{50,\ldots, 500\}.$

The curves associated with selective confidence intervals in the middle panel of
\cref{fig:CI_globalnull} exhibit a striking
shape: they are narrow both when \(\beta_1 \approx 0\) and when \(|\beta_1|\)
is large. This pattern can be understood by separately considering the cases
where (i) \(\left\lvert \beta_{1} \right\rvert\) is large, (ii) \(\left\lvert
  \beta_{1} \right\rvert\) is moderate, and (iii) \(\left\lvert \beta_{1}
\right\rvert\) is small. When \(\left\lvert \beta_{1} \right\rvert\) is large,
the event $E_1$ \eqref{eq:rejectFov} is highly likely, in which case conditional
inference roughly coincides with standard inference, and so the intervals are not appreciably wider than the standard (non-selective) intervals.
When \(\left\lvert \beta_{1} \right\rvert\) is moderate, the event $E_1$
\eqref{eq:rejectFov} is only moderately likely, and when it occurs, it often happens
that the data falls near the boundary of the selection region, and
little signal regarding \(\beta_{1}\) remains for inference, leading to wide
confidence intervals. Finally, when \(\left\lvert \beta_{1} \right\rvert\) is
small, the event $E_1$ \eqref{eq:rejectFov} is unlikely to occur; however, when
it does occur, selection is often driven by (chance) association between $X_2,\ldots,X_p$ and $Y$, in
which case sufficient signal about \(\beta_{1}\) remains for inference, leading
to moderately narrow confidence intervals.

\section{Data analysis} \label{sec:data}

In this section, we demonstrate a valid data analysis involving F-screening, both with and without access to the full data. First we re-analyze the data of \citet{xu2018DecreasingTrendBone}, and report selective p-values, confidence intervals, and point estimates for each coefficient. Then we conduct a retrospective analysis of \citet{keil2023LongitudinalSleepPatterns} and compute the selective p-values using only the standard \gls{anova} outputs because we do not have access to their data. In both analyses, we compare the selective quantities to their standard counterparts. 

\subsection{Prospective inference} 
\citet{xu2018DecreasingTrendBone} used the publicly-available \gls{nhanes} dataset from \citet{cdc2024about} to investigate possible changes in femoral neck \gls{bmd}\footnote{\gls{bmd} is coded as ``DXXNKBMD'' in the \gls{nhanes} dataset.} across the years 2005 to 2014 in US adults over the age of 30. In the \gls{nhanes} dataset, years are binned into 2005--2006, 2007--2008, 2009--2010, 2013--2014; no data from 2011--2012 is available for \gls{bmd}.  We restrict our analysis to the $n=7135$ males in the dataset, and consider the model 
\begin{equation}\label{eq:BMD}
    \text{BMD} =    \beta_{\text{AGE}}\cdot \text{AGE}  + \beta_{\text{5--6}} \cdot \bm{I}_{\text{5--6}} + \beta_{\text{7--8}} \cdot \bm{I}_{\text{7--8}}+\beta_{\text{9--10}} \cdot \bm{I}_{\text{9--10}}+\beta_{\text{13--14}} \cdot \bm{I}_{\text{13--14}} + \epsilon,
\end{equation} 
where $\bm{I}_{\text{5--6}}$, $\bm{I}_{\text{7--8}}$, $\bm{I}_{\text{9--10}}$, and $\bm{I}_{\text{13--14}}$ are indicator variables for the time periods 2005--2006, 2007--2008, 2009--2010, and 2013--2014, respectively. %
Note that \eqref{eq:BMD} contains age, as reported in the \gls{nhanes} dataset, as a covariate. For convenience, we regress out age: let $U \in \mathbb{R}^{7135 \times 7134}$ such that $UU^\top := I_{7135} - P_{\text{AGE}}$ and $U^\top U =I_{7135}$, and we consider the model 
\begin{equation}\label{eq:BMD-2}
  U^\top \text{BMD} = \beta_{\text{5--6}} \cdot U^\top  \bm{I}_{\text{5--6}} + \beta_{\text{7--8}} \cdot  U^\top  \bm{I}_{\text{7--8}}+\beta_{\text{9--10}} \cdot  U^\top  \bm{I}_{\text{9--10}}+\beta_{\text{13--14}} \cdot U^\top  \bm{I}_{\text{13--14}} + U^\top \epsilon.
\end{equation} 
In this setting, the overall null hypothesis is 
\begin{equation}\label{eq:BMD_omnibus}
\hov: \beta_{\text{5--6}} = \beta_{\text{7--8}} = \beta_{\text{9--10}} = \beta_{\text{13--14}} = 0.
\end{equation}

In line with \citet{xu2018DecreasingTrendBone}, we test whether the mean age-adjusted \gls{bmd} differs between time bins (05--06 versus 07--08, 05--06 versus 09--10, 05--06 versus 13--14, 07--08 versus 09--10, 07--08 versus 13--14, 09--10 versus 13--14). 
To compare mean age-adjusted \gls{bmd} in 2005--2006 to mean age-adjusted \gls{bmd} in 2007--2008, we reparametrize \eqref{eq:BMD-2} as 
\begin{align}\label{eq:BMD-3}
  U^\top \text{BMD}& =  \beta_{\text{5--6}} \cdot \left(  U^\top  \bm{I}_{\text{5--6}} + U^\top  \bm{I}_{\text{7--8}} \right) \nonumber\\
  & \hspace{5mm}+  \beta_{\text{7--8 vs. 5--6}} \cdot  U^\top  \bm{I}_{\text{7--8}}+\beta_{\text{9--10}}  \cdot U^\top  \bm{I}_{\text{9--10}}+\beta_{\text{13--14}} \cdot U^\top  \bm{I}_{\text{13--14}} + U^\top \epsilon,
\end{align} 
and test $H_0^{\text{ 7--8 vs. 5--6}}:\beta_{\text{7--8 vs. 5--6}}=0.$
The other five pairwise tests are conducted similarly. Results are shown in \cref{table:propsective}. In this case, accounting for the rejection of $\hov$ \eqref{eq:BMD_omnibus} has virtually no impact---most of the selective and standard quantities are identical up three decimals.

\begin{table}[ht]
\centering
\setlength{\tabcolsep}{4pt}
\begin{tabular}{|l|ccc|ccc|}
\hline
 & \multicolumn{3}{c|}{Standard} & \multicolumn{3}{c|}{Selective} \\
\hline
Years & Estimate & 95\% CI & p-value & Estimate & 95\% CI & p-value \\ 
\hline
05--06 vs. 07--08 & -0.004 & (-0.013,  0.005) & 0.423 & -0.004  & (-0.013, 0.005)  & 0.426 \\
\hline
05--06 vs. 09--10 & -0.003 & (-0.012,  0.006) & 0.473 & -0.003  & (-0.012, 0.006)  & 0.468 \\
\hline
05--06 vs. 13--14 & -0.020 & (-0.029, -0.010) & 0.000 & -0.020  & (-0.029, -0.010)  & 0.000 \\
\hline
07--08 vs. 09--10 &  0.000 & (-0.008, 0.009) & 0.923 & 0.000 & (-0.008, 0.009) & 0.926 \\
\hline
07--08 vs. 13--14 & -0.016 & (-0.025, -0.007) & 0.000 & -0.016  & (-0.025, -0.007) & 0.001 \\
\hline
09--10 vs. 13--14 & -0.017 & (-0.025, -0.008) & 0.000 & -0.017  & (-0.025, -0.008) & 0.000 \\
\hline
\end{tabular}
\caption{For \gls{nhanes} \gls{bmd} data, we perform F-screening for the model in \eqref{eq:BMD}, and compare selective point estimates, confidence intervals, and p-values with standard counterparts following the rejection of the overall $F$-test of \eqref{eq:BMD_omnibus}. The statistics in the ``Standard'' column do \emph{not} account for rejection of \eqref{eq:BMD_omnibus} (see Box~\ref{box:box1}), whereas statistics in the ``Selective'' column do account for it. \label{table:propsective}}
\end{table}

\subsection{Retrospective inference}\label{subsec:keil_retro}

\citet{keil2023LongitudinalSleepPatterns} explores the relationship between sleep and cognitive impairment in older adults. The authors considered four response variables: (i) \gls{cesd}, (ii) sleep duration, (iii) \gls{mmse}, and (iv) \gls{mdrs} (see references in \citet{keil2023LongitudinalSleepPatterns} for details). For each response variable corresponding to the rows in \cref{fig:table1}, \citet{keil2023LongitudinalSleepPatterns} fit a model of the form 
$Y = \beta_{1}\cdot  \bm{I}_{<65}  + \beta_{2} \cdot \bm{I}_{\text{65--85}}  + \beta_{3}\cdot  \bm{I}_{>85} + \epsilon, $
where $\bm{I}_{<65}$, $\bm{I}_{\text{65--85}}$, and $\bm{I}_{>85}$ denote indicator variables for the age group (in years) of $n=826$ study participants.
The second-to-last column of \cref{fig:table1} displays the p-values for a test of $\hov: \beta_1=\beta_2=\beta_3=0$, and the last column reports p-values for standard tests of $\hoj: \beta_j=0$ for $j=1,2,3$.

\citet{keil2023LongitudinalSleepPatterns} are not explicit about whether they would have tested $\hoj$, $j=1,2,3$ had they not rejected $\hov$, and it is not possible to ascertain their intentions from their results, since in their Table 1 (a subset of which is displayed in our \cref{fig:table1}), all hypotheses of the form $\hov$ tested with ANOVA are rejected. Notably, they refer to  tests of $\hoj$ as ``post hoc tests''; this phrasing is typically reserved in the applied literature to refer to tests that are conducted \emph{after} rejecting  $\hov$ \citep{everitt2010CambridgeDictionaryStatistics}. Thus, for the sake of demonstration we shall assume that they  would not have tested $\hoj$ for $j=1,2,3$ in their Table 1 if $\hov$ had not been rejected.

Accordingly, we apply the methods in Section~\ref{sec:retrospective} and \ref{sec:anova} to conduct a retrospective analysis, using only the summary statistics reported in \cref{fig:table1}. For each of the four outcomes in \cref{fig:table1}, we compute $\psel$ in \eqref{eq:pselchisq} for $\hoj: \beta_j=0$, for $j=1,2,3$, alongside the standard p-value in \eqref{eq:pnaive} which does not account for the fact that we would not have tested $\hoj$ if $\hov$ had not be rejected. The standard p-values that we compute match those in \cref{fig:table1} if the {\v S}id{\'a}k correction is applied \citep{sidak1967RectangularConfidenceRegions}.

The results are displayed in Table~\ref{table:retro}. Generally, the values of $\psel$ \eqref{eq:pselchisq} are equal to or larger than the corresponding values of $\phoM$ \eqref{eq:pnaive}. For example, the p-value for the mean difference in \gls{cesd} between the under-65 age group and the 65-85 age group is 0.003 using $\phoM$ versus 0.118 using $\psel$, and the p-value for the mean difference in \gls{cesd} between the 65-85 age group and over-85 age group is 0.016 versus 0.100. At the significance level of 0.05 used by \citet{keil2023LongitudinalSleepPatterns}, the selective p-value $\psel$ does not lead to rejection for those two post hoc tests, whereas the standard p-value $\phoM$ does.

\begin{table}
\centering 
\begin{tabular}{!{\vrule width 1.5pt}c!{\vrule width 1.5pt}c|c!{\vrule width 1.5pt}c|c!{\vrule width 1.5pt}c|c!{\vrule width 1.5pt}c|c!{\vrule width 1.5pt}}
\hline
 & \multicolumn{2}{c!{\vrule width 1.5pt}}{\gls{cesd}} & \multicolumn{2}{c!{\vrule width 1.5pt}}{Sleep} & \multicolumn{2}{c!{\vrule width 1.5pt}}{\gls{mmse}} & \multicolumn{2}{c!{\vrule width 1.5pt}}{\gls{mdrs}} \\ 
\hline
  & $\phoM$ & $\psel$ & $\phoM$ & $\psel$ & $\phoM$ & $\psel$ & $\phoM$ & $\psel$ \\ 
\hline
$<65$ vs. 65-85 & 0.003 & 0.118 & 0.001 & 0.001 & 0.000	 & 0.000 & 0.003	 & 0.003 \\
\hline
$<65$ vs. $\geq 85$ & 0.397	 & 0.397 & 0.000	 & 0.003 & 0.000	 & 0.000 & 0.000	 & 0.000 \\
\hline
65-85 vs. $\geq85$ & 0.016 & 0.100 & 0.060	 & 0.061 & 0.000	 & 0.000 & 0.000		& 0.000 \\
\hline
\end{tabular}
\caption{For the retrospective analysis conducted on the linear regression outputs from \citet{keil2023LongitudinalSleepPatterns} (\cref{fig:table1}), we display the selective ($\psel$ in \eqref{eq:pselchisq}) and standard ($\phoM$ in \eqref{eq:pnaive}) p-values from four analyses, each involving a separate outcome variable (\gls{cesd}, Sleep, \gls{mmse}, \gls{mdrs}). For each of these outcome variables, we replicate the one-way \gls{anova} and post hoc tests conducted in \citet{keil2023LongitudinalSleepPatterns}, which are reported in the last two columns of the table in \cref{fig:table1}. For the post hoc tests, we report $\phoM$ \eqref{eq:pnaive} and $\psel$ \eqref{eq:pselchisq}. If we were to apply the {\v S}id{\'a}k correction \citep{sidak1967RectangularConfidenceRegions} to the values of $\phoM$ \eqref{eq:pnaive} in this table, they would match the p-values in the last column of \cref{fig:table1}. \label{table:retro}}
\end{table}

\section{Discussion}
\label{sec:discussion}
We have shown that a common practice in applied statistics---reporting the outputs of a linear regression model if and only if an overall $F$-test is rejected---leads to invalid inference. We have proposed an easily-implemented approach, based on conditional selective inference, that remedies the problem. We present a toolbox for valid inference in this setting that contains the most critical linear model outputs: p-values, confidence intervals, and point estimates. Furthermore, our selective p-values can be computed even without access to the raw data, using only the standard outputs of a multiple linear regression.

Some natural questions that we leave for future work are as follows:
\begin{itemize}
\item Our results relied on an assumption that the errors are Gaussian. Can this be relaxed using ideas from e.g.~\citet{tian2017AsymptoticsSelectiveInference,panigrahi2023CarvingModelfreeInference}? (We note, of course, that standard inference for a linear model also assumes Gaussianity in finite samples---thus, reliance on Gaussianity is only a very minor limitation of our method.)
\item Can these results be extended beyond least squares linear regression, e.g.~to generalized linear models, using ideas from \citet{tian2018SelectiveInferenceRandomized}?
\item How should our selective p-values be adjusted to account for multiple comparisons? For instance, how can $\psel$ for $M  \in \{M_1, \ldots, M_L\}$ be corrected in order to control family-wise selective Type 1 error rate? An extension of Scheff\'{e}'s correction \citep{scheffe1953MethodJudgingAll}  to the selective setting may be possible. 
\end{itemize}

An \verb=R= package called \verb=lmFScreen=, implementing the methods proposed in this paper, is available at %
\url{https://github.com/mcgougho/lmFScreen}.
\verb=R= scripts to reproduce all numerical 
results can be found at %
\url{https://github.com/mcgougho/lmFScreen-paper}. 

\section*{Acknowledgements}
The authors gratefully acknowledge funding from NSF DMS 2322920, NSF DMS 2514344,  NIH 5P30DA048736, and  ONR N00014-23-1-2589 to DW; NSF DMS-2303371, the Pacific Institute for the Mathematical Sciences, the Simons Foundation, and the eScience Institute at the University of Washington to DK; and NSF Graduate Research Fellowship Grant No. DGE-2140004 to OM.

\clearpage
\appendix

\makeatletter
\def\theHfigure{\thefigure}
\def\theHequation{\theequation}

\section{Connections to multiple testing}
\label{app:multipletesting}
\renewcommand{\thefigure}{A.\arabic{figure}}
\renewcommand{\theequation}{A.\arabic{equation}}
\setcounter{figure}{0}
\setcounter{equation}{0}

\subsection{Family-wise error rate control does not imply selective Type 1 error rate control}

One might wonder whether the problem addressed in this paper could be
more simply tackled by controlling the family-wise error rate (FWER).  FWER control has been studied for almost a century  \citep{tukey1949ComparingIndividualMeans, scheffe1953MethodJudgingAll, dunn1961MultipleComparisonsMeans, holm1979Simple,marcus1976ClosedTestingProcedures,hochberg1987MultipleComparisonProcedures}, and continues to be an area of active research: for instance, the  ``PoSI" proposal of \citet{berk2013ValidPostselectionInference}---foundational to the field of selective inference---targets the FWER.

However, we see in the right-hand panel of \cref{fig:qqplots-appendix} that, in the context of F-screening in a setting with $p=3$, control of the FWER does not imply control of the selective Type 1 error rate. In particular, though Bonferroni, Scheff\'{e}, and closed testing corrections are guaranteed to control FWER, they do not control the selective Type 1 error rate \citep{fithian2017OptimalInferenceModel}, i.e.~the probability of a Type 1 error, conditional on rejection of $H_0^{1:3}$.

\begin{figure}[h!]
    \begin{tabular}{cc}
    \hspace{8.5em}\textbf{All realizations}& \hspace{4.5em} \textbf{\makecell{Only when  $\hov$ is rejected} } 
    \end{tabular} \\
  \includegraphics[width = \textwidth]{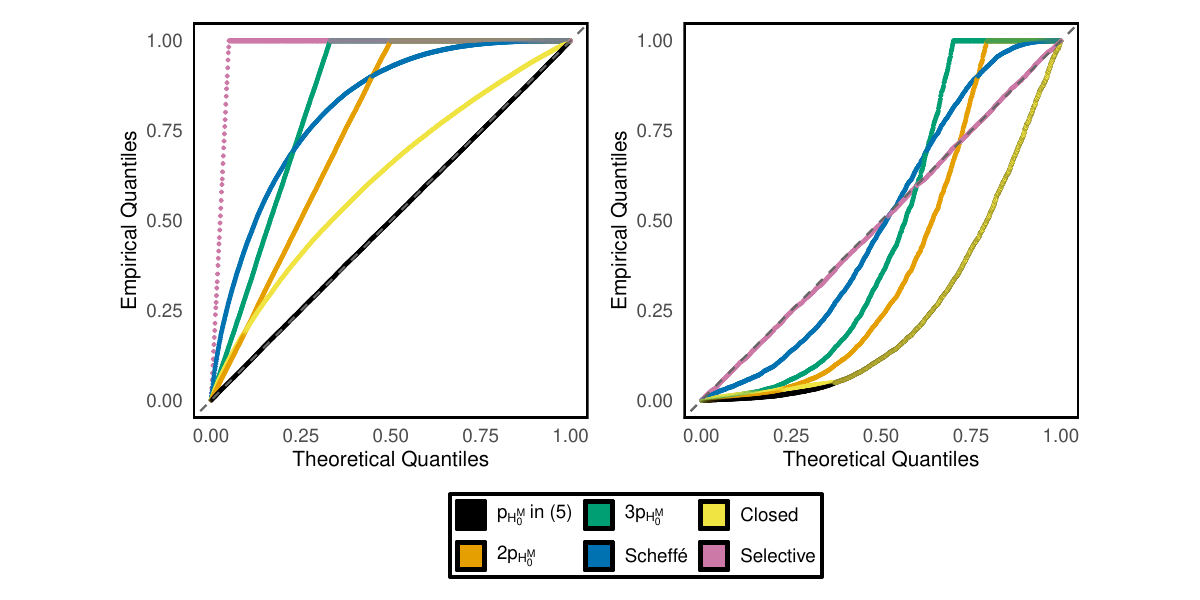}
  \caption{We simulate data according to \eqref{eq:model} with $n=100$, $p=3$, $\sigma^2=1$, and $\beta_1=\beta_2=\beta_3=0$, so that $\hov$ \eqref{eq:hov} holds. For the test of $\hoone: \beta_1=0$, we report (i) the uncorrected p-value that does not account for F-screening, i.e. $\phoM$ in \eqref{eq:pnaive}; (ii) $2\phoM$, resulting from a Bonferroni correction to $\phoM$ that accounts for the fact that two tests ($\hov$ and $\hoone$) were conducted; (iii) $3\phoM$, a Bonferroni correction that accounts for the fact that three post hoc tests ($\hoj: \beta_j=0$ for $j \in \left\{ 1, 2, 3 \right\}$) could have been conducted; (iv) the p-value that results from using Scheff\'{e}'s method; (v) the p-value associated with the closed testing procedure as described in \cref{subsec:closedtest}; and (vi) 
  the selective p-value \eqref{eq:pselchisq} that accounts for F-screening. 
  \emph{Left:} For each of 100,000 simulated datasets, the QQ-plot displays the p-values for $\hoone$. The uncorrected p-value $\phoM$ follows a $\operatorname{Unif}(0,1)$ distribution, while the selective p-value is extremely conservative because results for all simulated datasets are displayed, regardless of rejection of $\hov$. Since we are only considering a single test, the Bonferroni-corrected, Scheff\'{e}-corrected, and closed-testing p-values are also conservative.  \emph{Right:} For the 5\% of the 100,000 simulated datasets for which $\hov$ is rejected at level $\alphaov=0.05$, we display the p-values for $\hoone$. The uncorrected p-value $\phoM$ and the closed testing p-value are anti-conservative, as are the Bonferroni-corrected and Scheff\'{e}-corrected p-values when small quantiles are considered. However, the selective p-value follows a $\operatorname{Unif}(0,1)$  distribution. }
  \label{fig:qqplots-appendix}
\end{figure}

Why is this the case? 
Consider the setting of \cref{fig:qqplots-appendix}. There are three possibilities associated with the  null hypotheses $H_0^{1:3}: \beta_{1:3}=0$ and $\hoone: \beta_1=0$: (i) both are true; (ii) $H_0^{1:3}$ is false, but
\(\hoone\) is true; or (iii)
both are false. In (iii), a Type 1 error is
impossible, so we consider only (i) and (ii). Let \(\delta_{1:3}\) be the
decision function associated with the test of $H_0^{1:3}$: it is a function of the
(random) data and takes value \(1\) when we reject $H_0^{1:3}$ and value \(0\)
otherwise; let \(\delta_{1}\) similarly denote the decision function associated
with testing \(\hoone\). In scenario (i), the FWER is \(\E_{H_0^{1:3}} \left[ \max\{\delta_{1:3}, \delta_{1} \}\right]\).
In scenario (ii), the FWER is \(\E_{\hoone} \left[ \delta_{1} \right]\).
By contrast, in  scenario (i)  the selective Type 1 error---which is controlled by our procedure---is \(\E_{H_0^{1:3}} \left[ \delta_{1} \mid \delta_{1:3}
= 1\right]\), while in scenario (ii) it is \(\E_{\hoone} \left[ \delta_{1} \mid \delta_{1:3}
= 1\right]\). Clearly, the former targets are not the same as the latter targets, and control of one does not imply control of the other.

\subsection{Connection to closed testing} \label{subsec:closedtest}

At first glance, there appears to be a connection between the F-screening procedure  described in Box~\ref{box:box1}, and so-called  ``closed testing procedures" \citep{marcus1976ClosedTestingProcedures}. %

We begin by describing the
closed testing procedure of \cite{marcus1976ClosedTestingProcedures}, which controls the FWER.  Let \(H_{1}, H_{2}, \dots, H_{q}\) denote the set of
\emph{elementary} hypotheses to be tested. In addition to the elementary
hypotheses, we also consider all possible intersections \(H_{I} = \cap_{i \in I}
H_{i}\) for all nonempty \(I \subseteq \left\{ 1, 2, \dots, q \right\}\). Next,
we record the results of a (preliminary) test at level \(\alpha\) for each of the hypotheses
(including the elementary \emph{and} intersection hypotheses). We then consider an elementary hypothesis $H_i$ to be rejected if and only if
the preliminary test $H_I$ is rejected for \emph{all} \(H_{I}\) where \(i \in I\).

How can we apply closed testing in our setting? First, consider elementary
hypotheses \(\hoM : \beta_{M} = 0\) and \(\honM : \beta_{-M} = 0\). In order to conduct a closed testing procedure, the only
additional hypothesis we need to consider is \(\hov = \hoM \cap \honM\). In
order to reject \(\hoM\), closed testing dictates that our preliminary tests must have rejected
\(\hov\) \emph{and} \(\hoM\); of course, we can conduct the preliminary test of the former before we conduct the preliminary test of the latter. 
Thus, in this setting, closed testing bears a resemblance to the procedure in Box~\ref{box:box1}. The p-value associated with the test of $\hoM$ in this particular instance of closed-testing is the maximum of the two p-values associated with the preliminary tests of $\hov$ and $\hoM$. In \cref{fig:qqplots-appendix}, we plot this p-value in the case that $p=3$ and the test of interest in Step 2 of Box~\ref{box:box1} is $\hoone:\beta_1=0$.  We see that the closed testing p-value does not control the selective Type 1 error.

How can we reconcile the FWER guarantees of closed testing with the
empirical results depicted in the right-hand panel of Figure~\ref{fig:qqplots-appendix}? %
The issue hinges on how we ``count'' the
decision regarding 
\(\hoM\) in scenarios wherein we fail to reject \(\hov\). We illustrate this point with a simple example, in which $\hov$ holds. Suppose that one hundred researchers
independently collect data and then  perform closed testing at level $\alpha$.  This entails conducting preliminary tests of  $\hov$, $\hoM$, and $\honM$. Suppose that 
five of the researchers' preliminary tests reject \(\hov\); of these five, four find that their preliminary tests also
reject \(\hoM\). Then, the closed testing procedure dictates that these four researchers should reject $\hoM$. In this setting, the family-wise error proportion is $5/100 = 5\%$, because five of the 100 researchers falsely rejected at least one null hypothesis: this is consistent with the fact that closed testing guarantees FWER control. In contrast, the resulting selective Type 1 error proportion---defined as the proportion of researchers who rejected $\hoM$ among those who rejected $\hov$---is $4/5 = 80\%$; this is because closed testing does not control the selective Type 1 error rate. %

While other work on closed testing has built upon the proposal of \cite{marcus1976ClosedTestingProcedures} and has considered other targets besides FWER \citep{goeman2011Multiple}, this literature does not target the selective Type 1 error rate that is the focus of our paper.

\section{Comparison to \citet{heller2019PostSelectionEstimationTesting}} \label{app:heller}
\renewcommand{\thefigure}{B.\arabic{figure}}
\renewcommand{\theequation}{B.\arabic{equation}}
\setcounter{figure}{0}
\setcounter{equation}{0}

In this appendix we compare our approach to the method proposed in  \cite{heller2019PostSelectionEstimationTesting}. We begin with a summary of their method.

For a parameter vector $\beta$ and estimator $\hat{\beta}\sim \mathcal{N}_m(\beta,\Sigma) $, their proposal is as follows:
\begin{enumerate} 
\item Conduct an aggregate test of the form $S = \hat{\beta}^\top K \hat{\beta}> S_{1-t_1}$, where $K$ is a positive semi-definite matrix, and $S_{1-t_1}$ is the $(1-t_1)$ quantile of the distribution of $S.$ 
\item Only if the test in Step 1 is rejected, perform inference on $\beta$ conditional on this rejection. 
\end{enumerate}
Their Steps 1 and 2  assume that the covariance matrix $\Sigma$ is known. Their Theorem 1 establishes that, for any linear contrast vector $\eta$, the statistic $\eta\hat{\beta}$ follows a  truncated normal distribution conditional on $\{S>S_{1-t_1}, \vec{W}\}$, where $\vec{W} = (I_m - (\eta^\top \Sigma \eta)^{-1} \Sigma \eta\eta^\top )\hat\beta$. This result is then used to construct a selective p-value, which controls the selective Type 1 error (i.e., the Type 1 error conditional on rejection in Step 1).

If we attempt to apply \cite{heller2019PostSelectionEstimationTesting}'s proposal to the  F-screening setting in our paper (Box~\ref{box:box1}), two challenges arise. 
\begin{list}{}{}
\item{Challenge 1.} Step 1 of their procedure involves a $\chi^2$-test, not an $F$-test. Does the distinction between F-screening and $\chi^2$-screening in Step 1 matter?
There is asymptotic equivalence between an $F$-test for $\beta=0$ and a $\chi^2$-test for $\beta=0$ as $n\rightarrow \infty$, if a consistent estimator for $\sigma$ is used in the $\chi^2$-test. However, in finite samples the two tests will disagree. 

\item{Challenge 2.} Their main theoretical result, Theorem 1, requires that the  covariance be known. However, in the linear model setting of our paper, $\sigma$ is unknown. Their paper provides no guarantees about the validity of their procedure when the covariance is estimated.

\end{list}

Our proposal avoids Challenge \#1 because it considers the \emph{exact} screening step used by data analysts in practice: i.e.~Step 1 of Box~\ref{box:box1} really involves an $F$-test, rather than a $\chi^2$ approximation to an $F$-test. 

Our method also avoids Challenge \#2 because we have no need to estimate $\sigma^2$. We achieve this by following the arguments in Section 4.1 of  \citet{fithian2017OptimalInferenceModel},  and conditioning on $\|Y\|^2$, which is a sufficient statistic for $\sigma^2$ under the model  $Y\sim \mathcal{N}(X\beta,\sigma^2 I_{n-1})$. 

To show that Challenges \#1 and \#2 faced by \citet{heller2019PostSelectionEstimationTesting} are real problems in practice, we make the following points in \cref{fig:heller}:
\begin{enumerate}
\item  \emph{The datasets that pass Step 1 of F-screening differ from those that pass the $\chi^2$-screening rule proposed in \citet{heller2019PostSelectionEstimationTesting}} (Challenge \#1). See the left-hand panel of \cref{fig:heller}. Whether we use the conservative estimate of $\sigma^2$ 
\begin{equation}\label{eq:sigma_tilde}
    \tilde\sigma^2 = \frac{(Y-\frac{1}{n}\sum_i Y_i)^\top(Y-\frac{1}{n}\sum_i Y_i)}{n-1},
\end{equation}
denoted \(\hat{\sigma}^2_{\text{null}}\) in \citet{heller2019PostSelectionEstimationTesting}, or use the true value of $\sigma$, the F-screening and \(\chi^2\)-screening events differ drastically. On the other hand, when plugging-in the standard estimate of $\sigma^2$,
\begin{equation}\label{eq:sigma_hat}
    \hat\sigma^2 = \frac{(Y-X\hat\beta)^\top(Y-X\hat\beta)  }{n-p-1},
\end{equation}
the $\chi^2$-screening event more closely aligns with our screening event, though there are still discrepancies: $5.2\%$ of the datasets were rejected by both the $F$-test and $\chi^2$-test with the standard plug-in, while $7.2\%$ of the datasets were rejected by the $\chi^2$-test with the standard plug-in and \emph{not} rejected by the $F$-test.

\item  \emph{When $n$ is small, the proposal of \citet{heller2019PostSelectionEstimationTesting} that uses the standard plug-in estimator of $\sigma$, $\hat\sigma^2,$ does not control the  selective Type 1 error in Step 2 of F-screening} (Challenge \#2). By contrast, our proposed selective p-value does. This is shown in the right-hand panel of \cref{fig:heller}. If instead the oracle $\sigma^2$ or the conservative estimate $\tilde\sigma^2$ \eqref{eq:sigma_tilde} is used to compute the selective p-value from \citet{heller2019PostSelectionEstimationTesting}, then those p-values become very conservative. 

We note that \citet[Section 4.3.1]{heller2019PostSelectionEstimationTesting} introduce a regime-switching variance estimator, which selects between \eqref{eq:sigma_hat} and  \eqref{eq:sigma_tilde} based on a more stringent aggregate test threshold. In the setting of \cref{fig:heller}, this threshold is rarely exceeded, so the resulting p-values very closely match with those obtained using the conservative estimator. For this reason, we omit these from \cref{fig:heller}.
\end{enumerate}

\begin{figure}[]
\begin{tabular}{ccc}
\hspace{5em}p-values for $\hov$
& \hspace{5em}
p-values for $\hoone: \beta_1=0$
\end{tabular}
  \includegraphics[width = \textwidth, trim=0 0 0 0, clip]{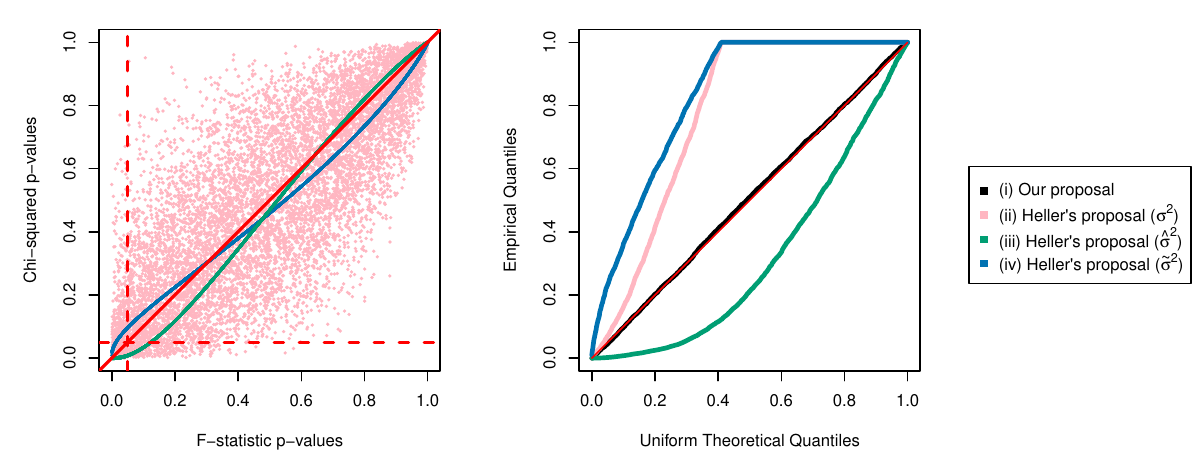}
  \caption{\textit{Left}: We generate \num{1e4} datasets $(X,Y)$ with $n=30$, $p = 10$, $\beta_1=\ldots = \beta_p = 0$, and $\sigma^2=1$. For each dataset, we test the  overall null hypothesis $H_0^{1:p}:\beta_1 = \ldots =\beta_p = 0$ at level 0.05 with (i) an $F$-test; (ii) a $\chi^2$-test that uses the true value of $\sigma^2$; (iii) a $\chi^2$-test that uses the standard plug-in estimate $\hat\sigma^2$ in \eqref{eq:sigma_hat}; and (iv) a $\chi^2$-test that uses the conservative plug-in estimate $\tilde\sigma^2$ in \eqref{eq:sigma_tilde}. The $x$-axis corresponds to the $F$-test p-value in (i), whereas the $y$-axis corresponds to the $\chi^2$-test p-values in (ii)-(iv). The red dashed horizontal and vertical lines indicate the $0.05$ cut-off for the F-statistic and $\chi^2$-statistic p-values, and the red solid line indicates the 45 degree line. To understand why the $\chi^2$ p-values in (iii) and (iv) are a monotone function of the F-statistic p-values in (i), note that the test statistics for (iii) after scaling by the estimated variance, (iv) after scaling by the estimated variance, and (i) are proportional to $\frac{Y^\top P_X Y}{Y^\top (I_{n-1}-P_X) Y} =: F$,  $\frac{Y^\top P_X Y}{Y^\top  Y}=\frac{F}{1+F}$, and $F$, respectively. By contrast, the $\chi^2$ p-values in (ii) are not a monotone function of the F-statistic p-values in (i) because they are scaled by the true variance, not an estimate of the variance. 
  \emph{Right:} We generate datasets $(X,Y)$ with $n=30$, $p = 10$, $\beta_1=\ldots = \beta_p = 0$, and $\sigma^2=1$ until we have \num{5e3} for which $\hov:\beta_1=\ldots=\beta_p=0$  is rejected with an $F$-test at $\alpha=0.05$. For each such dataset, for a test of $\hoone: \beta_1=0$, we compute (i) the selective p-value in \eqref{eq:pselchisq}, as well as the selective p-value based off of the truncated normal distribution in \citet{heller2019PostSelectionEstimationTesting} that uses (ii) the true value of $\sigma$, (iii) the standard plug-in estimate $\hat\sigma^2$ in \eqref{eq:sigma_hat}, and (iv) the conservative plug-in estimate $\tilde\sigma^2$ in \eqref{eq:sigma_tilde}. A uniform QQ-plot of these p-values is displayed. Our proposal's p-values are exactly uniform under the null hypothesis, whereas those of 
\citet{heller2019PostSelectionEstimationTesting} are anti-conservative when the standard variance estimate $\hat\sigma^2$ is used, and conservative when the true value of $\sigma^2$ or the conservative estimate $\tilde\sigma^2$ is used. The $\chi^2$ p-values obtained using $\hat\sigma^2$ are anti-conservative because conditional on rejection of $\hov$, the standard estimator $\hat\sigma^2$ is biased downwards; furthermore, the truncated normal approximation used to compute the $\chi^2$ p-values does not hold exactly. By contrast, the $\chi^2$ p-values obtained using the true value of $\sigma^2$ and the conservative estimate $\tilde\sigma^2$ are conservative because of instances where Step 1 in Box~\ref{box:box1} is rejected but the $\chi^2$-test would not have been (this corresponds to instances where the blue and pink points have an empirical quantile value of 1 in the right-hand panel). }
  \label{fig:heller}
\end{figure}

In conclusion, our contributions relative to \citet{heller2019PostSelectionEstimationTesting} are as follows:
\begin{itemize}
\item whereas \citet{heller2019PostSelectionEstimationTesting} would have analysts perform Step 1 with a $\chi^2$-test, we conduct Step 1 with an $F$-test, aligning with standard practice \citep{fisher1935DesignExperiments,lindquist1940StatisticalAnalysisEducational};
\item whereas \citet{heller2019PostSelectionEstimationTesting} provide no theoretical guarantees in the setting of unknown variance (i.e., the setting of a real-world data analysis), our method provides selective Type 1 error control and confidence interval coverage in finite samples;
\item whereas computation of the p-value in \citet{heller2019PostSelectionEstimationTesting} requires access to the complete data $(X,y)$, our p-value can be computed retrospectively per \cref{sec:retrospective}.
\end{itemize}

\section{Selective confidence intervals and point estimates} \label{app:ci_point_estimates}
\renewcommand{\thefigure}{C.\arabic{figure}}
\renewcommand{\theequation}{C.\arabic{equation}}
\renewcommand{\theremark}{C.\arabic{remark}}
\setcounter{figure}{0}
\setcounter{equation}{0}
\setcounter{remark}{0}

In this appendix, we derive confidence intervals and point estimates for the regression coefficients that are valid conditional on rejecting $\hov$ \eqref{eq:hov}.

\subsection{Point estimates} 
\label{sec:point-estimates}
In this section we will develop an approach to find the joint \emph{conditional}
maximum likelihood estimate (MLE) (conditional on rejecting \(\hov\)) of
\(\beta_{j}\) and $\sigma^2$ under \eqref{eq:model}. First, we will arrive at a convenient
expression for the conditional density. Once that is obtained, we will maximize the
associated conditional likelihood with respect to \(\beta_{j}\) and $\sigma^2$.

Recall the definition of $E_1$ in \eqref{eq:rejectFov}. Unfortunately, we encounter an immediate challenge in finding the conditional
density of \(Y\) with respect to \(\beta_{j}\) and $\sigma^2$ given \(Y\in E_{1}\): the
probability of the event $E_1$ depends on nuisance parameters $\beta_{-j}$
through the projection $P_{X_{-j}}Y$. To eliminate this dependence, we
additionally condition on the event $\tilde{E}_2(y) =\{Y : P_{X_{-j}}Y =
P_{X_{-j}}y\}$ (where \({y}\) is a realization of \(Y\)); this
fixes the component of $Y$ that lies in the column space of $X_{-j}$. 

\begin{remark}
In \cref{subsec:method}, we conditioned not only on $E_1$ in \eqref{eq:rejectFov} and $E_2(y)$ in \eqref{eq:E2E3}, but also on $E_3(y) $ in \eqref{eq:E2E3}. This is because (as discussed in \cref{rmk:conditioning}) the natural parameters are $\tfrac{\beta_1}{\sigma^2}, \ldots, \tfrac{\beta_p}{\sigma^2}, \tfrac{1}{2\sigma^2}$ with sufficient statistics $X_1^\top Y, \ldots, X_p^\top Y, \|Y\|$. Thus, conditioning on $E_1$, $X_{-j}^\top Y,$ and $\|Y\|$ leaves only $\beta_j/\sigma^2$ in the conditional distribution. Since a test of $H_0:\beta_j/\sigma^2 = 0$ is equivalent to a test of $H_0:\beta_j = 0$ for any $\sigma>0$, we were able to construct the test of interest.

In the context of point estimation, the story is a bit different. Conditioning on $E_1$, $X_{-j}^\top Y,$ and $\|Y\|$ again yields a conditional distribution that involves only $\beta_j/\sigma^2$. However, this enables estimation only of $\beta_j/\sigma^2$, not of $\beta_j$. 

Therefore, in this section, we condition on only $E_1$ and $X_{-j}^\top Y,$ not on $\|Y\|$. This yields a conditional distribution involving both $\beta_j$ and $\sigma^2$. We jointly maximize it with respect to both parameters.
\end{remark}

\begin{proposition}
  \label{prop:cond-dens}
  Let \(Y \sim \mathcal{N}_{n - 1} (X \beta, \sigma^2 I_{n - 1}) \), and define
  \begin{equation*}
  D =
  \begin{bmatrix}
    I_{n-1}-P_X\\
    P_X-P_{X_{-j}}\\
    P_{X_{-j}}
  \end{bmatrix} =
  \begin{bmatrix}
    D_1 \\
    D_2 \\
    D_3
  \end{bmatrix}.
\end{equation*}
Define new random variables \(\tZ = D_1 Y\), \(\tW = D_2 Y\), and
\(\tV = D_{3} Y\) (note that $\tZ^\top \tZ /\sigma^2
\overset{d}{=} Z$ and $\tW^\top \tW / \sigma^2 \overset{d}{=} W$,
for $Z$ and $W$ that appear in \cref{prop:chisq}). Let \(c\) be as defined in
\eqref{eq:acdr} and \(E_{1}\) as defined in \eqref{eq:rejectFov}, let
\({y}\) denote a realization of \(Y\), and define \(\tilde{d} = P_{X_{-j}}
y\). Define $\tilde{E}_2(y) =\{Y : P_{X_{-j}}Y =
P_{X_{-j}}y\}$. (Note that $Y\in \tilde E_2(y)$ implies $Y\in E_2(y)$ in \eqref{eq:E2E3}.)

Then, the conditional density of \(Y\) given \(Y \in E_{1}\) and \(Y \in \tilde{E}_2(y)\) is
\begin{equation}
  \label{eq:cond-dens}
  f_{Y \mid Y \in E_1 \cap \tilde{E}_2(y)} (\tilde y)
  = \frac
  {
    f_{\tZ} \left( D_1 \tilde y \right) f_{\tW} \left( D_2 \tilde y \right)
    \bm{I}_{ \left\{ \tilde y^{\top} D_2 \tilde y - \tilde y^{\top} D_1 \tilde y \geq - \tilde{d}^{\top} \tilde{d} \right\}}(\tilde y)
    \bm{I}_{ \left\{ D_3 \tilde y = \tilde{d} \right\}}(\tilde y)
  }
  {
    \Pr \left( \tW^{\top} \tW - c \tZ^{\top} \tZ \geq \tilde{d}^{\top} \tilde{d} \right)
  },
\end{equation}
where the marginal densities of \(\tZ\) and \(\tW\) are given in \eqref{eq:z-w-distn} and $\bm{I}_{\{\cdot\}}(\tilde y)$ denotes the indicator that $\tilde y$ satifies the condition inside the brackets.

\end{proposition}
\begin{proof}
  First, we can characterize \(E_1\) in terms of \(\left( \tZ, \tW, \tV
  \right)\) by observing that (almost surely)
  \begin{equation}
    \label{eq:E1-ZWV}
    \begin{aligned}
      Y \in E_1
      &\iff \frac{Y^{\top} P_{X} Y}{Y^{\top} \left( I_{n - 1} - P_{X} \right) Y} \geq c \\
      &\iff \frac{Y^\top (P_X-P_{X_{-j}}) Y + Y^\top P_{X_{-j}}Y}{Y^\top (I_{n-1}-P_X)Y}\geq c \\
      &\iff \frac{\tW^{\top} \tW + \tV^{\top} \tV}{\tZ^{\top} \tZ} \geq c \\
      &\iff \tW^{\top} \tW - c \tZ^{\top} \tZ \geq - \tV^{\top} \tV,
    \end{aligned}
  \end{equation}
  where the third equivalence exploits the idempotence and symmetry of projection matrices.
  We can similarly  characterize \(\tilde{E}_2(y)\) in terms of \(\tV\) by noting that
  \begin{equation}
    \label{eq:Etilde-ZWV}
    \begin{aligned}
      Y \in \tilde{E}_2({y})
      &\iff P_{X_{-j}} Y = P_{X_{-j}} {y} \\
      &\iff P_{X_{-j}} Y = \tilde{d} \\
      &\iff \tV = \tilde{d}.
    \end{aligned}
  \end{equation}
  Taken together, \eqref{eq:E1-ZWV} and \eqref{eq:Etilde-ZWV} enable us to
  characterize \(E_{1} \cap \tilde{E}_2({y})\) in terms of \(\left( \tZ,
    \tW, \tV \right)\) by observing that (almost surely)
  \begin{equation}
    \label{eq:E1Etilde-ZWV}
    Y \in E_1 \cap \tilde{E}_2({y})
    \iff
    \tW^{\top} \tW - c \tZ^{\top} \tZ \geq - \tilde{d}^{\top} \tilde{d}
    \text{ and } \tV = \tilde{d}.
  \end{equation}
  Next, note that because the transformation from \(Y\) to \(\left( \tZ, \tW, \tV \right)\) is given by \(D\), and \(D^{\top} D = I_{n - 1}\), a change of variables is straightforward:
  \begin{equation}
    \label{eq:cov}
    f_{Y} (\tilde y) = f_{\tZ, \tW, \tV} \left( D \tilde y \right).
  \end{equation}
  Then, we obtain the conditional density of \(Y\) given \(Y \in
  \tilde{E}_2({y})\) and find
  \begin{equation}
    \begin{aligned}
      \label{eq:cond-step-1}
      f_{Y \mid Y \in \tilde{E}_2({y})} (\tilde y)
      &= f_{Y \mid \tV = \tilde{d}} (\tilde y) \\
      &= f_{\tZ, \tW, \tV \mid \tV = \tilde{d}} (D \tilde y) \\
      &= f_{\tZ \mid \tV = \tilde{d}} \left( D_1 \tilde y \right) f_{\tW \mid \tV = \tilde{d}} \left( D_2 \tilde y \right) f_{\tV \mid \tV = \tilde{d}} \left( D_3 \tilde y \right) \\
      &= f_{\tZ} \left( D_1 \tilde y \right) f_{\tW} \left( D_2 \tilde y \right) f_{\tV \mid \tV = \tilde{d}} \left( D_3 \tilde y \right) \\
      &= f_{\tZ} \left( D_1 \tilde y \right) f_{\tW} \left( D_2 \tilde y \right) \bm{I}_{ \left\{ D_{3} \tilde y = \tilde{d} \right\}}(
      \tilde y),
    \end{aligned}
  \end{equation}
  where the first equality follows from \eqref{eq:Etilde-ZWV}, the second
  equality follows from \eqref{eq:cov}, the third from mutual independence of $\tZ,$ $\tW$, and $\tV$ per \cref{lemma:helper}, the fourth from the
  independence of \(\tZ\) from \(\tV\) and of \(\tW\) from
  \(\tV\), and the last equality follows because the conditional
  distribution of \(\tV \mid \tV = \tilde{d}\) is simply an atom at
  \(\tV = \tilde{d}\).

  Next, we (additionally) condition on \(Y \in E_1\) (since we had already
  conditioned on \(Y \in \tilde{E}_2({y})\) in \eqref{eq:cond-step-1}, we
  are now effectively conditioning on \(Y \in E_1 \cap \tilde{E}_2({y})\)).
  Standard application of the definition of conditional density yields
  \begin{equation}
    \label{eq:cond-step-2}
    \begin{aligned}
      f_{Y \mid Y \in E_1 \cap \tilde{E}_2({y})} (\tilde y)
      &= f_{Y \mid \tW^{\top} \tW - c \tZ^{\top} \tZ \geq - \tilde{d}^{\top} \tilde{d}, \tV = \tilde{d}} (\tilde y) \\
      &= f_{\left\{ Y \mid \tV = \tilde{d}  \right\} \mid \tW^\top \tW - c \tZ^\top \tZ \geq - \tilde{d}^\top \tilde{d}} (\tilde y) \\
      &= \frac
        {
        f_{ Y \mid \tV = \tilde{d}} (\tilde y)
        \bm{I}_{\left\{ \tilde y\; : \;\tilde y^{\top} D_2 \tilde y- c \tilde y^{\top} D_1 \tilde y  \geq -\tilde{d}^{\top} \tilde{d}\right\}}(\tilde y)
        }
        {
        \Pr \left( \tW^\top \tW- c\tZ^\top \tZ\geq -\tilde{d}^{\top} \tilde{d}  \right)
        } \\
      &= \frac
        {
        f_{\tZ} \left( D_1 \tilde y \right) f_{\tW} \left( D_2 \tilde y \right)
        \bm{I}_{\left\{ D_3 \tilde y = \tilde{d} \right\}}(\tilde y)
        \bm{I}_{\left\{ \tilde y\; : \;\tilde y^{\top} D_2 \tilde y- c \tilde y^{\top} D_1 \tilde y  \geq -\tilde{d}^{\top} \tilde{d}\right\}}(\tilde y)
        }
        {
        \Pr \left( \tW^\top \tW- c\tZ^\top \tZ\geq -\tilde{d}^{\top} \tilde{d}  \right)
        },
    \end{aligned}
  \end{equation}
  where the first equality follows from \eqref{eq:E1Etilde-ZWV}, the second
  equality follows because sequential conditioning is equivalent to simultaneous
  conditioning, the third equality from the standard definition of conditional
  density and the fact that $\tW = D_2 Y$ and $\tZ = D_1 Y$ from the statement of the  proposition, and the fourth by substituting the result of \eqref{eq:cond-step-1}.
\end{proof}

Now, we wish to view \eqref{eq:cond-step-2} as a likelihood, i.e., a function of
\(\beta_{j}\) and $\sigma^2$ with the realization of the data \(y\) fixed. First, note that
because \(Y \sim \mathcal{N}_{n - 1} \left( X \beta, \sigma^2 I_{n - 1}
\right)\),
\begin{equation}
  \label{eq:z-w-distn}
  \begin{aligned}
    \tZ &= \left(I_{n - 1} - P_{X}\right) Y \sim \mathcal{N}_{n - 1} \left( 0_{n - 1}, \sigma^2 \left( I_{n - 1} - P_{X} \right) \right), \\
    \tW &= \left( P_{X} - P_{X_{-j}} \right) Y \sim \mathcal{N}_{n - 1} \left( \beta_{j} \left( P_{X} - P_{X_{-j}}  \right) X_{j}, \sigma^2 \left( P_{X} - P_{X_{-j}} \right)\right).
  \end{aligned}
\end{equation}
By inspection, we can observe that the distribution of \(\tZ\) is invariant to
\(\beta_{j}\), while \(\tW\) involves \(\beta_{j}\) but none of \(\beta_{-j}\).
When \eqref{eq:cond-step-2} is viewed as a conditional likelihood under the data
realization $y$, the indicators must be satisfied (since we observed the data),
and maximizing it with respect to \(\beta_{j}\) and $\sigma^2$ is equivalent to maximizing
\begin{equation}
\label{eq:lik2}
L(\beta_{j}, \sigma^2)
:= \frac{f_{\tZ}(\sigma^2;\tilde z(y)) \cdot f_{\tW} (\beta_{j},\sigma^2; \tilde w(y)) }{\Pr_{\beta_{j},\sigma^2} \left( \tW^\top \tW- c\tZ^\top \tZ\geq -y^\top P_{X_{-j}}y  \right)},
\end{equation}
where $\tilde w(y) = (P_X-P_{X_{-j}})y$, $\tilde z(y) =(I_{n-1}-P_X)y$, and we write \(\Pr_{\beta_{j},\sigma^2}\) to emphasize that the probability depends
upon \(\beta_{j}\) and $\sigma^2$.

Recall that we observe $y$ and hence can compute $\tilde w(y)$, $\tilde z(y),$ and $y^\top P_{X_{-j}}y.$
For any candidate value $(\beta_j,\sigma^2)$, the
numerator can be computed exactly, and the denominator can be approximated via Monte Carlo. 
Therefore, to
approximate the conditional MLE, we numerically jointly maximize the logarithm of the
ratio in \eqref{eq:lik2} with respect to $\beta_j$ and $\sigma^2$.

We now comment on \cref{fig:point_estimates}. In the left-hand panel, where the overall null $\hov$
\eqref{eq:hov} holds and we condition on the rejection of $\hov$, the standard MLE has a much larger spread than the
conditional MLE derived in this section. This is because the overall null is rejected when
the F-statistic is sufficiently large, which tends to occur when $\hat\beta_1$ and/or $\hat\beta_2$ is large. Thus, the
standard MLE for $\beta_1$ will tend to be far from zero \emph{conditional on
  the rejection of the overall null}. The selective estimate, on the other hand,
adjusts for the selection event, and is therefore more tightly concentrated around it. In the right-hand panel, where $\beta_1=0.5$ and $\beta_2=0$, the unconditional MLE distribution is fairly tightly concentrated, but centered around 0.8, indicating upward bias after conditioning on rejection of $\hov.$ The conditonal MLE has a wider spread but is centered much closer to the true value.

\begin{figure}[h]
  \includegraphics[width = \textwidth, trim=0 0 0 0, clip]{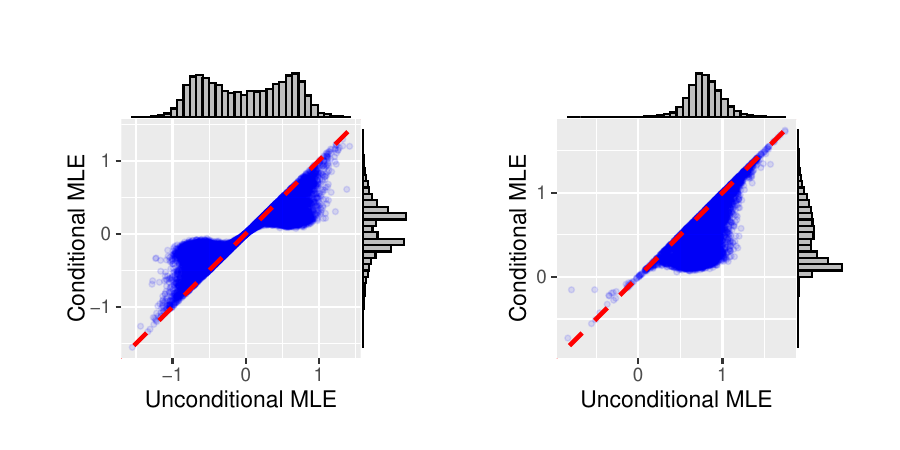}
  \caption{We generate datasets $(X,Y)$ with $n=50$, $p = 2$, and $\sigma^2=4$ until we have \num{1e4} that satisfy
    \eqref{eq:rejectFov}. In the left-hand panel, we set $\beta_1 = \beta_2=0$, and in the right-hand panel we set $\beta_1 =0.5$ and  $\beta_2=0$. For each dataset, we compute the standard,
    unconditional MLE of $\beta$, $\hat\beta_{\text{MLE}} = (X^\top X)^{-1} X^\top Y$.
    For the same datasets, we compute the conditional MLE of $\beta_1$ that
    arises from numerically maximizing the conditional likelihood in
    \eqref{eq:lik2}. Each dataset corresponds to a point on the above
    scatterplots, with the standard point estimate on the horizontal axis and
    the conditional point estimate on the vertical axis.}
  \label{fig:point_estimates}
\end{figure}

\subsection{Confidence intervals}\label{app:CI}

We now consider constructing selective confidence intervals for $\beta_j$, i.e.,
confidence intervals that attain the nominal coverage, conditional on having
rejected $\hov$. First, we extend our selective testing framework to
accommodate more general hypotheses of the form $\hob:\beta_j=b$. With this in
hand, confidence intervals can be constructed by means of test inversion, i.e.,
a level \(\left( 1 - \alpha \right)\) confidence interval will contain all
values of \(b\) for which we would \emph{not} reject \(\hob\) at level
\(\alpha\).

In order to construct a selective p-value for $\hob$, we consider the transformed data vector $Y-X_jb$, as suggested by \citet[Section 4.1]{fithian2017OptimalInferenceModel}. That is, we evaluate $\FhoM$ in \eqref{eq:fstat} at $y-X_jb$ instead of $y$ in our selective p-value. Further, instead of conditioning on $Y\in E_2(y) \cap E_3(y)$ in \eqref{eq:E2E3}, we define 
\begin{equation}\label{eq:tildeE2E3}
    \tilde E_2(y) := \{Y:P_{X_{-j}}Y = P_{X_{-j}}y\}, \quad \tilde E_3^b(y) :=\{Y: \|Y-X_jb\|^2 = \|y-X_jb\|^2\},
\end{equation}
and condition on 
\begin{equation}\label{eq:Etildeb}
    Y\in \tilde E^b(y) := E_1\cap \tilde E_2(y) \cap \tilde E_3^b(y) .
\end{equation}
Thus we consider the selective p-value
\begin{equation}
\pselb (y)  := \Pr_{\beta_j=b}  \left(  \frac{(Y-X_jb)^\top(P_X-P_{X_{-j}})(Y-X_jb)}{( Y-X_jb)^\top(I_{n-1}-P_X)(Y-X_jb)}  \geq
r (y-X_jb) \;\middle|\;
Y\in \tilde E^b(y)\right),
\label{eq:pselb1}
\end{equation}
where $r(\cdot)$ is defined in \eqref{eq:acdr}.

Observe that 
$$Y\in \tilde E_2(y) \iff P_{X_{-j}}Y = P_{X_{-j}}y \iff P_{X_{-j}}(Y-X_jb)  = P_{X_{-j}}(y-X_jb).$$ Thus \eqref{eq:pselb1} is equivalent to 
\begin{align}
    \pselb (y)  &:= \Pr_{\beta_j=b}  \Bigg(  \frac{(Y-X_jb)^\top(P_X-P_{X_{-j}})(Y-X_jb)}{( Y-X_jb)^\top(I_{n-1}-P_X)(Y-X_jb)}  \geq
r (y-X_jb)\nonumber\\
 & \hspace{1.5cm} \;\Bigg|\;
Y\in  E_1, \; P_{X_{-j}}( Y-X_jb) = P_{X_{-j}}(y-X_jb), \; \| Y -X_jb\|^2 =\| y-X_jb\|^2\Bigg).\label{eq:pselb2}
\end{align}
In order to represent \eqref{eq:pselb2} in a way that we can approximate via Monte Carlo, we first write the relevant test  statistics, 
$$\frac{(Y-X_jb)^\top(P_X-P_{X_{-j}})(Y-X_jb)}{(Y-X_jb)^\top(I_{n-1}-P_X)(Y-X_jb)} \quad \text{and} \quad \frac{Y^\top P_X Y}{Y^\top (I_{n-1}-P_X) Y},$$ 
in terms of the following three functions of our transformed data $Y - X_jb$:
\begin{equation}\label{eq:t1t2t3}T_1 := \|Y-X_jb\|^2,\quad T_2:=P_{X_{-j}}(Y-X_jb),\:\:\text{and} \:\:\:T_3:=X_j^\top(I_{n-1}-P_{X_{-j}})(Y-X_jb).\end{equation} 
Note that
\[
\operatorname{Span}(X)
=
\operatorname{Span}(X_{-j})
\oplus
\operatorname{Span}\big((I_{n-1}- P_{X_{-j}})X_j\big).
\]
Therefore, $ P_X=P_{X_{-j}}+P_{(I_{n-1}- P_{X_{-j}})X_j},$
and hence $P_{(I_{n-1}- P_{X_{-j}})X_j}=P_X - P_{X_{-j}}.$
It follows that 
\begin{align*}
&(Y-X_jb)^\top P_X (Y-X_jb)
\\ &=
(Y-X_jb)^\top P_{X_{-j}} (Y-X_jb)+ (Y-X_jb)^\top (P_X - P_{X_{-j}}) (Y-X_jb)\\
&=(Y-X_jb)^\top P_{X_{-j}} (Y-X_jb)+
(Y-X_jb)^\top P_{(I_{n-1}- P_{X_{-j}})X_j} (Y-X_jb)\\
& = T_2^\top T_2 +(Y-X_jb)^\top (I_{n-1}- P_{X_{-j}})X_j(X_j^\top (I_{n-1}- P_{X_{-j}})X_j)^{-1}X_j^\top (I_{n-1}- P_{X_{-j}})(Y-X_jb)\\
& =T_2^\top T_2+T_3^\top \lambda^{-1}T_3,
\end{align*}
where $\lambda = X_j^\top (I_{n-1}-P_{X_{-j}})X_j.$
Thus the overall $F$-statistic can be written as 
\begin{equation}\label{eq:Fhov_sample}
\begin{aligned}
\Fov &=
\frac{Y^\top P_X Y}{Y^\top (I_{n-1}- P_X)Y}\\
& = \frac{\left((Y-X_jb) + X_jb\right)^\top P_X ((Y-X_jb) + X_jb)}{((Y-X_jb) + X_jb)^\top (I_{n-1}-P_X)((Y-X_jb) + X_jb)}\\
& =\frac{(Y-X_jb)^\top P_X (Y-X_jb) + 2(Y-X_jb)^\top X_jb + b^2 X_j^\top X_j}{(Y-X_jb)^\top (I_{n-1}-P_X)(Y-X_jb)}\\
& =\frac{(Y-X_jb)^\top P_X (Y-X_jb) + 2(Y-X_jb)^\top (I_{n-1}-P_{X_{-j}})X_jb + 2(Y-X_jb)^\top P_{X_{-j}}X_jb+ b^2 X_j^\top X_j}{(Y-X_jb)^\top (Y-X_jb) -(Y-X_jb)^\top P_X (Y-X_jb)  }\\
& =\frac{T_2^\top T_2+ T_3^\top \lambda^{-1}T_3 + 2T_3^\top b + 2T_2^\top X_jb+ b^2 X_j^\top X_j}{T_1
-T_2^\top T_2- T_3^\top \lambda^{-1}T_3  }.
\end{aligned}
\end{equation}
Note that $Y \in \tilde E_3^b(y)$ \eqref{eq:tildeE2E3} fixes the value of $T_1 $ \eqref{eq:t1t2t3}, and $Y\in \tilde E_2(y)$ \eqref{eq:tildeE2E3} fixes the value of $T_2 $ \eqref{eq:t1t2t3}. Thus---once we condition on $Y\in \tilde E^b(y)$ \eqref{eq:Etildeb}---the only randomness in the above representation of $\Fov $ is in the random variable $T_3$.

Next we write the $F$–statistic $F_{\hob}$ for $\hob:\beta_j=b$ in terms of $T_1$, $T_2$ and $T_3$:
\begin{equation}\label{eq:FhoM_sample}
\begin{aligned}
F_{\hob}(Y)  &= 
 \frac{ (Y-X_jb)^\top (P_X - P_{X_{-j}}) (Y-X_jb) } {(Y-X_jb)^\top (I_{n-1}- P_X) (Y-X_jb)} \\
 &= \frac{ (Y-X_jb)^\top P_{(I_{n-1}- P_{X_{-j}})X_j} (Y-X_jb) } {(Y-X_jb)^\top (I_{n-1}- P_X) (Y-X_jb)}\\
 & =\frac{(Y-X_jb)^\top (I_{n-1}- P_{X_{-j}})X_j(X_j^\top (I_{n-1}- P_{X_{-j}})X_j)^{-1}X_j^\top (I_{n-1}- P_{X_{-j}})(Y-X_jb)}{ \|Y-X_jb\|^2
-(Y-X_jb)^\top P_X (Y-X_jb)}\\
& =\frac{T_3^\top \lambda^{-1}T_3}{ T_1
-T_2^\top T_2
-T_3^\top \lambda^{-1}T_3
}.
\end{aligned}
\end{equation}
With fixed $T_1$ and $T_2$, the only randomness in the above expression for $\FhoM$ is in $T_3$.

Finally, we characterize the conditional distribution of $T_3 \mid \{T_1, \;T_2 \}$, i.e.~the conditional distribution of 
$$X_j^\top (I_{n-1}- P_{X_{-j}})(Y-X_jb)
\ \Big|\ 
\left\{\|Y-X_jb\|^2,\ P_{X_{-j}}(Y-X_jb) \right\}.$$
Under $H_0:\beta_j = b,$ 
\begin{align*}
    (I_{n-1}- P_{X_{-j}})(Y-X_jb) &\sim \mathcal{N}\left((I_{n-1}- P_{X_{-j}})(X\beta - X_jb), \sigma^2 (I_{n-1}- P_{X_{-j}})\right) \\
    &\equiv \mathcal{N}(0, \sigma^2 (I_{n-1}- P_{X_{-j}})).
\end{align*}
Let $Z \sim \mathcal N(0,\sigma^2 I_{n-p})$ and $Q\in \R^{(n-1)\times (n-p)}$ such that $QQ^\top = I_{n-1} - P_{X_{-j}}$ and $Q^\top Q = I_{n-p}$. It follows that, under $H_0:\beta_j = b$,
\[
(I_{n-1}- P_{X_{-j}})(Y-X_jb)\overset{d}{=}QZ.
\]
Further, letting 
\begin{equation}\label{eq:adtilde}\tilde a= \|y - X_jb\|^2 \quad \text{and} \quad \tilde d=P_{X_{-j}}(y-X_jb),\end{equation} 
note that 
\begin{align*}
&X_j^\top (I_{n-1}- P_{X_{-j}})(Y-X_jb)
\ \Big|\ 
\left\{\|(Y-X_jb)\|^2 = \tilde a,\ P_{X_{-j}}(Y-X_jb) = \tilde d\right\}\\
\overset{d}{=} \;&X_j^\top (I_{n-1}- P_{X_{-j}})(Y-X_jb)
\ \Big|\ 
\Big\{(Y-X_jb)^\top (I_{n-1}- P_{X_{-j}})(Y-X_jb) = \tilde a - \tilde d^\top \tilde d, \\
& \hspace{11cm} \;P_{X_{-j}}(Y-X_jb) = \tilde d\Big\}\\
\overset{d}{=} \;&X_j^\top (I_{n-1}- P_{X_{-j}})(Y-X_jb)
\ \Big|\ 
\left\{(Y-X_jb)^\top (I_{n-1}- P_{X_{-j}})(Y-X_jb) = \tilde a - \tilde d^\top \tilde d\right\},
\end{align*}
where the second equivalence follows from independence of $(I_{n-1}-P_{X_{-j}})Y$ and $P_{X_{-j}}Y$ (or, more rigorously, from the fact that if $A$, $B$, and $C$ are random events where $(A, B)$ is independent of $C$ and $B$ is independent of $C$, then $\Pr(A|B, C) = \Pr(A|B)$). Hence, noting that $\|QZ\| = \|Z\|$,
\begin{align*}
&X_j^\top (I_{n-1}- P_{X_{-j}})(Y-X_jb)
\ \Big|\ 
\left\{\|Y-X_jb\|^2 = \tilde a,\ P_{X_{-j}}(Y-X_jb) = \tilde d \right\}\: \\ 
& \overset{d}{=}\;X_j^\top QZ \;\Big| \left\{\|Z\| = \sqrt{\tilde a - \tilde d^\top \tilde d}\right\}.
\end{align*} 
Finally, for $\tilde Z \sim \mathcal{N}(0,I_{n-p}),$
\begin{align*}
    Z\; \Big| \left\{\|Z\| =\sqrt{\tilde a - \tilde d^\top \tilde d}\right\}&\overset{d}{=} \;\|Z\|\cdot \frac{Z}{\|Z\|}\;\Big|\; \left\{\|Z\| =\sqrt{\tilde a - \tilde d^\top \tilde d}\right\}\\
    &\overset{d}{=} \;\sqrt{\tilde a - \tilde d^\top \tilde d}\cdot \frac{Z}{\|Z\|}\;\Big|\; \left\{\|Z\| =\sqrt{\tilde a - \tilde d^\top \tilde d}\right\}\\
    &\overset{d}{=} \;\sqrt{\tilde a - \tilde d^\top \tilde d}\cdot \frac{Z}{\|Z\|}\\
    &\overset{d}{=} \;\sqrt{\tilde a - \tilde d^\top \tilde d}\cdot \frac{\tilde Z}{\|\tilde Z\|},
\end{align*}
where the third line follows from the fact that the direction $\frac{Z}{\|Z\|}$ and magnitude $\|Z\|$ of an isotropic normal random variable are independent, and the fourth line follows from the fact that $Z \overset{d}{=} \sigma \widetilde Z.$ Hence, we can draw $\tilde Z \sim \mathcal{N}(0,I_{n-p})$ and obtain 
\begin{equation}\label{eq:G}
\begin{aligned}
\tfrac{\sqrt{\tilde a -\tilde d^\top \tilde d }}{\|\tilde Z\|}\cdot X_j^\top Q\tilde Z\: \overset{d}{=}\:\left[X_j^\top (I_{n-1}- P_{X_{-j}})(Y-X_jb)\ \Big|\ \|Y-X_jb\|^2 = \tilde a,\ P_{X_{-j}}(Y-X_jb) = \tilde d \right],
\end{aligned}
    \end{equation}
    under $\hob:\beta_j = b.$
    
For $\tilde Z \sim \mathcal{N}(0,I_{n-p})$, define $G := \tfrac{\sqrt{\tilde a -\tilde d^\top \tilde d }}{\|\tilde Z\|}\cdot X_j^\top Q\tilde Z$. With $\tilde a$ and $\tilde d$ defined in \eqref{eq:adtilde}, the p-value in \eqref{eq:pselb1} for the test of $\hob:\beta_j=b$ can be written as 
\begin{equation}\label{eq:pselb_chisq}
\begin{aligned}
        \pselb (y)  & = \Pr_{\beta_j=b}  \left(  \tfrac{(Y-X_jb)^\top(P_X-P_{X_{-j}})(Y-X_jb)}{( Y-X_jb)^\top(I_{n-1}-P_X)(Y-X_jb)}  \geq
r (y-X_jb) \;\middle|\;
Y\in E_1\cap \tilde E_2(y) \cap \tilde E_3^b(y)\right)\\
& = \frac{\Pr_{\beta_j=b}  \left(  \tfrac{(Y-X_jb)^\top(P_X-P_{X_{-j}})(Y-X_jb)}{( Y-X_jb)^\top(I_{n-1}-P_X)(Y-X_jb)}  \geq
r (y-X_jb) , \; Y\in E_1\;\middle|\;
Y\in  \tilde E_2(y) \cap \tilde E_3^b(y)\right)}{\Pr_{\beta_j=b}  \left(   Y\in E_1\;\middle|\;
 \tilde E_2(y) \cap \tilde E_3^b(y)\right)}\\
&=  \frac{\Pr \left ( 
        \tfrac{\tilde d^\top \tilde d + G^\top \lambda^{-1} G }{\tilde a- \tilde d^\top \tilde d - G^\top \lambda^{-1} G}\geq r(y-X_jb)\,,\: \frac{\tilde d^\top \tilde d + G^\top \lambda^{-1} G + 2 G^\top b + 2 \tilde d^\top X_j b + b^2X_j^\top X_j}{\tilde a-\tilde d^\top \tilde d - G^\top \lambda^{-1} G} \geq c\right)}
        {\Pr\left(
        \frac{\tilde d^\top \tilde d + G^\top \lambda^{-1} G + 2 G^\top b + 2 \tilde d^\top X_j b + b^2X_j^\top X_j}{\tilde a-\tilde d^\top \tilde d - G^\top \lambda^{-1} G} \geq c\right)}\\
        & = \Pr \left( 
        \tfrac{\tilde d^\top \tilde d + G^\top \lambda^{-1} G }{\tilde a- \tilde d^\top \tilde d - G^\top \lambda^{-1} G}\geq r(y-X_jb)\;\middle| \; \tfrac{\tilde d^\top \tilde d + G^\top \lambda^{-1} G + 2 G^\top b + 2 \tilde d^\top X_j b + b^2X_j^\top X_j}{\tilde a-\tilde d^\top \tilde d - G^\top \lambda^{-1} G} \geq c\right),
    \end{aligned}
    \end{equation}
    where the third equality follows from the derivations of \eqref{eq:Fhov_sample} and \eqref{eq:FhoM_sample}.
    Further, \eqref{eq:pselb_chisq} can be readily approximated via Monte Carlo. 
    Using this p-value, we define
selective confidence intervals:
\begin{equation}\label{eq:CI}
    \CI:= \left\{b\in \R: \pselb(y) \geq \alpha\right\}.
\end{equation}

The following proposition establishes that $\CI$ attains the nominal coverage, conditional on rejection of $\hov$. We see this empirically in the left-hand panels of \cref{fig:CI_globalnull} and \cref{fig:CI_localnull}.
\begin{proposition}
  \label{prop:pselb-coverage}
For any $\alpha \in (0,1)$, $\CI$ in \eqref{eq:CI} provides nominal coverage for the parameter $\beta_j$ conditional on having rejected $\hov$ in \eqref{eq:hov}. That is, 
 \begin{equation}
\label{eq:selective_coverage}
\Pr\left( \beta_j \in \CI \mid Y\in E_1 \right) = 1-\alpha
\end{equation} 
for $E_1$ defined in \eqref{eq:rejectFov}.
\end{proposition}  

\begin{proof}[Proof of Proposition~\ref{prop:pselb-coverage}] First we note that
  under the null $\hob:\beta_j = b$, the p-value $\pselb$ in
  \eqref{eq:pselb_chisq} follows a $\operatorname{Unif}(0,1)$ distribution,
  conditional on $Y\in E_1$ \eqref{eq:rejectFov}. This follows from a very
  similar argument to that in the proof of \cref{thm:psel-type1}. Hence,
\begin{align*}
\Pr\left(\beta_j \in \CI \mid Y\in E_1\right) &= \Pr\left(\pselb>\alpha  \mid Y\in E_1\right)\\
&=1- \Pr\left(\pselb\leq\alpha  \mid Y\in E_1\right)\\
& = 1-\alpha.
\end{align*}
\end{proof}

We observe that the p-value in \eqref{eq:pselb_chisq} depends on $X_j$ in a way
that cannot be deduced from statistics such as $\RR^2$, $\RSE$
\eqref{eq:RSE_R2}, and $\FhoM$ \eqref{eq:fstat}. Thus confidence intervals
cannot be computed retrospectively, in the sense of
Section \ref{sec:retrospective}.

\section{Extended simulation results}\label{app:extended_sim}
\renewcommand{\thefigure}{D.\arabic{figure}}
\renewcommand{\theequation}{D.\arabic{equation}}
\setcounter{figure}{0}
\setcounter{equation}{0}

\subsection{Alternative signal regime}

In the setting of \cref{sec:simulation}, we modify the simulation so that $\beta_j = 0.1$ for $j>1$. The results are shown in \cref{fig:triple_plot_appendix} and \cref{fig:CI_localnull}.

\begin{figure}[h]
    \centering
    \includegraphics[width=\linewidth]{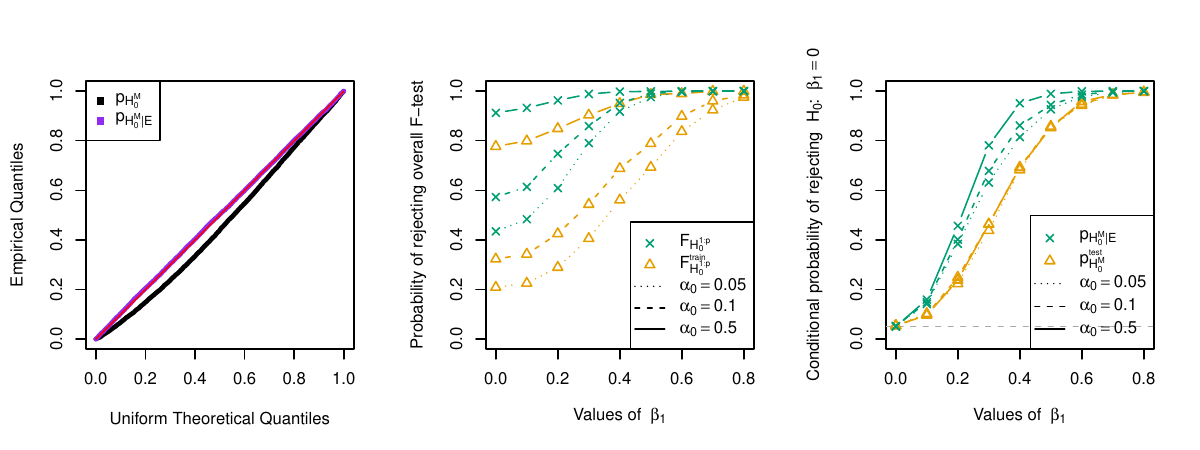}
    \caption{Details are as in \cref{fig:triple_plot}, but with $\beta_j = 0.1$ for $j>1$.}\label{fig:triple_plot_appendix}
\end{figure}

Compared to the setting in \cref{sec:simulation} and \cref{fig:triple_plot}, in the left panel of \cref{fig:triple_plot_appendix} the standard p-value $\phoM$ in \eqref{eq:pnaive} is less anti-conservative, and in the left panel of \cref{fig:CI_localnull} the standard confidence intervals achieve coverage closer to the nominal level, conditional on rejection of $\hov$ \eqref{eq:hov}. These improvements can be attributed to the fact that rejection of $\hov$ is highly likely when $\beta_2 = \ldots = \beta_p = 0.1$, so conditioning on $Y \in E$ \eqref{eq:E} has relatively little impact on the test of $\hoone$.  Thus, the performance of standard inference methods approaches the performance of selective inference methods.

The middle panel of \cref{fig:triple_plot_appendix} shows that when $\beta_2 = \cdots = \beta_p = 0.1$, the overall test of $\hov$ based on the full dataset has higher power than a test based only on the training data; both tests have higher power to reject $\hov$ than in the setting of \cref{fig:triple_plot}, in which $\beta_2 = \ldots =\beta_p=0$. Comparing the right panel of \cref{fig:triple_plot_appendix} to \cref{fig:triple_plot} reveals that the selective test of $\hoone$, which accounts for the rejection of $\hov$ and uses the full dataset, has greater power when $\beta_2 = \cdots = \beta_p = 0.1$ than when $\beta_2 = \ldots = \beta_p=0$. This is because in the former case, the rejection of $\hov$ is largely driven by the signal in the coefficients $\beta_j$ for $j > 1$, meaning that substantial signal about $\beta_1$ remains available for inference.

The middle and right-hand plot of \cref{fig:CI_localnull} exhibit similar patterns to those in \cref{fig:CI_globalnull}, so we do not provide further commentary.
\begin{figure}[h]
    \centering
    \includegraphics[width=\linewidth]{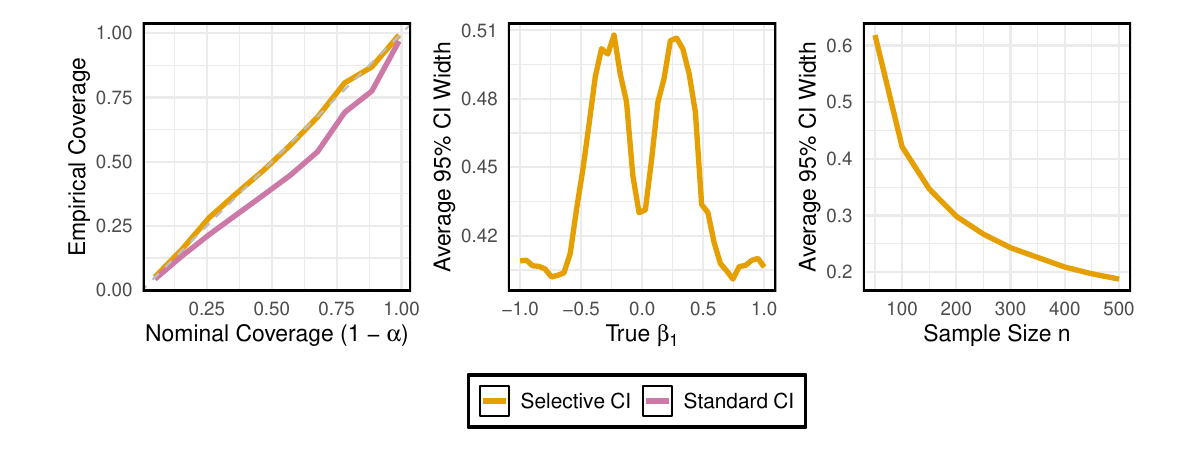}
    \caption{Details are as in \cref{fig:CI_globalnull}, but with $\beta_j = 0.1$ for $j>1$.}\label{fig:CI_localnull}
\end{figure}

\subsection{Sensitivity to departures from assumptions}

We study the behavior of our selective p-values under heteroskedasticty (\cref{fig:heterosked}), and non-normal errors (\cref{fig:non_normal} and \cref{fig:non_normal_anova}) via simulation.

In \cref{fig:heterosked}, we consider a heteroskedastic ANOVA setting. We generate a design matrix $X$ and data according to $Y=\epsilon$ (i.e., $\beta_1 = \ldots = \beta_p =0$ and so $\hov$ holds) as follows: 
\begin{equation}\label{eq:anova}
    X =  \begin{bmatrix}
    0 & 0 \\
    I_{20} & 0 \\
    0 & I_{20} 
\end{bmatrix}, \quad \epsilon\sim \mathcal{N}_{60}\left(0, \begin{bmatrix}
    1\cdot  I_{20} & 0 & 0 \\
    0 & 4\cdot I_{20} & 0\\
    0 & 0 & 2.25 \cdot I_{20}
\end{bmatrix}\right).
\end{equation}

The right-hand panel of \cref{fig:heterosked} compares the distribution of our selective p-values with the standard (unconditional) p-values in this heteroskedastic setting. While our p-value exhibits anti-conservative behavior, the deviation from the uniform is much less severe than that of the standard p-value. 

\begin{figure}[t]
\centering
\includegraphics[width = \textwidth]{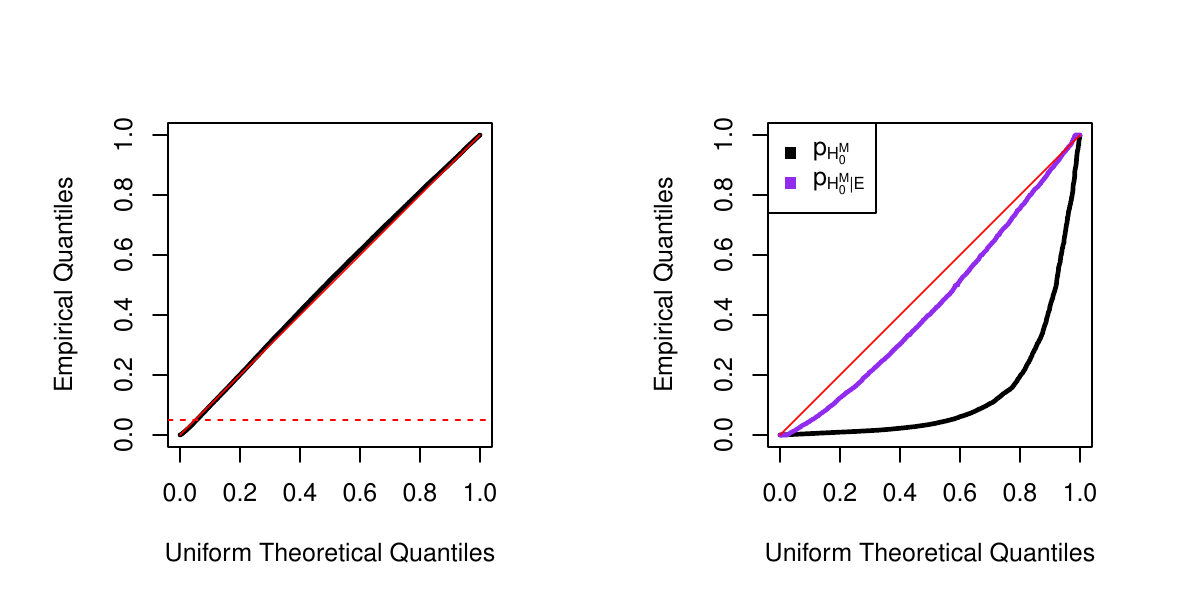}
\caption{For $n=60$, we generate the $n\times 2$ design matrix $X$ and data $Y$ according to \eqref{eq:anova} until 5000 datasets are obtained that reject $\hov$ with an $F$-test at level $\alpha = 0.05$. \emph{Left:} For all generated datasets (not just those that reject $\hov$), we display the p-value corresponding to the overall $F$-test on a uniform QQ-plot. We indicate the 45 degree line with a solid red line, and the horizontal dashed line indicates the 0.05 cutoff. Despite heteroskedasticity, these p-values appear uniform. \emph{Right:} For each of the 5000 datasets that reject $\hov$, we display both the proposed selective p-value \eqref{eq:pselchisq} and the standard (unconditional) p-value \eqref{eq:pnaive} for $\hoj$ on a uniform QQ-plot. The standard p-values are extremely anti-conservative, while the selective p-values for $\hoj$ are only mildly anti-conservative. }
\label{fig:heterosked}
\end{figure}

In \cref{fig:non_normal} and \cref{fig:non_normal_anova}, we examine the behavior of our p-values under non-Gaussian errors. Specifically, we consider the model $Y= \epsilon$ (i.e., $\beta_1 = \ldots = \beta_p =0$ and so $\hov$ holds) with error term $\epsilon \sim \mathcal{H}(0, 1)$, 
where $\mathcal{H}(0,1)$ denotes some parametric distribution with mean 0 and variance 1. In particular, we consider four choices for $\mathcal{H}$: (i) a Laplace distribution with location 0 and scale $1/\sqrt{2}$; (ii) an exponential distribution with rate 1, centered to have mean 0; (iii) a beta distribution with shape parameters 2 and 1, centered and scaled to have mean 0 and variance 1; and (iv) a Rademacher distribution that takes value $1$ with probability $0.5$ and $-1$ with probability $0.5$. These four distributions correspond to the four panels of \cref{fig:non_normal} and \cref{fig:non_normal_anova} with design matrix $X_{ij}
\overset{iid}{\sim}\mathcal{N}(0,1)$ and design matrix corresponding to the ANOVA model \eqref{eq:anova}, respectively. In \cref{fig:non_normal}, our selective p-values appear uniform, while the standard p-values are anti-conservative. In \cref{fig:non_normal_anova}, the standard p-values are extremely anti-conservative, while our selective p-values only appear anti-conservative in the case of Rademacher-distributed errors. The latter setting is somewhat pathological, as $Y = \epsilon$ and $\epsilon$ Rademacher implies that $Y\in \{-1,1\}$; in this setting it is not surprising that our proposal (with its finite-sample Gaussian guarantees) performs poorly.

\begin{figure}[]
\centering
\begin{tabular}{cc}

\textbf{Laplace distribution} & \textbf{Exponential distribution} \\[3pt]

\includegraphics[width=0.38\textwidth, trim={2cm 0cm 2cm 1.75cm},clip]{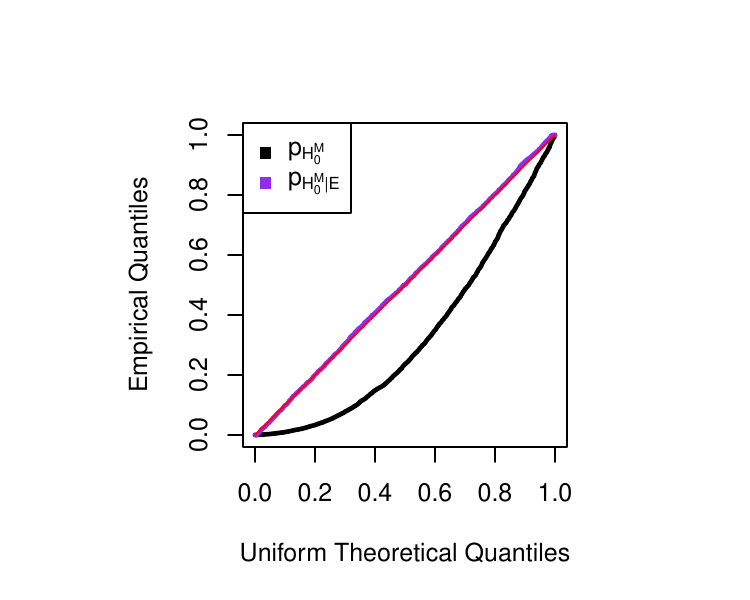} &
\includegraphics[width=0.38\textwidth, trim={2cm 0cm 2cm 1.75cm},clip]{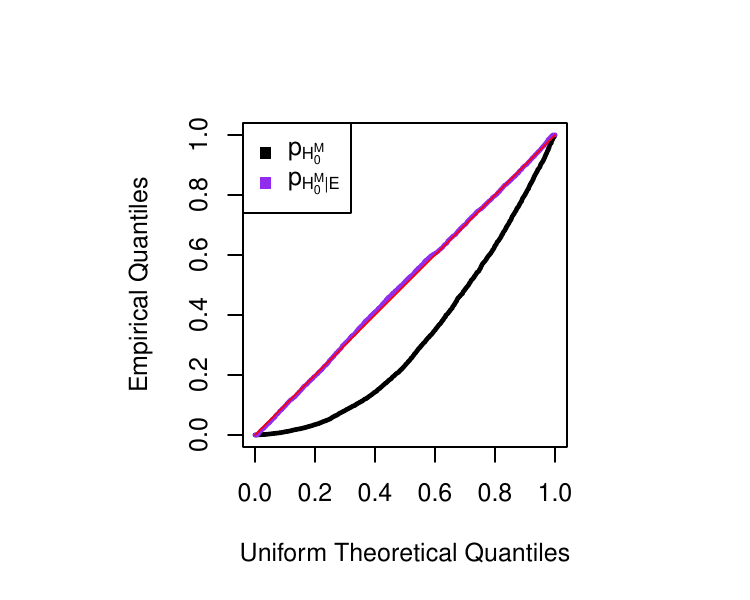} \\[6pt]

\textbf{Beta distribution} & \textbf{Rademacher distribution} \\[3pt]

\includegraphics[width=0.38\textwidth, trim={2cm 0cm 2cm 1.75cm},clip]{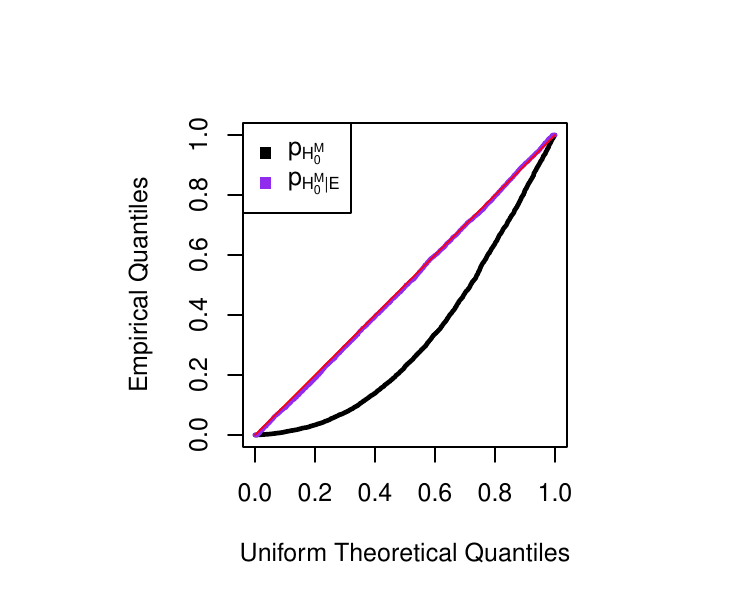} &
\includegraphics[width=0.38\textwidth, trim={2cm 0cm 2cm 1.75cm},clip]{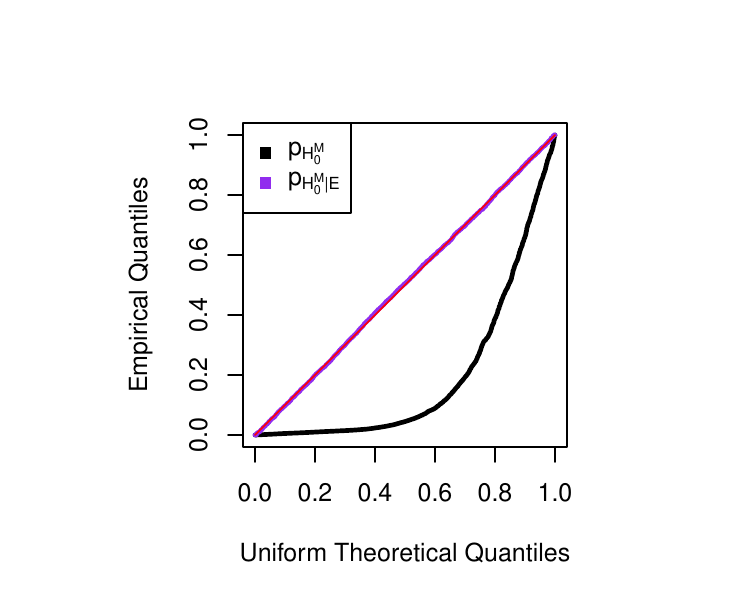}

\end{tabular}

\caption{For $n= 30, $ we simulate an $n\times 2$ design matrix $X$ with $X_{ij}
\overset{iid}{\sim}\mathcal{N}(0,1)$, and $n$ realizations of $Y\overset{iid}{\sim}\mathcal{H}(0,1)$, until we obtain 5000 datasets that reject $\hov$. Here, $\mathcal{H}$ takes four forms: (i) a Laplace distribution with location 0 and scale $1/\sqrt{2}$; (ii) an exponential distribution with rate 1, centered to have mean 0; (iii) a beta distribution with shape parameters 2 and 1, centered and scaled to have mean 0 and variance 1; and (iv) a Rademacher distribution that takes value $1$ with probability $0.5$ and $-1$ with probability $0.5$. For each dataset, we display the selective p-value \eqref{eq:pselchisq} and the standard p-value \eqref{eq:pnaive} for $\hoj$ on a uniform QQ-plot. Across all four non-Gaussian settings considered, our selective p-value for $\hoj$ appears uniform, while the standard p-values for $\hoj$ are anti-conservative because they do not account for rejection of $\hov.$
}
\label{fig:non_normal}
\end{figure}

\begin{figure}[]
\centering
\begin{tabular}{cc}

\textbf{Laplace distribution} & \textbf{Exponential distribution} \\[3pt]

\includegraphics[width=0.38\textwidth, trim={2cm 0cm 2cm 1.75cm},clip]{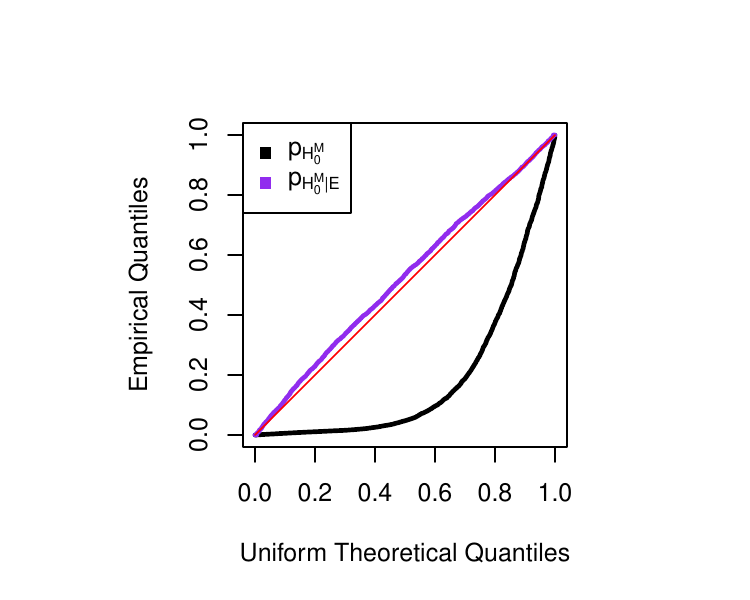} &
\includegraphics[width=0.38\textwidth, trim={2cm 0cm 2cm 1.75cm},clip]{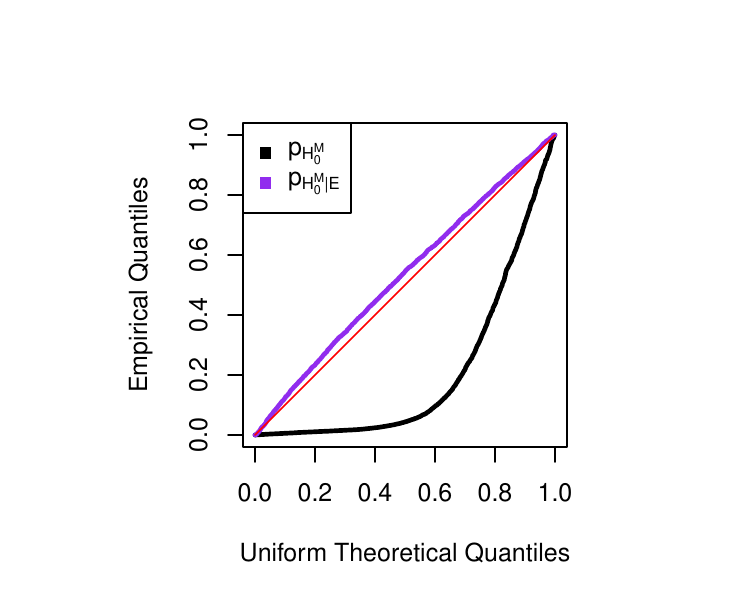} \\[6pt]

\textbf{Beta distribution} & \textbf{Rademacher distribution} \\[3pt]

\includegraphics[width=0.38\textwidth, trim={2cm 0cm 2cm 1.75cm},clip]{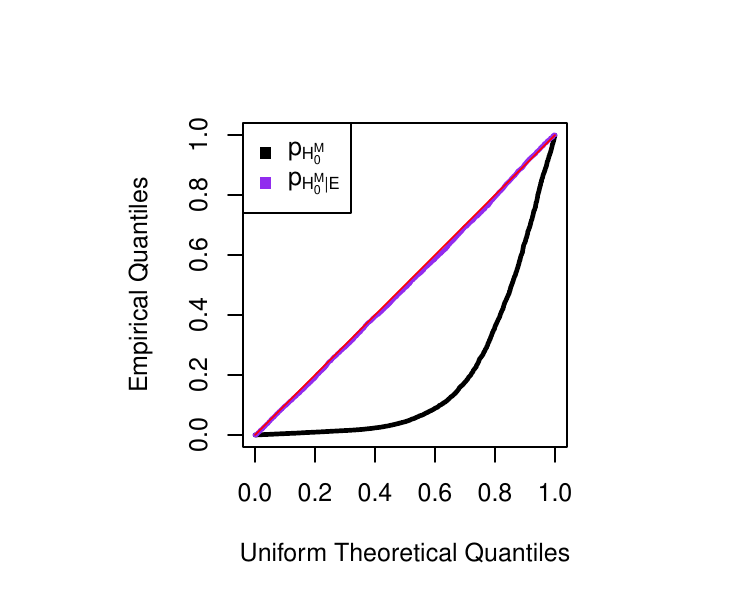} &
\includegraphics[width=0.38\textwidth, trim={2cm 0cm 2cm 1.75cm},clip]{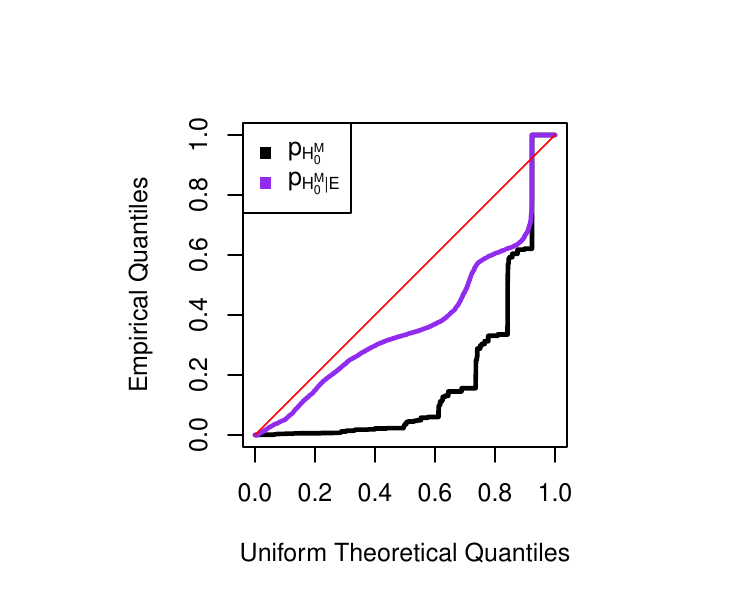}

\end{tabular}

\caption{For $n= 30, $ we simulate an $n\times 2$ design matrix with 
$X =  \begin{bmatrix}
    0 & 0 \\
    I_{10} & 0 \\
    0 & I_{10} 
\end{bmatrix}$, and $n$ realizations of $Y_i\overset{iid}{\sim}\mathcal{H}(0,1)$, until we obtain 5000 datasets that reject $\hov$. Here, $\mathcal{H}$ takes four forms: (i) a Laplace distribution with location 0 and scale $1/\sqrt{2}$; (ii) an exponential distribution with rate 1, centered to have mean 0; (iii) a beta distribution with shape parameters 2 and 1, centered and scaled to have mean 0 and variance 1; and (iv) a Rademacher distribution that takes value $1$ with probability $0.5$ and $-1$ with probability $0.5$. For each dataset, we display the selective p-value \eqref{eq:pselchisq} and the standard p-value \eqref{eq:pnaive} for $\hoj$ on a uniform QQ-plot. Across all four non-Gaussian settings considered, the standard p-values for $\hoj$ are highly anti-conservative, in part because they do not account for rejection of $\hov.$ For Laplace and exponential errors, our selective p-values for $\hoj$ are slightly conservative, whereas for Beta-distributed errors they appear uniform. In contrast, under Rademacher-distributed errors --- an extreme violation of our modeling assumptions since $Y=\epsilon$ with Rademacher errors implies that $Y$ is binary --- our selective p-values for $\hoj$ behave erratically and are anti-conservative. 
}
\label{fig:non_normal_anova}
\end{figure}

\section{Geometric intuition}
\label{app:geometry}
\renewcommand{\thefigure}{D.\arabic{figure}}
\renewcommand{\theequation}{D.\arabic{equation}}
\setcounter{figure}{0}
\setcounter{equation}{0}

In this section we explain the geometric view of the F-screening procedure
displayed in \cref{fig:3D_geometry} and \cref{fig:2D_geometry}. For simplicity, we assume that \(\beta_{0} =
0\) in the context of \eqref{eq:model}. We let $\alphaov = \alpha = 0.05.$

\subsection{ Figure~\ref{fig:2D_geometry}}
Let $ Y = \begin{bmatrix} Y_1 \\ Y_2 \end{bmatrix}$, $X = \begin{bmatrix} 1 & 0
  \\ 0 & 1 \end{bmatrix}$, and $\sigma^2 = 1.$ In this
setting, we assume that the error variance is known, and so we conduct $\chi^2$-tests instead of $F$-tests. 
\begin{itemize}
    \item Step 1 of Box~\ref{box:box1}: Under $ H_0^{1:p}: \beta_1 = \beta_2 = 0 $, we have that 
    $$
    \frac{Y^T P_X Y}{\sigma^2} = Y^T Y = Y_1^2 + Y_2^2 \sim \chi^2_2 .
    $$
    Thus, if the observed data $(y_1,y_2)$ lies outside the circle defined by $Y_1^2 + Y_2^2 = F_{\chi^2_2}^{-1}(1-\alphaov)$, then we reject this overall test, where $F_{\chi^2_p}$ is the CDF of a $\chi^2_p$ random variable. This corresponds to the circle in both plots of \cref{fig:2D_geometry}.
    \item Step 2 of Box~\ref{box:box1} (standard inference, left plot of Figure~\ref{fig:2D_geometry}): Under $\hoone: \beta_1 = 0$, we have that 
    \begin{align*}
    \frac{Y^T (P_X - P_{X_{-1}}) Y}{\sigma^2} &= Y^T \left(\begin{bmatrix} 1 & 0 \\ 0 & 1 \end{bmatrix} - \begin{bmatrix}0\\ 1\end{bmatrix}\left( \begin{bmatrix}0& 1\end{bmatrix} \begin{bmatrix}0\\ 1\end{bmatrix} \right)^{-1}\begin{bmatrix}0& 1\end{bmatrix} \right) Y \\
    &= Y_1^2\sim\chi_1^2,
    \end{align*}
    unconditional on the rejection of $\hov.$ Thus, if the observed data $(y_1,y_2)$ satisfies $|y_1|\geq \sqrt{F_{\chi^2_1}^{-1}(1-\alpha)}$, then we reject this standard test of $\hoone$. This corresponds to the vertical lines in the left plot of \cref{fig:2D_geometry}. 
    \item Step 2 of Box~\ref{box:box1} (selective inference, right plot in Figure~\ref{fig:2D_geometry}): We construct a selective p-value for $\hoone$ based on $\chi^2$-tests in a manner similar to \cref{prop:chisq}. We arrive at the p-value 
    \begin{align*}
    &p_{\chi}(y):= \Pr_{\beta_1 = 0}\Big( Y^\top(P_X-P_{X_{-1}})Y\geq \sigma^2y^\top(P_X-P_{X_{-1}})y\\
    & \quad \quad \quad \quad \quad \quad \mid  Y^\top(P_X-P_{X_{-1}})Y\geq \sigma^2 F_{\chi^2_2}^{-1}(1-\alphaov) -  y^\top P_{X_{-1}}y\Big)\\
    & = \Pr\left( W\geq \sigma^2y^\top(P_X-P_{X_{-1}})y\mid  W\geq \sigma^2 F_{\chi^2_2}^{-1}(1-\alphaov) -  y^\top P_{X_{-1}}y\right)\\
    & = \min \left \{1, \frac{\Pr\left( W\geq \sigma^2y^\top(P_X-P_{X_{-1}})y\right)}{\Pr\left(W\geq \sigma^2F_{\chi^2_2}^{-1}(1-\alphaov) -  y^\top P_{X_{-1}}y\right)}\right\} \\
    & = \min \left \{1, \frac{1-F_{\chi_1^2}\left( \sigma^2y^\top(P_X-P_{X_{-1}})y\right)}{1-F_{\chi_1^2}\left( \sigma^2F_{\chi^2_2}^{-1}(1-\alphaov) -  y^\top P_{X_{-1}}y\right)}\right\},
\end{align*}
where $W\sim\chi^2_1$. In our setting, $\sigma^2y^\top(P_X-P_{X_{-1}})y = y_1^2$ and $y^\top P_{X_{-1}}y =  y_2^2$. Thus, if the observed data $(y_1, y_2)$ lead to 
$$p_\chi(y)= \min \left \{1, \frac{1-F_{\chi_1^2}\left( y_1^2\right)}{1-F_{\chi_1^2}\left(F_{\chi^2_2}^{-1}(1-\alphaov) -  y_1^2\right)}\right\}\leq 
\alpha,$$ 
then we reject using the selective test of $\hoone.$ This corresponds to the purple region of the right plot in \cref{fig:2D_geometry}.
\end{itemize}

\subsection{Figure~\ref{fig:3D_geometry}}

Now, we assume that $\sigma^2$ is unknown, 
$$ Y = \begin{bmatrix} Y_1 \\ Y_2
  \\ Y_3 \end{bmatrix}, \quad \text{and} \quad X = \begin{bmatrix} 1 & 0 \\ 0 & 1
  \\ 0 & 0 \end{bmatrix} .$$ 

Since $\sigma^2$ is unknown, we use $F$-statistics to conduct the hypothesis tests. 
\begin{itemize}
    \item Step 1 of Box~\ref{box:box1}: Under $H_0^{1:p}: \beta_1 = \beta_2 = 0$, 
    $$
    \frac{Y^T P_X Y}{Y^T (I_{n-1}- P_X) Y}  = \frac{Y_1^2+Y_2^2}{Y_3^2}\sim F_{2,1}.
    $$
    Thus we reject $\hov$ if
    $$\frac{y_1^2+y_2^2}{y_3^2} \geq  F^{-1}_{2,1} (1-\alphaov).$$ The rejection boundary corresponds to the double cone in both panels of \cref{fig:3D_geometry}.
    \item Step 2 of Box~\ref{box:box1}: Under $\hoone: \beta_1 = 0$,
    \begin{align*}
     \frac{Y^T (P_X - P_{X_{-1}}) Y}{Y^T (I_{n-1}- P_X) Y}
    &= \frac{Y_1^2}{Y_3^2}\sim F_{1,1},
    \end{align*}
    unconditional on rejection of $\hov.$
    Therefore, we reject $\hoone$ if
    $$\frac{y_1^2}{y_3^2} \geq F^{-1}_{1,1} (1-\alpha) =: b.
    $$
    Thus the rejection boundary corresponds to the planes defined by $y_1/y_3\geq \sqrt{b}$ and $y_1/y_3\geq -\sqrt{b}$ in both panels of \cref{fig:3D_geometry}.
    
\end{itemize}

\section{Proofs of theoretical results in Sections~\ref{sec:method} and \ref{sec:retrospective}}
\renewcommand{\thefigure}{F.\arabic{figure}}
\renewcommand{\theequation}{F.\arabic{equation}}
\renewcommand{\theremark}{F.\arabic{remark}}
\setcounter{figure}{0}
\setcounter{equation}{0}
\setcounter{remark}{0}

\label{app:proofs}

In this appendix, we provide proofs for all theorems and
propositions from the main text. We begin by stating and proving a lemma
that will be useful.

\subsection{A useful lemma}
\begin{lemma} \label{lemma:helper}
The following statements hold under the model $Y= X\beta+ \epsilon$ where $X\in\R^{(n-1)\times p }$ and \(\epsilon \sim \mathcal{N}_{n - 1} \left( 0, \sigma^2 I_{n - 1}
\right)\):
\begin{enumerate}
\item Under $\hoM: \beta_M=0$, we have that $\tfrac{1}{\sigma^2}Y^\top (P_X - P_{{X}_{-M}}) Y  \sim \chi_m^2$. 
\item $\tfrac{1}{\sigma^2}Y^\top (I_{n-1} - P_X) Y \sim \chi^2_{n-p-1}$.
\item In general, the distribution of $Y^\top  P_{{X}_{-M}}  Y$ involves $\beta_j$ for $j\notin M$. 
\item $Y^\top (P_X - P_{{X}_{-M}}) Y $, $Y^\top (I_{n-1} - P_X) Y $, and $Y^\top P_{{X}_{-M}}  Y $ are mutually independent. 
   \end{enumerate}
\end{lemma}

\begin{proof}[Proof of \cref{lemma:helper}]

We begin by proving the fourth statement of \cref{lemma:helper}. First, note that
\(P_{X} - P_{X_{-M}}\), \(I_{n-1}- P_{X}\) and $P_{X_{-M}}$ are projection
matrices, and thus they are symmetric and idempotent. As a result, it follows that
$$Y^\top(P_X-P_{X_{-M}})Y = Y^\top(P_X-P_{X_{-M}})^\top (P_X-P_{X_{-M}}) Y = \left \lVert (P_X-P_{X_{-M}})Y \right \rVert_2^2.$$
Similarly, 
$$Y^\top (I_{n-1}-P_X)Y = \left \lVert (I_{n-1}-P_X)Y \right \rVert_2^2 \hspace{3mm} \text{and} \hspace{3mm} Y^\top P_{X_{-M}}Y =\left \lVert P_{X_{-M}}Y \right \rVert_2^2.$$
Thus it suffices to show that $(P_X-P_{X_{-M}})Y , $ $(I_{n-1}-P_X)Y , $ and
$P_{X_{-M}}Y$ are mutually independent.
Letting
\begin{equation*}
A =
\begin{bmatrix}
    P_X - P_{X_{-M}}\\
    I_{n-1} - P_X\\
    P_{X_{-M}}
\end{bmatrix},
\end{equation*}
observe that \(A Y \sim \mathcal{N}_{3(n-1)} \left(A X\beta, \sigma^2 A A^\top
\right)\).
Since this is a (degenerate) multivariate normal distribution, proving that $(P_X-P_{X_{-M}})Y , $ $(I_{n-1}-P_X)Y , $ and $P_{X_{-M}}Y$ are mutually independent is equivalent to proving that the off-diagonal blocks of the covariance matrix are all zero. The off-diagonal blocks of the covariance matrix, omitting the constant $\sigma^2$, are as follows: 
\begin{align*}
    &(P_X-P_{X_{-M}})(I_{n-1}-P_X)= (P_X-P_{X_{-M}}-P_X+P_{X_{-M}}P_X)= P_{X_{-M}}P_X-P_{X_{-M}}\\
    &(I_{n-1}-P_X)P_{X_{-M}} = P_{X_{-M}}- P_XP_{X_{-M}}\\
    &(P_X-P_{X_{-M}})P_{X_{-M}} = P_XP_{X_{-M}}-P_{X_{-M}}.
\end{align*}
Furthermore, noting that $R(X_{-M})\subseteq R(X)$ where $R(\cdot)$ denotes the column space, it follows from properties of projection matrices that $P_XP_{X_{-M}}=P_{X_{-M}}$, and hence also that $P_{X_{-M}}P_X=P_{X_{-M}}$.

The fourth statement of \cref{lemma:helper} follows.

To prove claims 1--3, we recall that if $Y \sim \mathcal{N}_{n-1}(\mu,\sigma^2 I_{n-1})$ and $P$ is a projection matrix of rank $r$, then $Y^\top PY/\sigma^2 \sim \chi^2_r (\lambda)$ where $\lambda = \mu^\top P\mu/(2\sigma^2)$ is the noncentrality parameter. Using this result, to prove part 1 of \cref{lemma:helper}, observe that $\frac{1}{\sigma} Y \sim \mathcal{N}_{n-1}\left(\frac{1}{\sigma}X\beta,I_{n-1}\right)$ and $I_{n-1}-P_X$ is a projection matrix of rank $n-p-1$. 
It then follows that 
$$\frac{1}{\sigma^2}Y^\top (I_{n-1}-P_X)Y \sim \chi^2_{n-p-1}(0 ) \equiv \chi^2_{n-p-1}.$$

To prove part 2, note that \begin{align*}
    (P_X-P_{X_{-M}})X\beta &= (P_X-P_{X_{-M}})X_M\beta_M + (P_X-P_{X_{-M}})X_{-M}\beta_{-M}\\
    & = (P_X-P_{X_{-M}})X_M\beta_M\\
    & = P_X X_M\beta_M - P_{X_{-M}}X_M\beta_M\\
    & = X_M\beta_M - P_{X_{-M}}X_M\beta_M\\
    & = (I_{n-1}- P_{X_{-M}})X_M\beta_M.\\
\end{align*}
Since $(P_X - P_{X_{-M}})$ is a projection matrix of rank $\left\lvert M \right\rvert =m$, it follows that 
\begin{equation}\label{eq:PXPXM}
\frac{1}{\sigma^2}Y^\top (P_{X} - P_{X_{-M}}) Y\sim \chi^2_m\left(\frac{(X_M\beta_M)^\top (I_{n-1}-P_{X_{-M}})X_M\beta_M}{2\sigma^2}\right).
\end{equation}
Therefore, under the null $\hoM :\beta_M=0$, we have $\frac{1}{\sigma^2}Y^\top (P_{X} - P_{X_{-M}}) Y \sim \chi^2_{m}.$ 

Finally, to prove the third statement, note that $P_{X_{-M}}$ is a projection matrix of rank $p-m$. Hence,
\begin{equation}\label{eq:PXM}
    \frac{1}{\sigma^2}Y^\top P_{X_{-M}} Y \sim \chi^2_{p-m}  \left(\frac{1}{2\sigma^2}(X\beta)^\top P_{X_{-M}} X\beta\right).
\end{equation}
Unless $X_M^\top X_M = 0$, the non-centrality parameter involves both $\beta_M$ and $\beta_{-M}.$

\end{proof}

\subsection{Proof of Theorem~\ref{thm:psel-type1}}
We begin with the following remark, which can be viewed as an extension of the probability integral transform to accommodate conditioning.
\begin{remark}\label{fact:pit}
Let $T$ be a random variable and $C$ an event. Let
$$F_{T \mid C}(t) = \Pr(T\leq t \mid C),$$
i.e., the conditional CDF of $T$ given $C$. Define $D = F_{T \mid C}(T)$. Then 
\begin{align*}
    \Pr(D\leq \alpha \mid C) & = \Pr(F_{T \mid C}(T)\leq \alpha \mid C) \\
    & = \Pr(T\leq F_{T \mid C}^{-1}(\alpha) \mid C)\\
    & = F_{T \mid C}(F_{T \mid C}^{-1}(\alpha))\\
    & = \alpha.
\end{align*}
Thus $D \sim \operatorname{Unif} \left( 0, 1 \right)$.
\end{remark}

\begin{proof}[Proof of \cref{thm:psel-type1}]
\cref{fact:pit} implies that, for $\alpha'\in (0,1)$ and $E$ defined in \eqref{eq:E},
$$Pr_{\beta_M=0}(\psel \leq \alpha'\mid Y\in E ) = \alpha'.$$ 
Equivalently, 
$$\mathbb{E}\left[\mathbbm{1}\{\psel\leq\alpha'\}\mid Y\in E \right] = \alpha'.$$
Because $Y\in E\implies Y\in E_1$ for $E_1$ defined in \eqref{eq:rejectFov} and $E(y)$ defined in \eqref{eq:E}, when we apply the law of total expectation, we obtain
\begin{align*}
    Pr_{\beta_M = 0}&(\psel\leq \alpha' \mid Y\in E_1) \\
    &= \mathbb{E}\left[\bm{I}\{\psel\leq\alpha'\}\mid Y\in E_1 \right]\\
    &= \mathbb{E}\Big[ \mathbb{E}\left[\bm{I}\{\psel\leq\alpha'\}\mid Y\in E_1\cap  E_2(y) \cap E_3(y)  \right]\mid Y\in E_1\Big]\\
    &= \mathbb{E}\Big[ \mathbb{E}\left[\bm{I}\{\psel\leq\alpha'\}\mid Y\in E \right]\mid Y\in E_1\Big]\\
    & = \mathbb{E} [\alpha'\mid Y\in E_1]\\
    & = \alpha'.
\end{align*}
\end{proof}

\subsection{Proof of Proposition \ref{prop:chisq}}
\begin{proof}[Proof of Proposition \ref{prop:chisq}]

First, note that both $\FhoM$ and $\Fov$ can be written in terms of the three quadratic forms $ Y^\top P_{X_{-M}}Y$, $\|Y\|^2$, and $ Y^\top (P_X-P_{X_{-M}})Y$. 
Indeed,
\begin{align*}
    \FhoM &= \tfrac{n-p-1}{m}\cdot \frac{Y^\top (P_X-P_{X_{-M}})Y}{Y^\top(I_{n-1}-P_X)Y} 
    = \tfrac{n-p-1}{m}\cdot\frac{Y^\top (P_X-P_{X_{-M}})Y}{Y^\top Y- Y^\top P_{X_{-M}}Y - Y^\top (P_X-P_{X_{-M}})Y} \\
    \Fov &= \tfrac{n-p-1}{p} \cdot \frac{Y^\top P_X Y}{Y^\top (I_{n-1}-P_X)Y } = \tfrac{n-p-1}{p} \cdot \frac{Y^\top P_{X_{-M}}Y +  Y^\top (P_X-P_{X_{-M}})Y }{Y^\top Y- Y^\top P_{X_{-M}}Y -  Y^\top (P_X-P_{X_{-M}})Y.}
\end{align*}

Next, we characterize the distribution
of 
 $Y^\top (P_X-P_{X_{-M}})Y \mid \{\|Y\|^2= a ,\; \|P_{X_{-M}}Y\|^2=d\},$
where $a \equiv a(y) = \|y\|^2$ and $d \equiv d(y) = \|P_{X_{-M}}y\|^2$ per \eqref{eq:acdr}. By independence of $ (P_X-P_{X_{-M}})Y $ and $P_{X_{-M}}Y$ per \cref{lemma:helper}, observe that 
\begin{align*}
    &Y^\top (P_X-P_{X_{-M}})Y \;\Big|\;\left \{\|Y\|^2= a ,\; \|P_{X_{-M}}Y\|^2=d\right\}\\
    &\overset{d} = Y^\top (P_X-P_{X_{-M}})Y \;\Big|\; \left\{\|(I_{n-1}-P_{X_{-M}})Y\|^2= a-d ,\; \|P_{X_{-M}}Y\|^2=d\right\}\\
    &\overset{d} = Y^\top (P_X-P_{X_{-M}})Y \;\Big|\; \left\{\|(I_{n-1}-P_{X_{-M}})Y\|^2= a-d \right\}\\
    &\overset{d} = \| (P_X-P_{X_{-M}})Y\|^2 \;\Big|\; \left\{\|(I_{n-1}-P_{X_{-M}})Y\|= \sqrt{a-d }\right\}\\
    &\overset{d} = \| (P_X-P_{X_{-M}})(P_{X_{-M}}Y + (I_{n-1}-P_{X_{-M}})Y)\|^2 \;\Big|\;\left\{\|(I_{n-1}-P_{X_{-M}})Y\|= \sqrt{a-d }\right\}\\
    &\overset{d} = \| (P_X-P_{X_{-M}})(I_{n-1}-P_{X_{-M}})Y\|^2 \;\Big|\; \left\{\|(I_{n-1}-P_{X_{-M}})Y\|= \sqrt{a-d }\right\}.
\end{align*}
Let $R\sim \mathcal{N}\left(0,\sigma^2 I_{(n-1)-(p-m)}\right)$ and $Q\in \R^{(n-1)\times ((n-1)-(p-m))}$ such that $Q Q^\top = I_{n-1}-P_{X_{-M}}$ and $Q^\top Q = I_{(n-1)-(p-m)}.$ Then, under $\hoM:\beta_M = 0$, $(I_{n-1}-P_{X_{-M}})Y \overset{d}{=} QR$. It follows that 
\begin{align*}
    &Y^\top (P_X-P_{X_{-M}})Y \;\Big|\;\left\{\|Y\|^2= a ,\; \|P_{X_{-M}}Y\|^2=d\right\}\\
     &\overset{d} =\| (P_X-P_{X_{-M}})QR\|^2 \;\Big|\; \left\{\|QR\|= \sqrt{a-d }\right\}\\
    &=\| (P_X-P_{X_{-M}})QR\|^2 \;\Big|\;\left\{ \|R\|= \sqrt{a-d }\right\}.
\end{align*}
Further, letting $\tilde R  \sim \mathcal{N}\left(0, I_{(n-1)-(p-m)}\right)$,
\begin{align*}
    &R^\top Q^\top  (P_X-P_{X_{-M}})QR \;\Big|\; \left\{\|R\|= \sqrt{a-d }\right\}\\
    &\overset{d}{=} \|R\|^2 \cdot \frac{R^\top }{\|R\|} Q^\top  (P_X-P_{X_{-M}})Q\frac{R}{\|R\|} \;\Big|\;\left\{\|R\|= \sqrt{a-d }\right\}\\
    &\overset{d}{=} (a-d) \cdot \frac{R^\top}{\|R\|} Q^\top  (P_X-P_{X_{-M}})Q\frac{R}{\|R\|} \;\Big| \;\left\{\|R\|= \sqrt{a-d }\right\}\\
    &\overset{d}{=} (a-d) \cdot \frac{R^\top}{\|R\|} Q^\top  (P_X-P_{X_{-M}})Q\frac{R}{\|R\|}\\
    &\overset{d}{=} (a-d) \cdot \frac{\tilde R^\top}{\|\tilde R\|} Q^\top  (P_X-P_{X_{-M}})Q\frac{ \tilde R}{\|\tilde R\|}
\end{align*}
where the fourth line follows from the fact that the direction $\frac{ R}{\|R\|}$ and magnitude $\|R\|$ of an isotropic normal random variable are independent, and the fifth line follows from the fact that $R \overset{d}{=} \sigma \tilde R.$ Observing that $\text{rank}(Q^\top (P_X-P_{X_{-M}})Q) = m$, we can represent the distribution of the above ratio of quadratic forms as
$$\frac{\tilde R^\top (Q^\top (P_X-P_{X_{-M}})Q)\tilde R}{\|\tilde R\|^2} \overset{d}{=} \frac{W}{Z+W}$$
for $W\sim \chi^2_m $ and $Z\sim \chi^2_{n-p-1}.$ We conclude that
\begin{align*}
    Y^\top (P_X-P_{X_{-M}})Y \;\Big|\;\left\{\|Y\|^2= a ,\; \|P_{X_{-M}}Y\|^2=d\right\}\;
    \overset{d}{=}\; (a-d)
    \cdot \frac{W}{Z+W}.
\end{align*}

Finally, we can represent the selective p-value in \eqref{eq:pselective} as follows:
\begin{align*}
    \psel &=   \Pr_{\beta_M=0} \left(  \frac{ Y^\top (P_X - P_{{X}_{-M}}) Y }{Y^\top (I_{n-1}- P_X) Y}  \geq  r(y)  \;\middle|\right.\\
     & \hspace{10mm}\left. \|Y\|^2 = a, \; \|P_{X_{-M}}Y\|^2 = d, \:\; \frac{Y^\top P_X Y}{Y^\top (I-P_X)Y} \geq c \right) \\
    & = \Pr_{\beta_M=0} \left(  \frac{Y^\top (P_X-P_{X_{-M}})Y}{Y^\top Y- Y^\top P_{X_{-M}}Y - Y^\top (P_X-P_{X_{-M}})Y} \geq  r(y)\;\middle|\right. \\
    &  \hspace{10mm}  \left.\|Y\|^2 = a, \; \|P_{X_{-M}}Y\|^2 = d, \: \frac{Y^\top P_{X_{-M}}Y +  Y^\top (P_X-P_{X_{-M}})Y }{Y^\top Y- Y^\top P_{X_{-M}}Y -  Y^\top (P_X-P_{X_{-M}})Y}\geq c \right)\\
    & = \Pr_{\beta_M=0} \left(  \frac{Y^\top (P_X-P_{X_{-M}})Y}{a- d - Y^\top (P_X-P_{X_{-M}})Y} \geq  r(y)\;\middle|\right. \\
    &  \hspace{10mm}  \left.\|Y\|^2 = a, \; \|P_{X_{-M}}Y\|^2 = d, \: \frac{d +  Y^\top (P_X-P_{X_{-M}})Y }{a-d -  Y^\top (P_X-P_{X_{-M}})Y}\geq c \right)\\
    & = \frac{\Pr_{\beta_M=0} \left(  \frac{Y^\top (P_X-P_{X_{-M}})Y}{a- d - Y^\top (P_X-P_{X_{-M}})Y} \geq  r(y), \; \frac{d +  Y^\top (P_X-P_{X_{-M}})Y }{a-d -  Y^\top (P_X-P_{X_{-M}})Y}\geq c\;\middle| \|Y\|^2 = a, \; \|P_{X_{-M}}Y\|^2 = d \right)}{\Pr_{\beta_M=0} \left( \frac{d +  Y^\top (P_X-P_{X_{-M}})Y }{a-d -  Y^\top (P_X-P_{X_{-M}})Y}\geq c\;\middle| \|Y\|^2 = a, \; \|P_{X_{-M}}Y\|^2 = d \right)}\\
    & = \frac{\Pr \left(  \frac{(a-d)\frac{W}{Z+W}}{a- d - (a-d)\frac{W}{Z+W}} \geq  r(y), \; \frac{d +  (a-d)\frac{W}{Z+W} }{a-d -  (a-d)\frac{W}{Z+W}}\geq c \right)}{\Pr \left( \frac{d +  (a-d)\frac{W}{Z+W} }{a-d -  (a-d)\frac{W}{Z+W}}\geq c \right)}\\
     & =  \Pr \left(  \frac{W}{Z} \geq  r(y), \;\middle|\; \frac{aW+dZ}{(a-d)Z}\geq c \right),\\
\end{align*}
where the last line follows from the definition of conditional probability and the fact that 
$$\frac{(a-d)\frac{W}{Z+W}}{a- d - (a-d)\frac{W}{Z+W}} = \frac{(a-d)W}{(a-d)(Z+W) - (a-d)W} = \frac{W}{Z},$$
$$\frac{d +  (a-d)\frac{W}{Z+W} }{a-d -  (a-d)\frac{W}{Z+W}} = \frac{d(W+Z) + (a-d)W}{(a-d)(W+Z)-(a-d)W} = \frac{aW+dZ}{(a-d)Z}.$$

\end{proof}

\subsection{Proof of Proposition~\ref{prop:test_consistency}}

We begin with two lemmas.

\begin{lemma}\label{lemma:delta}
    Consider a fixed sequence of design matrices \(X_{n} \in \mathbb{R}^{n \times p}\) and corresponding realizations of the linear model \(Y_{n} = X_{n} \beta + \epsilon_{n}\) from \eqref{eq:model} where $\beta \neq 0$. Suppose that \(\frac{1}{n} X_{n}^{\top} X_{n} \to Q\) for some positive definite $Q$. Then $\exists \delta >0$ and $N>0$ such that 
    $$\Pr\left( \frac{a(Y_n)W+d(Y_n)Z_n}{(a(Y_n)-d(Y_n))Z_n}\geq c_n \right)\geq \delta$$
    for all $n>N$, where \(W \sim \chi^2_{m}\), \(Z_{n} \sim \chi^2_{n - p - 1}\), \(c_{n} = \frac{p}{n - p - 1} F^{-1}_{p, n - p - 1} \left( 1 - \alpha_{0} \right)\), $a(Y_n) = Y_n^\top Y_n$ and \(d(Y_n) =  Y_{n}^{\top} P_{\left( X_{n} \right)_{-M}} Y_{n}\).
\end{lemma}

\begin{proof}[Proof of \cref{lemma:delta}]
In what follows, we suppress subscripts on $X$ involving \(n\) to enhance readability.

To show that $\Pr\left( \frac{a(Y_n)W+d(Y_n)Z_n}{(a(Y_n)-d(Y_n))Z_n}\geq c_n \right)\geq 0$ as $n\to \infty$, first observe that
\begin{itemize}
    \item $c_n\cdot Z_n$ converges in probability to $q:=(\chi^2_p)^{-1}(1-\alphaov)$ as $n\to \infty$,
    \item $d(Y_n)$ is a scaled $\chi^2$ random variable with $p-m$ degrees of freedom and non-centrality parameter $\tfrac{1}{2\sigma^2}(X\beta)^\top P_{X_{-M}}(X\beta)$, 
    \item $a(Y_n)$ is a scaled $\chi^2$ random variable with $(n-1)$ degrees of freedom and non-centrality parameter $\tfrac{1}{2\sigma^2}(X\beta)^\top (X\beta)$.
\end{itemize}
Since $\beta \neq 0,$ the non-centrality parameters of $d(Y_n)$ and $a(Y_n)$ are possibly non-zero. By our assumption that $\frac{1}{n}X^\top X \to Q$ deterministically for some positive definite matrix $Q\in \R^{p\times p}$, it follows that $\frac{1}{n}X^\top X_{-M}\to Q_{\cdot, -M}$ and $\frac{1}{n}X_{-M}^\top X_{-M}\to Q_{-M, -M}$ for $Q_{\cdot , -M}\in \R^{p\times (p-m)}$ and positive definite $Q_{-M,-M}\in \R^{(p-m)\times(p-m)}$. Hence,
\begin{align*}
    \tfrac{1}{n}\beta^\top X^\top P_{X_{-M}}X\beta & = \beta^\top\left(\tfrac{1}{n}X^\top X_{-M}\right)\left(\tfrac{1}{n}X_{-M}^\top X_{-M}\right)^{-1}\left(\tfrac{1}{n}X_{-M}^\top X \right)\beta\\
    &\to \beta^\top Q_{\cdot , -M} (Q_{-M,-M})^{-1}Q_{\cdot , -M}^\top \beta.
\end{align*}
and 
\begin{align*}
    \tfrac{1}{n}\beta^\top X^\top X\beta
    &\to \beta^\top Q \beta.
\end{align*}
Thus, when $\beta \neq 0$, $d(Y_n)$ and $a(Y_n)$ are each distributed as a scaled $\chi^2$ random variable with $O(n)$ non-centrality parameters. The sequence of expectations of $d(Y_n)/n$ and $a(Y_n)/n$ converges and each has variance $O(1/n).$ Thus by Chebyshev's inequality, it follows that $\phi_n := d(Y_n)/n\overset{P}{\to}\phi^*$ and $\varphi_n := a(Y_n)/n\overset{P}{\to}\varphi^*$ for some fixed $\phi^*$ and $\varphi^*$, where ``$\overset{P}{\to}$'' denotes convergence in probability. 

Therefore, for fixed $\varepsilon>0$ and $\eta>0$, we have that $\Pr(|\phi_n-\phi^*|\geq \varepsilon)<\eta$, $\Pr(|\varphi_n-\varphi^*|\geq \varepsilon)<\eta$ and $\Pr(|c_nZ_n-q|\geq \varepsilon)<\eta$ for sufficiently large $n$. For such an $n$, define 
$$G_n(\varepsilon)
=\{(\phi_n,Z_n):|\phi_n-\phi^*|< \varepsilon,\;|\varphi_n-\varphi^*|< \varepsilon,\;|c_nZ_n - q|<\varepsilon\},$$ so that $\Pr(G_n^c(\varepsilon)) \leq 3\eta$.
Note that 
\begin{align*}
|\varphi_n - \varphi^*|<\varepsilon &\iff \varphi_n <  \varphi^* +\varepsilon\quad \text{and} \quad \varphi_n > \varphi^*-\varepsilon,\\
|\phi_n - \phi^*|<\varepsilon & \iff \phi_n < \phi^* +\varepsilon \quad \text{and} \quad \phi_n > \varphi^*-\varepsilon,\\
|c_nZ_n - q|<\varepsilon &\iff c_n Z_n<q+\varepsilon  \quad \text{and} \quad c_nZ_n > q-\varepsilon.
\end{align*}

Choose $\varepsilon>0$ small enough that $\phi^* >\varepsilon$. Then
\begin{align*}
\Pr&\left( \frac{a(Y_n)W+d(Y_n)Z_n}{(a(Y_n)-d(Y_n))Z_n}\geq c_n \right)\\
&=\Pr\left( \frac{n\varphi_nW+n\phi_nZ_n}{(n\varphi_n-n\phi_n)Z_n}\geq c_n \right)\\
&=\Pr\left( \frac{\varphi_nW+\phi_nZ_n}{\varphi_n-\phi_n}\geq c_nZ_n\right)\\
&\geq\Pr\left( \frac{\varphi_nW+\phi_nZ_n}{\varphi_n-\phi_n}\geq c_nZ_n, \;G_n(\varepsilon)\right)\\
&\geq\Pr\left( \frac{(\varphi^*-\varepsilon)W+(\phi^*-\varepsilon)Z_n}{(\varphi^*+\varepsilon)-(\phi^*-\varepsilon)}\geq q+\varepsilon , \;G_n(\varepsilon)\right)\\
&\geq\Pr\left( \frac{(\varphi^*-\varepsilon)W+(\phi^*-\varepsilon)Z_n}{(\varphi^*+\varepsilon)-(\phi^*-\varepsilon)}\geq q+\varepsilon \right)-\Pr(G_n^c(\varepsilon))\\
&\geq\Pr\left( Z_n\geq \frac{(\varphi^*+\varepsilon)-(\phi^*-\varepsilon)}{\phi^*-\varepsilon}\left(q+\varepsilon
\right)-\frac{\varphi^*-\varepsilon}{\phi^*-\varepsilon}W \right)-3\eta\\
&\geq\Pr\left( Z_n\geq \frac{\varphi^*-\phi^* + 2\varepsilon}{\phi^*-\varepsilon}\left(q+\varepsilon \right)\right)-3\eta,
\end{align*}

where the last line follows from the fact that $(\varphi^*-\varepsilon)W/(\phi^*-\varepsilon)$ must be non-negative.
Note that the mean of $Z_n$ grows at rate $O(n)$, and $\tfrac{\varphi^*-\phi^* + 2\varepsilon}{\phi^*-\varepsilon}\left(q+\varepsilon \right)$ is a fixed finite value.
Since $\eta$ can be arbitrarily close to $0$, we conclude that $\Pr\left( Z_n\geq\frac{\varphi^*-\phi^* + 2\varepsilon}{\phi^*-\varepsilon}\left(q+\varepsilon \right)\right)-3\eta\to1$ as $n\to \infty$. Thus, 
$$\Pr\left( \frac{a(Y_n)W+d(Y_n)Z_n}{(a(Y_n)-d(Y_n))Z_n}\geq c_n \right)$$
is bounded away from 0 in the limit when $\beta\neq 0 $.

\end{proof}

\begin{lemma}\label{lemma:E1}
    For the same assumptions on $X_n$ and $Y_n$ as in \cref{lemma:delta}, suppose $f:\R^n\to\R$ is a function that satisfies $\Pr(f(Y_n)>\varepsilon)\to 0 $ for every $\varepsilon >0$. Then $\Pr(f(Y_n)>\varepsilon\mid Y_n \in E_1)\to 0 $ for every $\varepsilon >0$, where $E_1$ is defined in \eqref{eq:rejectFov}.
\end{lemma}

\begin{proof}[Proof of \cref{lemma:E1}]
    In what follows, we use the notation $\Pr_{A}(\cdot)$ to denote probability taken with respect to a random variable $A$ in order to make explicit which quantity is treated as random in each probability calculation. 
    
    Note that $\Pr_{Y_n}(Y_n\in E_1)\geq \alpha_0$ where $\alphaov>0$ is the level of the $F$-test of $\hov$ in Step 1 of Box~\ref{box:box1}. Thus
\begin{equation*}\label{eq:cond_uncond}
\begin{aligned}
    & \Pr_{Y_n}\left(f(Y_n) > \varepsilon\right) \to 0  \\
    & \quad\implies \Pr_{Y_n}\left(f(Y_n) > \varepsilon\mid Y_n\in E_1\right)\Pr_{Y_n}\left(Y_n\in E_1\right) \\
    & \quad\quad \quad \quad \quad + \Pr_{Y_n}\left(f(Y_n) > \varepsilon\mid Y_n\notin E_1\right)\Pr_{Y_n}\left(Y_n\notin E_1\right) \to 0\\
    &\quad\implies \Pr_{Y_n}\left(f(Y_n) > \varepsilon\mid Y_n\in E_1\right)\Pr_{Y_n}\left(Y_n\in E_1\right)\to 0 \\
    &\quad \implies \Pr_{Y_n}\left(f(Y_n) > \varepsilon\mid Y_n\in E_1\right)\to 0.
\end{aligned}
\end{equation*}
\end{proof}

We now continue with the proof of \cref{prop:test_consistency}.

\begin{proof}[Proof of \cref{prop:test_consistency}]
A test is consistent if the probability of rejection under any fixed alternative converges to 1 as $n\to\infty$. Thus the test of $\hoM$ based on the selective p-value in \eqref{eq:pselchisq} is consistent if and only if the selective p-value converges to 0 under any fixed alternative. That is,
we want to show that if $\beta\neq 0$, then $\forall \varepsilon>0$,
\begin{align*}
    \Pr_{Y_n}&\left(\psel(Y_n) > \varepsilon \mid Y_n\in E_1\right) \\
    &= \Pr_{Y_n}\left(\Pr_{W,Z_n} \left(  \frac{W}{Z_n} \geq  r(Y_n), \;\middle|\; \frac{a(Y_n)W+d(Y_n)Z_n}{(a(Y_n)-d(Y_n))Z_n}\geq c_n \right)>\varepsilon \;\middle|\;Y_n\in E_1\right) \to 0.
\end{align*}
By \cref{lemma:E1}, if $\Pr_{Y_n}(\psel(Y_n)>\varepsilon)\to 0  $ then $\Pr_{Y_n}(\psel(Y_n)>\varepsilon\mid Y_n\in E_1)\to 0  $.
Thus it is sufficient to show that whenever $\beta\neq0,$ 
\begin{align*}
\Pr_{Y_n}&\left(\psel(Y_n)>\varepsilon\right) \\
&=\Pr_{Y_n}\left(\Pr_{W,Z_n} \left(  \frac{W}{Z_n} \geq  r(Y_n), \;\middle|\; \frac{a(Y_n)W+d(Y_n)Z_n}{(a(Y_n)-d(Y_n))Z_n}\geq c_n \right)>\varepsilon \;\middle|\;Y_n\in E_1\right)\to 0.
\end{align*}

We begin by showing that when $\beta_M\neq 0$, it holds that
$$\lim_{n\to\infty}\Pr_{Y_n,W,Z_n}\left(\frac{W}{Z_n} > r(Y_n)\right)=\lim_{n\to\infty}\Pr_{Y_n}\left(\Pr_{W,Z_n}\left(\frac{W}{Z_n} > r(Y_n)\right)\right)=0.$$
We have assumed that $\frac{1}{n}X^\top X \to Q$ deterministically for some positive definite matrix $Q\in \R^{p\times p}$, and it therefore follows that $\frac{1}{n}X_M^\top X_{-M}\to Q_{M, -M}$, $\frac{1}{n}X_{-M}^\top X_{-M}\to Q_{-M, -M}$, and $\frac{1}{n}X_{M}^\top X_{M}\to Q_{M, M}$ for $Q_{M, -M}\in \R^{m\times (p-m)}$ and positive definite matrices $Q_{M, M}\in \R^{m\times m}$ and $Q_{-M,-M}\in \R^{(p-m)\times(p-m)}$. Thus
\begin{align*}
    \tfrac{1}{n}X_M^\top (I-P_{X_{-M}})X_M &=\tfrac{1}{n}X_M^\top  X_M - \tfrac{1}{n}X_M^\top P_{X_{-M}} X_M\\
    & = \tfrac{1}{n}X_M^\top X_M - \tfrac{1}{n}X_M^\top X_{-M}(X_{-M}^\top X_{-M})^{-1}X_{-M}^\top X_M\\
    & =\tfrac{1}{n}X_M^\top X_M - \left(\tfrac{1}{n}X_M^\top X_{-M}\right)\left(\tfrac{1}{n}X_{-M}^\top X_{-M}\right)^{-1}\left(\tfrac{1}{n}X_{-M}^\top X_M\right)\\
    & \to Q_{M,M} - Q_{M, -M}(Q_{-M,-M})^{-1}Q_{M,-M}^\top.
\end{align*} 
It follows that $\frac{1}{2\sigma^2} (X_M\beta_M)^\top (I-P_{X_{-M}})X_M\beta_M=O(n)$ when $\beta_M\neq 0$. Hence, by \eqref{eq:PXPXM}, when $\beta\neq 0$, 
$\frac{1}{\sigma^2}Y_n^\top (P_X-P_{X_{-M}})Y_n$ is distributed as a $\chi^2$ random variable with a non-centrality parameter that is $O(n).$ The sequence of expectations of $\frac{1}{n}Y_n^\top (P_X-P_{X_{-M}})Y_n$ converges, and the variance is $O(1/n)$. Hence, by Chebyshev's inequality, it follows that $\frac{1}{n} Y_n^\top (P_X-P_{X_{-M}})Y_n\overset{P}{\to}C$ for some constant $C> 0$. Additionally observe that $\frac{1}{n}Y_n^\top (I-P_X)Y_n\overset{P}{\to} \sigma^2$, again by Chebyshev's inequality. Hence, by the continuous mapping theorem
$$ r(Y_n) =\frac{Y_n^\top (P_X - P_{{X}_{-M}}) Y_n}{Y_n^\top (I_{n-1}- P_X) Y_n}=\frac{\tfrac{1}{n}Y_n^\top (P_X - P_{{X}_{-M}}) Y_n}{\tfrac{1}{n}Y_n^\top (I_{n-1}- P_X) Y_n}\overset{P}{\to }\frac{C}{\sigma^2}>0.$$
Further, 
$W/Z_n = \left(\tfrac{1}{n}W\right)/\left(\tfrac{1}{n}Z_n \right)\overset{P}{\to }0/1 = 0 .$
That is, for any $\varepsilon >0, $ $\Pr_{W,Z_n}(W/Z_n>\varepsilon )\to 0.$
We conclude that 
$\Pr\left(W/Z_n>r(Y_n) \right) \overset{P}{\to} 0$.

Next, note that for any $y\in \R^n$,
\begin{align*}
	\Pr_{W,Z_n}\left(\frac{W}{Z_n} > r(y) \;\middle|\; \frac{a(Y_n)W+d(Y_n)Z_n}{(a(Y_n)-d(Y_n))Z_n}\geq c_n \right) &= \frac{\Pr_{W,Z_n}\left(\frac{W}{Z_n} > r(y) ,\;  \frac{a(Y_n)W+d(Y_n)Z_n}{(a(Y_n)-d(Y_n))Z_n}\geq c_n  \right)}{\Pr_{W,Z_n}\left( \frac{a(Y_n)W+d(Y_n)Z_n}{(a(Y_n)-d(Y_n))Z_n}\geq c_n \right)}\\
	& \leq\frac{\Pr_{W,Z_n}\left(W/Z_n > r(y) \right)}{\Pr_{W,Z_n}\left( \frac{a(Y_n)W+d(Y_n)Z_n}{(a(Y_n)-d(Y_n))Z_n}\geq c_n  \right)}.
\end{align*}
Thus, 
\begin{align*}
    \Pr_{Y_n}&\left(\Pr_{W,Z_n} \left(  \frac{W}{Z_n} \geq  r(Y_n), \;\middle|\; \frac{a(Y_n)W+d(Y_n)Z_n}{(a(Y_n)-d(Y_n))Z_n}\geq c_n \right) >\varepsilon\right) \\
    &\hspace{30mm}\leq\Pr_{Y_n}\left(\frac{\Pr_{W,Z_n}\left(W/Z_n > r(Y_n) \right)}{\Pr_{W,Z_n}\left( \frac{a(Y_n)W+d(Y_n)Z_n}{(a(Y_n)-d(Y_n))Z_n}\geq c_n  \right)}>\varepsilon\right).
\end{align*}
Since the denominator $\Pr_{W,Z_n}\left( \frac{a(Y_n)W+d(Y_n)Z_n}{(a(Y_n)-d(Y_n))Z_n}\geq c_n  \right)$ is bounded away from zero in the limit by \cref{lemma:delta}, we conclude that 
\begin{align*}
\Pr_{W,Z_n}\left(W/Z_n > r(Y_n) \right)&\overset{P}{\to} 0 \implies \frac{\Pr_{W,Z_n}\left(W/Z_n > r(Y_n) \right)}{\Pr_{W,Z_n}\left( \frac{a(Y_n)W+d(Y_n)Z_n}{(a(Y_n)-d(Y_n))Z_n}\geq c_n  \right)} \overset{P}{\to} 0 \\
&\implies \Pr_{Y_n}\left(\frac{\Pr_{W,Z_n}\left(W/Z_n > r(Y_n) \right)}{\Pr_{W,Z_n}\left( \frac{a(Y_n)W+d(Y_n)Z_n}{(a(Y_n)-d(Y_n))Z_n}\geq c_n  \right)}>\varepsilon\right)\to 0\\
& \implies \Pr_{Y_n}\left(\Pr_{W,Z_n}\left(\frac{W}{Z_n} > r(Y_n) \;\middle|\; \frac{a(Y_n)W+d(Y_n)Z_n}{(a(Y_n)-d(Y_n))Z_n}\geq c_n \right) >\varepsilon\right) \to 0.
\end{align*}
Thus we have established that $\Pr_{Y_n}\left(\psel(Y_n) >\varepsilon\right)\to 0.$ By \cref{lemma:E1}, this implies that $$\Pr\left(\psel(Y_n) >\varepsilon\mid Y\in E_1\right)\to 0, $$ and equivalently that $\Pr\left(\psel(Y_n) <\varepsilon\mid Y\in E_1\right)\to 1. $ This completes the proof.
\end{proof}

\subsection{Proof of Proposition~\ref{prop:fisher-info}}

\begin{proof}[Proof of Proposition~\ref{prop:fisher-info}]
    
Recall that for a proportion $\rho\in(0,1)$ and a realization $\tilde y\in \R^n$ we define 

\begin{equation*}
\begin{aligned}
    \Rrho(X)  & := \{Y \in  \R^{(\rho  n) }: Y^\top P_X Y \geq c_\chi\},\\
    R(\tilde y;X) &:= \{Y\in \R^n: Y^\top P_XY \geq c_\chi,\; P_{X_{-1}}Y = P_{X_{-1}}\tilde y\},\\
    c_\chi & := \sigma^2 (\chi_p^2)^{-1}(1-\alphaov),\\
    e_{\tilde y} & := c_\chi - \tilde y^\top P_{X_{-1}} \tilde y.
\end{aligned}
\end{equation*} 
Let $(\Xtrain, \Ytrain)$ be the training dataset with $\rho$ proportion of the observations and $(\Xtest, \Ytest)$ the test dataset with $(1-\rho)$ proportion of the observations. We begin by deriving  the leftover Fisher information in $\Ytest$ about $\beta_1$ conditional on  $\Ytrain \in \Rrho(\Xtrain)$, where we treat $\beta_2,\ldots,\beta_p$ as nuisance parameters. That is, we will find an expression for
\begin{eqnarray}
\left(\left[\left[\mathcal{I}_{\Ytest\mid \Ytrain \in \Rrho(\Xtrain) }(\beta;\Ytrain \in \Rrho(\Xtrain))\right]^{-1}\right]_{1,1}\right)^{-1} ,
\end{eqnarray}
where $\mathcal{I}_{\Ytest\mid \Ytrain \in \Rrho(\Xtrain) }(\beta;\Ytrain \in \Rrho(\Xtrain))$ is the Fisher information matrix for the entire parameter vector $\beta$. We consider the inverse of the $(1,1)$ element of the inverted Fisher information matrix in order to isolate the information in $\Ytest$ about $\beta_1$ while accounting for the fact that $\beta_{-1}$ is also unknown.

Since $\Ytrain$ and $\Ytest$ are independent, the density of $\Ytest$ conditional on $\Ytrain\in R_\rho(\Xtrain) $ is equivalent to the unconditional density of $\Ytest$, i.e. 
$$f_{\Ytest \mid \Ytrain \in R_\rho(\Xtrain)  }(y) = f_{\Ytest}(y)  = (2\pi\sigma^2)^{-n/2}\exp\left\{\frac{1}{2\sigma^2}(y  - \Xtest\beta)^\top (y -\Xtest \beta)\right\}. $$ 
Since we are concerned with the derivative of this expression with respect to $\beta,$ we will now consider this function as a likelihood of $\beta$ with $y\in \R^n$ fixed.  
Then the Hessian of the log likelihood is 
$$\frac{d^2}{d(\beta_1,\; \beta_{-1})d(\beta_1,\; \beta_{-1})^\top}\log f_{\Ytest}(\beta;y)
= -\frac{1}{\sigma^2}\begin{bmatrix}
 (\Xtest)_{1}^\top (\Xtest)_{1} & (\Xtest)_{1}^\top (\Xtest)_{-1}\\
(\Xtest)_{-1}^{\top}(\Xtest)_{1} & (\Xtest)_{-1}^{\top} (\Xtest)_{-1}\\
\end{bmatrix}.$$ Note that this Hessian is non-random because the test design matrix $\Xtest$ is fixed. Treating $\beta_{-1}$ as a nuisance parameter, we take the Schur complement to obtain the $(1,1)$ entry of the inverse Fisher information matrix and conclude that the leftover Fisher information is 
\begin{align*}
&\left(\left[\left[\mathcal{I}_{\Ytest\mid \Ytrain \in \Rrho(\Xtrain) }(\beta;\Ytrain \in \Rrho(\Xtrain))\right]^{-1}\right]_{1,1}\right)^{-1}  \\
& \quad = \frac{1}{\sigma^2}\left((\Xtest)_{1}^\top (\Xtest)_{1}  -[(\Xtest)_{1}^\top (\Xtest)_{-1}][(\Xtest)_{-1}^{\top}(\Xtest)_{-1}]^{-1}[(\Xtest)_{-1}^{\top}(\Xtest)_{1}]\right)\\
& \quad = \frac{1}{\sigma^2}\left((\Xtest)_{1}^\top (I - P_{(\Xtest)_{-1}})(\Xtest)_{1}\right).
\end{align*}

Next, we derive an expression for 
$\mathcal{I}_{Y\mid   Y \in R(\tilde y; X)}(\beta_1;Y \in R(\tilde y; X))$ for some realization $\tilde y\in \R^n$. Let
\begin{enumerate}
    \item $U_1\in \R^{n\times (n-p)}$ such that $U_1U_1^\top = I_{n}-P_X$ and $U_1^\top U_1= I_{n-p}$, 
    \item $U_2\in \R^{n}$ such that $U_2U_2^\top = P_X-P_{X_{-1}}$ and $U_2^\top U_2= 1$, 
    \item $U_3\in \R^{n\times (p-1)}$ such that $U_3U_3^\top =P_{X_{-1}}$ and $U_3^\top U_3= I_{p-1}$. 
\end{enumerate} 
Furthermore, define 
$$\tilde Z = U_1^\top Y\in \R^{n-p}, \quad  \tilde W = U_2^\top Y\in \R, \quad \tilde V = U_{3}^\top Y\in \R^{p-1}, \quad U = \begin{bmatrix}U_1& U_2& U_3\end{bmatrix}.$$ 
Then note that
\begin{align*}
U^\top U &= \begin{bmatrix} U_1^\top \\ U_2^\top \\ U_3^\top \end{bmatrix} \begin{bmatrix}U_1& U_2& U_3\end{bmatrix}  = \begin{bmatrix}U_1^\top U_1& 0 & 0 \\ 0 & U_2^\top U_2& 0 \\ 0 & 0 & U_3^\top U_3 \end{bmatrix}= I_{n},\\
UU^\top &= \begin{bmatrix}U_1& U_2& U_3\end{bmatrix} \begin{bmatrix} U_1^\top \\ U_2^\top \\ U_3^\top \end{bmatrix}= U_1U_1^\top + U_2U_2^\top + U_3U_3^\top =I_n,\\
Y & =  U U^\top Y = U_1U_1^\top Y + U_2U_2^\top Y + U_3U_3^\top Y = U_1 \tilde Z + U_2 \tilde W + U_3 \tilde V,\\
Y&\in R(\tilde y;X)  \iff \tilde W^2\geq e_{\tilde y} \text{ and } \tilde V = U_3^\top \tilde y.
\end{align*}
Similarly to \eqref{eq:cond-step-2}, this allows us to write the conditional density of $Y$ given $Y\in R(\tilde y;X)$ as 
$$ f_{Y \mid Y \in R(\tilde y;X)}(y)  = \frac{f_{\tilde Z}(U_1^\top y) f_{\tilde W} (U_2^\top y) \bm{I}_{\{y\;:\;U_3^\top y = U_3^\top \tilde y\}}(y)\bm{I}_{\{y\;:\;(U_2^\top  y)^2 \geq e_{\tilde y}\}}(y)}{\Pr\left( \tilde W^2\geq e_{\tilde y}\right)}.$$
We now switch to viewing this expression as a likelihood of $\beta$ for a fixed realization $y\in \R^n,$ and thus the indicator $\bm{I}_{\{y\;:\;U_3^\top y = U_3^\top \tilde y\}}(y)$ must be satisfied. Note that the density $f_{\tilde Z}$ does not involve $\beta$. Then, 
$$ f_{Y \mid Y \in R(\tilde y;X)}(
\beta;y)  \propto \frac{ f_{\tilde W} (\beta;U_2^\top y) \bm{I}_{\{y\;:\;(U_2^\top  y )^2\geq e_{\tilde y}\}}(y)}{\Pr\left( \tilde W^2\geq e_{\tilde y}\right)}= f_{\tilde W \mid \tilde W^2 \geq e_{\tilde y}}(\beta;U_2^\top y). $$
Note that $\tilde W = U_2^\top Y\sim N(U_2^\top X\beta , \sigma^2)$, or equivalently $\tilde W\sim N_1(U_2^\top X_1\beta_1 , \sigma^2)$.
Thus, letting $\tilde w = U_2^\top y$,
$$f_{\tilde W|\tilde W^2\geq e_{\tilde y}}(\beta;\tilde w)  =\frac{\frac{1}{\sigma}\phi \left(\frac{\tilde w-U_2^\top X_1\beta_1}{\sigma}\right)\bm{I}_{\{\tilde w^2\geq e_{\tilde y}\}}(\tilde w)}{ 1 +\bm{I}_{\{e_{\tilde y}>0\} }\left(- \Phi\left(\frac{\sqrt{e_{\tilde y}}-U_2^\top X_1\beta_1}{\sigma}\right) + \Phi\left(\frac{-\sqrt{e_{\tilde y}}-U_2^\top X_1\beta_1}{\sigma}\right)\right) },$$
where $\phi$ is the density and $\Phi$ the CDF of a standard normal random variable and $\bm{I}_{\{e_{\tilde y}>0\} }$ is the indicator that $\tilde y$ is such that $e_{\tilde y}>0$. Noting that the above display is free of $\beta_{-1}$, we rewrite this in terms of $Y$ as 
$$f_{Y \mid Y \in R(\tilde y;X)}(\beta_1;y)  \propto \frac{\frac{1}{\sigma}\phi \left(\frac{U_2^\top y-U_2^\top X_1\beta_1}{\sigma}\right)\bm{I}_{\{(U_2^\top y)^2\geq e_{\tilde y}\}}(y)}{ 1 +\bm{I}_{\{e_{\tilde y}>0\} }\left(- \Phi\left(\frac{\sqrt{e_{\tilde y}}-U_2^\top X_1\beta_1}{\sigma}\right) + \Phi\left(\frac{-\sqrt{e_{\tilde y}}-U_2^\top X_1\beta_1}{\sigma}\right)\right)}.$$
Accordingly, we simply consider the second derivative of the log likelihood with respect to $\beta_1$ instead of the entire Hessian (with respect to $\beta)$. 

After some calculus, we find that if %
\begin{equation}\label{eq:S(Y)=1}
\bm{I}_{\{y\;:\;U_3^\top y = U_3^\top \tilde y\}}(y)=1\; \text{ and }\; \bm{I}_{\{y\;:\;(U_2^\top  y )^2\geq e_{\tilde y}\}}(y)=1,    
\end{equation}
then the second derivative has the expression
\begin{align*}
    \frac{d^2}{(d\beta_1)^2}&\log f_{Y \mid Y \in R(\tilde y;X)}(\beta_1;y)\\
    &= \frac{X_1^\top (P_X-P_{X_{-1}}) X_1}{\sigma^2} \left[ -1 + \bm{I}_{\{e_{\tilde y}>0\}}\left( \frac{[\phi(a) -\phi(b)]^2}{[1- \Phi(b)+\Phi(a)]^2}+\frac{\phi(a)\cdot a - \phi(b) \cdot b}{1-\Phi(b)+\Phi(a)}\right)\right]
\end{align*}
where 
$$a=\frac{-\sqrt{e_{\tilde y}}-U_2^\top X_1\beta_1}{\sigma}, \quad b=\frac{\sqrt{e_{\tilde y}}-U_2^\top X_1\beta_1}{\sigma}.$$ This second derivative is non-random. Additionally, the indicators in \eqref{eq:S(Y)=1} are equivalent to $Y \in \R(\tilde y ;X)$. Thus the leftover Fisher information about $\beta_1$ in $Y$ conditional on $Y \in \R(\tilde y ;X)$ is 
\begin{equation*}
\mathcal{I}_{Y\mid   Y \in R(\tilde y; X)}(\beta_1;Y \in R(\tilde y; X)) =\tfrac{X_1^\top (I-P_{X_{-1}})X_1 }{\sigma^2} \left[ 1 - \bm{I}_{\{e_{\tilde y}>0\}}\left(\tfrac{[\phi(a) -\phi(b)]^2}{[1-\Phi(b)+\Phi(a)]^2}+\tfrac{\phi(a)\cdot a - \phi(b) \cdot b}{1-\Phi(b)+\Phi(a)}\right)\right].
\end{equation*} 
\end{proof}

\subsection{Proof of Proposition~\ref{prop:retropsel}}

\begin{proof}[Proof of \cref{prop:retropsel}]

When we plug in the definitions of $\RSE$, $\RR^2$, and $\FhoM$, we find that $\aretro = a(y) $ and $\dretro = d(y)$ for $a(y)$ and $d(y)$ defined in \eqref{eq:acdr} and $\aretro$ and $\dretro$ defined in Proposition~\ref{prop:retropsel}. That is, 
$$\aretro = \tfrac{\RSE^2\cdot (n-p-1)}{1-\RR^2} = \tfrac{\frac{y^\top (I_{n-1}- P_X) y}{n-p-1}\cdot (n-p-1)}{\frac{ y^\top (I_{n-1}-P_X)y }{ y^\top y}} = y^\top y = a(y) ,$$
and 
\begin{align*}
    \dretro &= \RSE^2 \cdot \left( \tfrac{\RR^2\cdot (n-p-1)}{1-\RR^2}-\FhoM \cdot m \right)\\
    &= \frac{y^\top (I_{n-1}- P_X) y}{n-p-1} \cdot \left(\tfrac{\left(1 - \frac{ y^\top (I_{n-1}-P_X)y }{ y^\top y}\right)\cdot (n-p-1)}{ \frac{ y^\top (I_{n-1}-P_X)y }{ y^\top y}} -\tfrac{n-p-1}{m}\cdot \tfrac{y^\top (P_X - P_{X_{-M}}) y }{y^\top (I_{n-1}- P_X) y} \cdot m \right)\\
    &=y^\top y - y^\top (I_{n-1}-P_X)y - y^\top (P_X - P_{X_{-M}}) y \\
    & = y^\top P_{X_{-M}}y\\
    &  = d(y).
\end{align*}
Hence, beginning from the characterization of $\psel$ in \eqref{eq:pselchisq},
\begin{align*}
\psel(y) &= \Pr \left ( 
        \frac{W}{Z} \geq r(y)\,\middle|\frac{a(y)\cdot W + d(y)\cdot  Z}{(a(y)-d(y))\cdot Z}\geq c\right) \\
        &= \Pr \left (  \frac{W}{Z} \geq \FhoM \cdot \frac{n-p-1}{m}\,\,\middle|\frac{\aretro\cdot W + \dretro\cdot  Z}{(\aretro-\dretro)\cdot Z}\geq c\right),\end{align*}
        as claimed.

\end{proof}

\bibliographystyle{agsm}

\bibliography{refs-zotero}

\end{document}